\UseRawInputEncoding
\documentclass[11pt,reqno]{amsart}
\allowdisplaybreaks[4]
\usepackage{graphicx} 
\usepackage{subfigure}
\usepackage{epsfig}
\usepackage{amssymb}
\usepackage{amsmath,xypic}
\usepackage{cite,color}
\newtheorem{theorem}{Theorem}

\newtheorem{proposition}[theorem]{Proposition}
\newtheorem{corollary}[theorem]{Corollary}

\newtheorem{lemma}[theorem]{Lemma}

\newcommand{\be}{\begin{equation}}
\newcommand{\ee}{\end{equation}}
\newcommand{\bea}{\begin{eqnarray}}
\newcommand{\eea}{\end{eqnarray}}
\newcommand{\ba}{\begin{array}}
	\newcommand{\ea}{\end{array}}
\newcommand{\bean}{\begin{eqnarray*}}
	\newcommand{\eean}{\end{eqnarray*}}

\newcommand{\pa}{\partial}

\begin{document}

\title{Bilinear equations in Darboux transformations by Boson-Fermion correspondence}
\author{Yi Yang, Jipeng Cheng$^*$}
\dedicatory { School of Mathematics, China University of
Mining and Technology, \\Xuzhou, Jiangsu 221116, P.\ R.\ China}
\thanks{*Corresponding author. Email: chengjp@cumt.edu.cn.}
\begin{abstract}
Bilinear equation is an important property for integrable nonlinear evolution equation. Many famous research objects in mathematical physics, such as Gromov-Witten invariants, can be described in terms of bilinear equations to show their connections with the integrable systems. Here in this paper, we mainly discuss the bilinear equations of the transformed tau functions under the successive applications of the Darboux transformations for the KP hierarchy, the modified KP hierarchy (Kupershmidt-Kiso version) and the BKP hierarchy, by the method of the Boson-Fermion correspondence. The Darboux transformations are considered in the Fermionic picture, by multiplying the different Fermionic fields on the tau functions. Here the Fermionic fields are corresponding to the (adjoint) eigenfunctions, whose changes under the Darboux transformations are showed to be the ones of the squared eigenfunction potentials in the Bosonic picture, used in the spectral representations of the (adjoint) eigenfunctions. Then the successive applications of the Darboux transformations are given in the Fermionic picture. Based upon this, some new bilinear equations in the Darboux chain are derived, besides the ones of $(l-l')$ -th modified KP hierarchy. The corresponding examples of these new bilinear equations are given.

\end{abstract}
\maketitle
\tableofcontents
\section{Introduction}
Bilinear equations \cite{Babelon2003,Jimbo1983,Date1983,Hirota2004} are an important integrability for the nonlinear evolution equations, which are usually expressed in the Hirota bilinear forms of the tau functions by the Hirota bilinear operators\cite{Hirota2004} or bilinear residue identities of tau functions (or wave functions) \cite{Jimbo1983,Date1983,Miwa2000,Kac1998}. There exists a geometric interpretation for these kinds of bilinear equations. In fact by viewing tau functions as points in the infinite Grassmannian, the bilinear equations can be written as the Pl\"ucker relations \cite{Sato1982,Date1983,Miwa2000,Willox2004}. As for the algebraic aspects, the bilinear equations of the KP hierarchy (or the BKP hierarchy) are constructed as equivalent conditions of the orbit of $GL_\infty$ (or $O_\infty$) acting on the corresponding highest weight vectors\cite{Jimbo1983,Date1983,Kac2019,Kac1998,Wang2019,You1988}. Bilinear equations play an important role in mathematical physics. Besides in seeking various solutions of the nonlinear evolution equations (for example \cite{Feng2018, Ma2018}), the bilinear equations can also connect the famous Gromov-Witten invariants with the integrable systems, for example \cite{Milanov2007,Milanov2016,Carlet2013,Cheng2019}. Here in this article, we mainly discuss bilinear equations of the transformed tau functions under the Darboux transformations.

Darboux transformation is a kind of powerful method to construct solutions of the integrable system\cite{Matveev1991,Chau1992,Oevelpa1993,Oevelrmp1993,He2002}. In the KP hierarchy, there exist two types of elementary Darboux transformations\cite{Chau1992,Oevelpa1993,Oevelrmp1993,He2002}, that is, the differential type $T_d(\Phi)=\Phi\partial^{-1}\Phi^{-1}$ and the integral type $T_i(\Psi)=\Psi^{-1}\partial^{-1}\Psi$ (see Subsection \ref{subsectionkptdti} for more details). If denote $\tau^{[l]}(t)$ to be the transformed tau functions from $\tau(t)$ under $l$-step Darboux transformation $T_d$, also called the $T_d$ Darboux chain (see \eqref{tdphi}),  then it can be found that $\tau^{[l]}(t)$ satisfy \cite{Willox2004},
 \begin{align}
 &{\rm Res}_{\lambda}\lambda^{l-l'}\tau^{[l]}(t-\varepsilon(\lambda^{-1}))
\tau^{[l']}(t'+\varepsilon(\lambda^{-1}))e^{\xi(t-t',\lambda)}=0,
\quad l\geq l', \label{bilinnmkp}
 \end{align}
which is just the bilinear equation of the $(l-l')$-th modified KP hierarchy\cite{Jimbo1983,Kac2018}, or the one of the discrete KP hierarchy\cite{Adler1999,Haine2000,Dickey1999} with the discrete variables $l$ and $l'$ only taking the non-negative values. The 0-th modified KP hierarchy is the usual KP hierarchy\cite{Miwa2000,Date1983}, while the $1$st modified KP hierarchy can be used to describe the generating functions of open and closed intersection numbers \cite{Alexandrov2015}. Note that the discussion above is only restricted to the case of $T_d$. So it will be very natural to ask what the bilinear equations are in the case of $T_i$ or in the more generalized case of mixed using $T_d$ and $T_i$ in the KP hierarchy, and whether there are similar results for other integrable hierarchies. In what follows, we will discuss these questions for the KP, modified KP (the Kupershshmidt-Kiso version)\footnote{There are many versions of the modified KP hierarchy\cite{Chengjgp2018,Chengjnmp2018}. In this paper, the modified KP hierarchy always means the Kupershmidt-Kiso version \cite{Kupershmidt1985,Kiso1990,Chengjgp2018, Oevelrmp1993,Takebe2006} except for special illustrations.} and BKP hierarchies, since these three integrable hierarchies are closely connected with each other. Note that there are the Miura links between the KP and modified KP hierarchies\cite{Oevelrmp1993,Shaw1997}, and the BKP hierarchy can be viewed as their sub-hierarchies\cite{Date1983,Yang2020}. Further these relations lead to the close relations for their Darboux transformations\cite{Shaw1997,He2006,He2007,Yang2020}.

It is usually very difficult to obtain the bilinear equations of the tau functions in the chain of Darboux transformations by the Bosonic approach.
In fact, the efficient way to construct the bilinear equations is the Boson-Fermion correspondence\cite{Jimbo1983,Kac2019,Kac1998,Kac2018,Alexandrov2013,Date1983,Miwa2000}, i.e., linking the Fermionic Fock space and the Bosonic Fock space (see Section \ref{secprekpmkpbkp} for more details). To establish the bilinear equations in the Darboux transformations, the primary step is to express the Darboux transformation in the Fermionic picture, which is always presented in the changes of the Fermionic tau function \cite{Chau1992,Willox1998}. Under Darboux transformations, the changes of the Bosonic tau functions are closely related with the (adjoint) eigenfunctions \cite{Chau1992,Oevelpa1993,Oevelrmp1993,Chengjnmp2018,He2002,He2007,Yang2020}. So the key is to find suitable Fermionic fields to express the (adjoint) eigenfunctions. Fortunately the (adjoint) wave functions are corresponding to the generating functions of Fermions \cite{Jimbo1983,Miwa2000} (i.e., $\psi(\lambda)$ and $\psi^*(\lambda)$ in (\ref{psilambda})), while any (adjoint) eigenfunctions are related to the (adjoint) wave functions through the spectral representations (see Section \ref{secprekpmkpbkp} for more details) \cite{Aratyn1998,Chengjgp2018,Cheng2010,Willox1998,Loris1999} with the spectral density given by squared eigenfunction potential (SEP). Here the SEPs \cite{Oevelpa1993,Oevel1998} are usually defined as the integrals of the products of the eigenfunction and the adjoint eigenfunction with respect to $x$ in the KP case (other cases can be found in Section \ref{secprekpmkpbkp}), whose predecessor is so called Cauchy-Baker-Akhiezer kernel\cite{Grinevich1989}.

By using these facts, any (adjoint) eigenfunction is showed to be corresponding to one Fermionic field, which further leads to the correspondence of SEP and this Fermionic field.  Therefore under the Darboux transformations, the changes of the Fermionic fields corresponding to the (adjoint) eigenfunction are in fact the ones of the corresponding SEPs. Based upon these, we establish the Darboux transformations for the modified KP and BKP hierarchies in the Fermionic picture. Further we have obtained the transformed Fermionic tau functions under successive applications of Darboux transformations, especially the mixed using $T_d$ and $T_i$ in the KP case, which is one of the most important results in this paper. With this, the corresponding bilinear equations are obtained by using some important relations of fermions in Subsection \ref{subsectionrelationfreeferm}. At last, some new examples of the bilinear relations are given. Particularly, the bilinear equations involving the binary Darboux transformations are very important in the proof of the Adler-Shiota-van Moerbeke (ASvM) formulas in the additional symmetries \cite{Adler1995,Dickey1995,Chengjgp2018,Tu2007,Li2015}, and in the derivation of the bilinear equations \cite{Cheng1994,Chenjnmp2019,Shen2011,Loris1997} of the symmetry constraints of the KP, modified KP and BKP hierarchies.

This paper is organized in the way below. Firstly in Section 2, some basics facts on the constructions of the KP, modified KP and BKP hierarchies by free Fermions are reviewed, and the spectral representations of these three integrable hierarchies are discussed by using the relations among them.  Also some important relations on the free Fermions are established. Next based upon these, the transformed Fermionic tau functions and the bilinear equations under the successive applications of Darboux transformations for the KP, modified KP and BKP hierarchies are investigated in Section 3-5 respectively. At last, some conclusions and discussions are presented in Section 6.

\section{Preliminaries on the KP, modified KP and BKP hierarchies}\label{secprekpmkpbkp}
In this section, we firstly review some basic facts on the free Fermions and Boson-Fermion correspondence, and construct the $(l-l')$-th modified KP hierarchy. Then starting from the bilinear equations, we reviewed the basic facts of the KP, modified KP and BKP hierarchies, including the dressing structures and Lax equations, spectral representations of the eigenfunctions. Different from before, the spectral representations for the modified KP and the BKP hierarchy are constructed from the ones of the KP hierarchy by considering the relations among them. At last, we derived some important relations about the free Fermions, which will be used in the constructions of the bilinear equations in Darboux transformations.
\subsection{Boson-Fermion correspondence and $(l-l')$-th modified KP hierarchy}\label{subsectionbosonfermi}
In this subsection, we will review the Fermionic approach in the construction of the integrable systems and some facts on the $(l-l')$-th modified KP hierarchy. One can refer to \cite{Jimbo1983,Miwa2000,Kac1998} for more details. Let $\mathcal{A}$ be the Clifford algebra generated by the free Fermions $\psi_j$ and $\psi^*_j$ $(j\in \mathbb{Z})$, satisfying the following relations
\begin{eqnarray}\label{cliffordrelation}
[\psi_i,\psi_j]_+=[\psi^*_i,\psi^*_j]_+=0, \quad [\psi_i,\psi^*_j]_+=\delta_{ij},
\end{eqnarray}
where $[A,B]_+=AB+BA$. If define the vacuum vector $|0\rangle$ and the dual vacuum vector $\langle0|$ as follows
\begin{eqnarray}
\psi_i|0\rangle=0\quad(i<0), && \psi^*_i|0\rangle=0\quad(i\geq0),\nonumber\\
\langle0|\psi_i=0\quad(i\geq0), && \langle0|\psi^*_i=0\quad(i<0),\label{vacuum}
\end{eqnarray}
then one can obtain the Fermionic Fock space $\mathcal{F}=\mathcal{A}|0\rangle$ and its dual space $\mathcal{F}^*=\langle 0|\mathcal{A}$. The pairing between $\mathcal{F}$ and $\mathcal{F}^*$ are given as follows,
\begin{align*}
\mathcal{F}^*\times\mathcal{F}&\longrightarrow \mathbb{C}\\
(\langle 0 |a_1, a_2|0\rangle)&\longmapsto \langle 0 |a_1 a_2|0\rangle,
\end{align*}
where $\langle 0 |a_1 a_2|0\rangle$ can be computed by $\langle 0|1|0\rangle=1$, the relations \eqref{cliffordrelation} \eqref{vacuum} and the Wick theorem.
The Fermionic Fock space $\mathcal{F}$ can be used as the representation space of infinite dimensional Lie algebra $gl_\infty$ and its corresponding group $GL_\infty$. Here the infinite Lie algebra $gl_\infty$ is defined by
\begin{eqnarray*}
gl_\infty=\{\sum_{i,j\in Z}a_{i,j}:\psi_i\psi^*_j:|\text{there exists an N such that $a_{ij}=0,|i-j|>N$}\}\oplus \mathbb{C}
\end{eqnarray*}
and the corresponding group $GL_\infty$ is given by
\begin{eqnarray*}
GL_\infty=\{e^{X_1}e^{X_2}\cdots e^{X_k}|X_i\in gl_\infty\}.
\end{eqnarray*}

The charges of $\mathcal{F}$ and $\mathcal{F}^*$ can be defined the way below,
\begin{align*}
&\text{charge of $\psi_{j}=1$, charge of $\psi^*_{j}=-1,$}\\
&\text{charge of $|0\rangle$=0, charge of $a|0\rangle$= that of $a$},\\
&\text{charge of $\langle0|$=0, charge of $\langle0|a$= -that of $a$},
\end{align*}
then one can decompose $\mathcal{F}$ and $\mathcal{F}^*$ according to different charges
\begin{eqnarray*}
\mathcal{F}=\bigoplus_{m\in Z}\mathcal{F}_m,\quad \mathcal{F^*}=\bigoplus_{m\in Z}\mathcal{F}^{*}_m.
\end{eqnarray*}
For $m>0$, define the following vectors with charge $m$,
\begin{eqnarray*}
&&|m\rangle=\psi_{m-1}\cdots\psi_{0}|0\rangle,\quad |-m\rangle=\psi^*_{-m}\cdots\psi^*_{-1}|0\rangle\\
&&\langle m|=\langle 0|\psi^*_{0}\cdots\psi^*_{m-1},\quad \langle-m|=\langle 0|\psi_{-1}\cdots\psi_{-m}.
\end{eqnarray*}
By the definitions,
\begin{eqnarray}
&&\psi_{n}|m\rangle=0 \text{ for $n<m$}, \quad \psi^*_{n}|m\rangle=0\text{ for $n\geq m$}, \nonumber\\
&&\langle m|\psi_{n}=0\text{ for $n\geq m$}, \quad \langle m|\psi^*_{n}=0\text{ for $n<m$}.
\end{eqnarray}
If define $S=\sum_{j}\psi_j\otimes\psi_j^*$, one can find
\begin{align*}
S(|l\rangle\otimes|l'\rangle)=0, \quad S^{l-l'}(|l'\rangle\otimes|l\rangle)=(-1)^{\frac{(l-l')(l-l'-1)}{2}}(l-l')!
|l\rangle\otimes|l'\rangle,\text{ for $l\geq l'$}.
\end{align*}
Further if set $f_l=g|l\rangle$ for $g\in GL_{\infty}$ and assume $l\geq l'$, then by the fact $S$ can commute with $g\otimes g$,
\begin{align}
S(f_l\otimes f_{l'})&=0 \text{ ,}\label{klmkp}\\
S^{l-l'}(f_{l'}\otimes f_l)&=(-1)^{\frac{(l-l')(l-l'-1)}{2}}(l-l')!f_l\otimes f_{l'}.\label{klcmkp}
\end{align}
Here \eqref{klmkp} is the $(l-l')$-th modified KP hierarchy \cite{Jimbo1983,Kac1998, Kac2018} in the Fermionic picture. \eqref{klmkp} is showed to be equivalent to \eqref{klcmkp} in \cite{Kac1998, Kac2018}. Particularly, the 0-th modified KP hierarchy is the usual KP hierarchy\cite{Miwa2000,Date1983} and \eqref{klcmkp} for $l=l'+1$ is the modified KP hierarchy in the Kupershmidt-Kiso version \cite{Kupershmidt1985,Kiso1990,Chengjgp2018, Oevelrmp1993,Takebe2006}, which is equivalent to the $1$st modified KP hierarchy. Next we will mainly discuss these two particular cases. In order to rewritten \eqref{klmkp} and \eqref{klcmkp} into the usual forms, i.e., the Bosonic forms, we will next review the Boson-Fermion correspondence.

Introduce the generating sums of free Fermions
\begin{eqnarray}
\psi(\lambda)=\sum_{i\in Z}\psi_i\lambda^i,&& \psi^*(\lambda)=\sum_{i\in Z}\psi^*_i\lambda^{-i}\label{psilambda}
\end{eqnarray}
and define $\sum_{n\in\mathbb{Z}} H_n\lambda^{-n}=:\psi(\lambda)\psi(\lambda)^*:$
with the normal order $:AB:=AB-\langle 0|AB|0\rangle$. Then it can be proved that $H_n$ satisfies Heisenberg algebraic relations
\begin{align}
[H_m,H_n]=m\delta_{m,-n}.
\end{align}
For the time variables $t=(t_1=x,t_2,t_3,\cdots)$,  define $H(t)=\sum_{n=1}^{\infty}t_nH_n$.
Then
\begin{eqnarray}
e^{H(t)}\psi(\lambda)e^{-H(t)}=e^{\xi(t,\lambda)}\psi(\lambda),&&
e^{H(t)}\psi^*(\lambda)e^{-H(t)}=e^{-\xi(t,\lambda)}\psi^*(\lambda),
\end{eqnarray}
where $\xi(t,\lambda)=\sum_{n=1}^{\infty}t_n\lambda^n$. If introduce the Bosonic Fock space $\mathcal{B}=\mathbb{C}[z,z^{-1},t_1,t_2,t_3,\cdots]$, then there exists an isomorphism $\sigma_t:\mathcal{F}\longrightarrow \mathcal{B}$ given by
\begin{eqnarray*}
\sigma_t(a|0\rangle)=\sum_{j}z^j\langle j|e^{H(t)}a|0\rangle,\quad\text{$a\in\mathcal{A}$}.
\end{eqnarray*}
Note that if charge of $a=l$, then the terms $j\neq l$ on the right hand side will be zero according to the Wick Theorem. By using the isomorphism $\sigma_t$ and the following formulas \cite{Jimbo1983,Miwa2000},
\begin{eqnarray}
\langle l|\psi(\lambda)e^{H(t)} =\lambda^{l-1}\langle l-1|e^{H(t-\varepsilon(\lambda^{-1}))},
&&\langle l|\psi^*(\lambda)e^{H(t)}=\lambda^{-l}\langle l+1|e^{H(t+\varepsilon(\lambda^{-1}))},\label{langlepsi}
\end{eqnarray}
with $\varepsilon(\lambda^{-1})=(\lambda^{-1},\lambda^{-2}/2,\lambda^{-3}/3,\cdots)$,
one can realize the Fermions $\psi_i$ and $\psi_i^*$ in the forms of the Bosonic operators \cite{Jimbo1983,Miwa2000}
\begin{align*}
&\sigma_t\cdot\psi(\lambda)\cdot\sigma_t^{-1}
=e^{\xi(t,\lambda)}e^{-\xi(\tilde{\partial},\lambda^{-1})}e^K\lambda^{H_0},\quad \sigma_t\cdot\psi^*(\lambda)\cdot\sigma_t^{-1}
=e^{-\xi(t,\lambda)}e^{\xi(\tilde{\partial},\lambda^{-1})}e^{-K}\lambda^{1-H_0},
\end{align*}
where $\tilde{\partial}=\Big(\frac{\partial}{\partial x},\frac{1}{2}\frac{\partial}{\partial t_2},\frac{1}{3}\frac{\partial}{\partial t_3},\cdots\Big)$, and the operators $\lambda^{H_0}$ and $e^K$ act on $\mathcal{B}$ by the formulas below
\begin{align*}
\big(\lambda^{H_0}f\big)(z,t)\triangleq f(\lambda z,t),\quad
\big(e^K f\big)(z,t)\triangleq zf(z,t).
\end{align*}
This realization is called the Boson-Fermion correspondence.

After the preparation above, now we can write the bilinear equations (\ref{klmkp}) and (\ref{klcmkp}) into the Bosonic forms. If denote $\tau_l(t)=\langle l|e^{H(t)}g|l\rangle$, then by Boson-Fermion correspondence, (\ref{klmkp}) and (\ref{klcmkp}) will become into
\begin{align}
&{\rm Res}_{\lambda}\lambda^{l-l'}\tau_l(t-\varepsilon(\lambda^{-1}))
\tau_{l'}(t'+\varepsilon(\lambda^{-1}))e^{\xi(t-t',\lambda)}=0,\quad \label{klmkpbilinear}\\
 &{\rm Res}_{\lambda_1}\cdots{\rm Res}_{\lambda_{l-l'}}
\tau_{l'}(t-\sum_{i=1}^{l-l'}\varepsilon(\lambda_i^{-1}))
\tau_l(t'+\sum_{i=1}^{l-l'}\varepsilon(\lambda_i^{-1}))
\nonumber\\
&\times\prod_{i=1}^{l-l'}
\lambda_i^{l'-l}\prod_{i<j}(\lambda_i-\lambda_j)^2 e^{\sum_{i=1}^{l-l'}\xi(t-t',\lambda_i)}=(-1)^{\frac{(l-l')(l-l'-1)}{2}}(l-l')!
\tau_l(t)\tau_{l'}(t'). \label{klcmkpbilinear}
\end{align}
Here ${\rm Res}_{\lambda}\sum_i a_i\lambda^i=a_{-1}$.
These two bilinear equations are equivalent. Next, we will only consider two typical ones, i.e., (\ref{klmkpbilinear}) for $l=l'=0$ and (\ref{klcmkpbilinear}) for $l=l'+1=1$, which are just the usual KP hierarchy and the modified KP hierarchy \cite{Miwa2000,Date1983} of the Kupershmidt-Kiso version \cite{Chengjgp2018, Oevelrmp1993}. Other cases can be discussed in the similar method.
\subsection{Basic facts on the KP hierarchy}\label{subseckp}
In this subsection, we will review the basic facts of the KP hierarchy including the dressing and Lax equations, the spectral representations of the (adjoint) eigenfunctions.

Firstly when $l=l'=0$, (\ref{klmkpbilinear}) is just
\begin{eqnarray}
{\rm Res}_{\lambda}\tau(t-\varepsilon(\lambda^{-1}))
\tau(t'+\varepsilon(\lambda^{-1}))e^{\xi(t-t',\lambda)}=0\label{kptaubilinear}
\end{eqnarray}
with $\tau(t)=\tau_0(t)$. By introducing the wave function $\psi(t,\lambda)$ and the adjoint wave function $\psi^*(t,\lambda)$ in the way below,
\begin{eqnarray}
&&\psi(t,\lambda)= \frac{\tau(t-\varepsilon(\lambda^{-1}))}{\tau(t)}e^{\xi(t,\lambda)}=\frac{\langle 1|e^{H(t)}\psi(\lambda)g|0\rangle}{\langle 0|e^{H(t)}g|0\rangle}, \label{kpwavefunction}\\
&&
\psi^*(t,\lambda)=\frac{\tau(t+\varepsilon(\lambda^{-1}))}{\tau(t)}e^{-\xi(t,\lambda)}
=\frac{1}{\lambda}\frac{\langle -1|e^{H(t)}\psi^*(\lambda)g|0\rangle}{\langle 0|e^{H(t)}g|0\rangle},\label{kpadwavefunction}
\end{eqnarray}
the bilinear equations (\ref{kptaubilinear}) becomes into
\begin{eqnarray}
{\rm Res}_{\lambda}\psi(t,\lambda)\psi^*(t',\lambda)=0.\label{kpwavebilinear}
\end{eqnarray}
Introduce the following pseudo-differential operators (denote $\pa=\pa_x$)
\begin{align*}
W=1+\sum_{j=1}^\infty w_j\pa^{-j},\quad \tilde{W}=1+\sum_{j=1}^\infty \tilde{w}_j\pa^{-j},
\end{align*}
satisfying $\psi(t,\lambda)=W(e^{\xi(t,\lambda)})$ and $\psi^*(t,\lambda)=\tilde{W}(e^{-\xi(t,\lambda)})$, that is $w_i=\frac{p_i(-\tilde{\partial})\tau(t)}{\tau(t)}$ and $\tilde{w}_i=\frac{p_i(\tilde{\partial})\tau(t)}{\tau(t)}$. Here $p_i(t)$ is Schur polynomial defined by $\exp(\xi(t,\lambda))=\sum_ip_i(t)\lambda^i$.  Then one can obtain \cite{Date1983}
\begin{align}
\tilde{W}=(W^{-1})^*,\quad W_{t_n}=-(W\pa^n W^{-1})_{<0}W,\label{wtneq}
\end{align}
where $*$ is the adjoint operation defined by $\Big(\sum_{i}a_i\partial^i\Big)^*=\sum_{i}(-1)^i\partial^i a_i$.
If define the Lax operator $L$ as follows
\begin{eqnarray*}
L=W\pa W^{-1}=\pa+u_1\pa^{-1}+u_2\pa^{-2}+\cdots,
\end{eqnarray*}
then we can obtain the Lax equation of the KP hierarchy\cite{Miwa2000,Date1983}
\begin{eqnarray}
L_{t_n}=[(L^n)_{\geq0},L].\label{KPlax}
\end{eqnarray}
\noindent{\bf Remark:}
By considering the $\partial^{-1}$-terms in the second equation of \eqref{wtneq}, one can obtain ${\rm Res}_\partial L^n=(\log\tau)_{xt_n}$ with ${\rm Res}_\partial\sum_ia_i\partial^i=a_{-1}$. So each $u_i$ in the Lax operator $L$ can be expressed in terms of $(\log\tau)_{xt_n}$. Thus for a fixed Lax operator $L$, the corresponding tau function $\tau$ can be determined up to a multiplication of $c\exp(\sum_ia_it_i)$. For convenience, we let $\tau\sim\tau'$ if two tau functions $\tau$ and $\tau'$ determine the same Lax operator of the KP hierarchy. Further, denote $\tau\approx\tau'$ when $\tau$ and $\tau'$ determine the same dressing operator $W$. It is obviously that $\tau\approx\tau'$ if and only if $\tau'=c\tau$.

In what follows, the spectral representation for the KP hierarchy will be needed in the discussion of the Darboux transformation. That is, for the eigenfunction $\Phi$ and the adjoint eigenfunction $\Psi$ of the KP hierarchy defined by $$\Phi_{t_n}=(L^n)_{\geq0}(\Phi),\ \Psi_{t_n}=-(L^n)^*_{\geq0}(\Psi),$$
 it is showed in \cite{Aratyn1998,Willox1998} that they can be expressed by the wave function $\psi(t,\lambda)$ and the adjoint wave function $\psi^*(t,\lambda)$ of the KP hierarchy respectively, that is,
\begin{align}
\Phi(t)={\rm Res}_{\lambda}\rho(\lambda)\psi(t,\lambda),\quad \Psi(t)={\rm Res}_{\lambda}\rho^*(\lambda)\psi^*(t,\lambda),\label{spectralkp}
\end{align}
where $\rho(\lambda)=-\Omega(\Phi(t'),\psi^*(t',\lambda))$ and $\rho^*(\lambda)=\Omega(\psi(t',\lambda),\Psi(t'))$ belong to $\mathbb{C}((\lambda^{-1}))$.
Here $\Omega(f,g)$ is the squared eigenfunction potential (SEP)\cite{Oevelpa1993,Oevel1998}, determined by
\begin{eqnarray*}
\Omega(f(t),g(t))_x=f(t)g(t),\quad
\Omega(f(t),g(t))_{t_n}={\rm Res}_{\pa}(\pa^{-1}g(t) (L^n)_{\geq0}f(t)\pa^{-1})
\end{eqnarray*}
for the eigenfunction $f(t)$ and the adjoint eigenfunction $g(t)$ of the KP hierarchy, up to a constant.

The expressions of SEP can be derived in the way below. Note that
\begin{align*}
\Omega(\Phi(t),\psi^*(t,\lambda))
=(-\Phi(t)+\mathcal{O}(\lambda^{-1}))\lambda^{-1}e^{-\xi(t,\lambda)},\quad
\Omega(\psi(t,\lambda),\Psi(t))=(\Psi(t)+\mathcal{O}(\lambda^{-1}))\lambda^{-1}e^{\xi(t,\lambda)}.
\end{align*}
So if letting $t=t'+\varepsilon(\lambda^{-1})$ for in first relation of \eqref{spectralkp} and $t'=t+\varepsilon(\lambda^{-1})$ for the second one, then the expressions of SEPs can be derived which are given in the proposition \cite{Aratyn1998,Willox1998} below.
\begin{proposition}\label{sepexpprop} Given the eigenfunction $\Phi(t)$ and the adjoint eigenfunction $\Psi(t)$ of the KP hierarchy,
\begin{align*}
\Omega(\Phi(t),\psi^*(t,\lambda))=&-\frac{1}{\lambda}\Phi(t+\varepsilon(\lambda^{-1}))
\psi^*(t,\lambda),\\
\Omega(\psi(t,\lambda),\Psi(t))=&\frac{1}{\lambda}\Psi(t-\varepsilon(\lambda^{-1}))
\psi(t,\lambda).
\end{align*}
Particularly,
\begin{align*}
\Omega(\psi(t,\mu),\psi^*(t,\lambda))=&-\frac{1}{\lambda}\psi(t+\varepsilon(\lambda^{-1}),\mu)
\psi^*(t,\lambda)+\delta(\lambda,\mu)\nonumber\\
=&\frac{1}{\mu}\psi(t,\mu)\psi^*(t-\varepsilon(\mu^{-1}),\lambda),
\end{align*}
where $\delta(\lambda,\mu)=\frac{1}{\mu}\sum_{n\in\mathbb{Z}}
\left(\frac{\mu}{\lambda}\right)^n=\frac{1}{\lambda}\frac{1}{1-\mu/\lambda}
+\frac{1}{\mu}\frac{1}{1-\lambda/\mu}$.
\end{proposition}

Therefore the spectral representation (\ref{spectralkp}) can be rewritten into
\begin{align}
\Phi(t)&={\rm Res}_{\lambda}\frac{1}{\lambda}\psi(t,\lambda)\psi^*(t',\lambda)\Phi(t'+\varepsilon(\lambda^{-1})),\nonumber\\
\Psi(t)&={\rm Res}_{\lambda}\frac{1}{\lambda}\psi^*(t,\lambda)\psi(t',\lambda)\Psi(t'-\varepsilon(\lambda^{-1})).
\end{align}
If denote $\tau^+(t)=\Phi(t)\tau(t)$ and $\tau^-(t)=\Psi(t)\tau(t)$, then one can find that the spectral representation (\ref{spectralkp}) of $\Phi(t)$ and $\Psi(t)$ are equivalent to the following bilinear relations involving $\tau^\pm(t)$ and $\tau(t)$, that is
\begin{align}
&{\rm Res}_{\lambda}\lambda^{-1}\tau(t-\varepsilon(\lambda^{-1}))
\tau^+(t'+\varepsilon(\lambda^{-1}))
e^{\xi(t-t',\lambda)}=\tau^+(t)\tau(t'),\nonumber\\
&{\rm Res}_{\lambda}\lambda^{-1}\tau^-(t-\varepsilon(\lambda^{-1}))
\tau(t'+\varepsilon(\lambda^{-1}))
e^{\xi(t-t',\lambda)}=\tau(t)\tau^-(t'),\label{tauhattildetaubilinear}
\end{align}
which are just the bilinear equations of the modified KP hierarchy of the Kupershmidt-Kiso version.
\subsection{The mKP hierarchy and the Miura links}
Just as showed in the last subsection, the bilinear equation of the modified KP hierarchy is given as follows (that is (\ref{klcmkpbilinear}) for $l=l'+1=1$),
\begin{eqnarray}
{\rm Res}_{\lambda}\lambda^{-1}\tau_0(t-\varepsilon(\lambda^{-1}))
\tau_1(t'+\varepsilon(\lambda^{-1}))
e^{\xi(t-t',\lambda)}=\tau_1(t)\tau_0(t').\label{mkptaubilinear}
\end{eqnarray}
Here $(\tau_0,\tau_1)$ is called the tau pair of the modified KP hierarchy.
Firstly, we can introduce the wave function $w(t,\lambda)$ and the adjoint wave function $w^*(t,\lambda)$ of the modified KP hierarchy\cite{Chengjgp2018}
\begin{eqnarray}
 &&w(t,\lambda)= \frac{\tau_0(t-\varepsilon(\lambda^{-1}))}{\tau_1(t)}e^{\xi(t,\lambda)}=\frac{\langle 1|e^{H(t)}\psi(\lambda)g|0\rangle}{\langle 1|e^{H(t)}g|1\rangle},\label{mkpwave}\\
 &&w^*(t,\lambda)=\frac{\tau_1(t+\varepsilon(\lambda^{-1}))}{\tau_0(t)}\lambda^{-1}e^{-\xi(t,\lambda)}
 =\frac{\langle 0|e^{H(t)}\psi^*(\lambda)g|1\rangle}{\langle 0|e^{H(t)}g|0\rangle}\lambda^{-1},\label{adwavefunction}
\end{eqnarray}
so that (\ref{mkptaubilinear}) can be written into
\begin{eqnarray}
{\rm Res}_{\lambda}w(t,\lambda)w^*(t',\lambda)=1.\label{mkpwavebilinear}
\end{eqnarray}
If introduce the dressing operator $Z=\sum_{i=0}^\infty z_i\pa^{-i}$ such that $w(t,\lambda)=Z\left(e^{\xi(t,\lambda)}\right)$, then by the similar method to the KP case, one can get
\begin{eqnarray} w^*(t,\lambda)=(Z^{-1}\pa^{-1})^*\left(e^{-\xi(t,\lambda)}\right),&& Z_{t_n}=-(Z\pa^n Z^{-1})_{<1} Z.\label{evoeqofZ}
\end{eqnarray}
Then the Lax operator $\mathcal{L}$ of the modified KP hierarchy can be introduced as $\mathcal{L}=Z\pa Z^{-1}=\pa+v_0+v_1\pa^{-1}+v_2\pa^{-2}+\cdots$, satisfying the Lax equation \cite{Chengjgp2018, Oevelrmp1993} below
\begin{eqnarray}
\mathcal{L}_{t_n}=[(\mathcal{L}^n)_{\geq1},\mathcal{L}]\label{mKPkklax}
\end{eqnarray}
\noindent{\bf Remark:}
Note that \cite{Chenijmpa2019} ${\rm Res}_\partial\mathcal{L}^n=(\log\tau_0)_{xt_n}$ and ${\rm Res}_\partial(\partial \mathcal{L}^n\partial^{-1})^*=-(\log\tau_1)_{xt_n}$, which can be obtained by comparing the $\partial^0$ and $\partial^{-1}$-terms in the evolution equation (\ref{evoeqofZ}) of the dressing operator $Z$. Thus each $v_i$ in $\mathcal{L}$ can be expressed by $(\log\tau_0)_{xt_n}$ and $v_0=(\log(\tau_1/\tau_0))_x$, or  $(\log\tau_1)_{xt_n}$ and $v_0$. Therefore the tau pair
\begin{align*}
(\tau_0',\tau_1')=e^{ax}\Big(c_0\exp(\sum_{i>1} a_it_i)\tau_0,c_1\exp(\sum_{i>1} b_it_i)\tau_1\Big)
\end{align*}
has the same Lax operator of the modified KP hierarchy as the tau pair $(\tau_0,\tau_1)$. In this case, we denote $(\tau_0',\tau_1')\sim (\tau_0,\tau_1)$. Further when $(\tau_0',\tau_1')=c(\tau_0,\tau_1)$, they share the same dressing structure, denoted by $(\tau_0',\tau_1')\approx (\tau_0,\tau_1)$.

Before further discussion, the lemma \cite{Chengjgp2018} below is needed, which can be proved by applying $\pa_{x'}$ on the both sides of (\ref{mkptaubilinear}) and set $t-t'=\varepsilon(z^{-1})$.
\begin{lemma}\label{mkpfaylemma} If $\tau_0$ and $\tau_1$ are tau functions of the modified KP hierarchy satisfying (\ref{mkptaubilinear}), then
\begin{align}
\pa_x\left(\frac{\tau_1(t+\varepsilon(\lambda^{-1}))}{\lambda\tau_0(t)}\right)
=\frac{\tau_1(t+\varepsilon(\lambda^{-1}))}{\tau_0(t)}
-\frac{\tau_0(t+\varepsilon(\lambda^{-1}))\tau_1(t)}{\tau_0(t)^2},\label{fayequation}
\end{align}
and therefore the wave function $w(t,\lambda)$ and the adjoint wave function $w^*(t,\lambda)$ of the modified KP hierarchy satisfy
\begin{align}
w(t,\lambda)_x=\lambda \frac{\tau_0(t)\tau_1(t-\varepsilon(\lambda^{-1}))}{\tau_1(t)^2}e^{\xi(t,\lambda)},\quad w^*(t,\lambda)_x=- \frac{\tau_0(t+\varepsilon(\lambda^{-1}))\tau_1(t)}{\tau_0(t)^2}e^{-\xi(t,\lambda)}.
\label{wxwstarxmkp}
\end{align}
\end{lemma}

By this lemma, we can obtain the proposition \cite{Chengjgp2018} below.
\begin{proposition}\label{mkp2tau}
If $\tau_0$ and $\tau_1$ are tau functions of the modified KP hierarchy, i.e.,
\begin{eqnarray*}
{\rm Res}_{\lambda}\lambda^{-1}\tau_0(t-\varepsilon(\lambda^{-1}))
\tau_1(t'+\varepsilon(\lambda^{-1}))
e^{\xi(t-t',\lambda)}=\tau_1(t)\tau_0(t'),
\end{eqnarray*}
then $\tau_0$ and $\tau_1$ are the tau functions of the KP hierarchy,
\begin{align*}
{\rm Res}_{\lambda}\tau_i(t-\varepsilon(\lambda^{-1}))
\tau_i(t'+\varepsilon(\lambda^{-1}))
e^{\xi(t-t',\lambda)}=0,\quad i=1,2.
\end{align*}
\end{proposition}

The eigenfunction $\hat\Phi(t)$ and the adjoint eigenfunction $\hat\Psi(t)$ of the modified KP hierarchy are defined by
\begin{align*}
  \hat\Phi_{t_n}=(\mathcal{L}^{n})_{\geq 1}(\hat\Phi),\quad \hat\Psi_{t_n}=-\Big(\pa^{-1}(\mathcal{L}^{n})^*_{\geq 1}\pa\Big)(\hat\Psi).
\end{align*}
The spectral representation of the (adjoint) eigenfunction of the modified KP hierarchy can be derived from the ones in the KP case, by using the Miura links between the modified KP hierarchy (\ref{mKPkklax}) and the KP hierarchy (\ref{KPlax}), which are given as follows.

$\bullet$ Anti-Miura transformations (KP$\longrightarrow$modified KP)
\begin{equation*}
L \rightarrow \mathcal{L}=
\begin{cases}
T_mL T_m^{-1},& T_m(\Phi)=\Phi^{-1},\\
T_nL T_n^{-1},& T_n(\Psi)=\pa^{-1}\Psi,
\end{cases}
\end{equation*}
where $\Phi$ and $\Psi$ are the eigenfunction and the adjoint eigenfunction of the KP hierarchy respectively.

$\bullet$ Miura transformation (modified KP$\longrightarrow$KP)
\begin{equation*}
\mathcal{L} \rightarrow L=
\begin{cases}
T_\mu\mathcal{L} T_\mu^{-1},& T_\mu=z_0^{-1},\\
T_\nu\mathcal{L} T_\nu^{-1},& T_\nu=z_0^{-1}\pa,
\end{cases}
\end{equation*}
where $z_0$ to be the coefficient of $\pa^0$-term in the dress operator $Z$ of the modified KP hierarchy.

The corresponding changes of the dressing operators, the eigenfunctions and the adjoint eigenfunctions, the wave and the adjoint wave functions, and tau functions under the Miura links are given in the two tables below\cite{Shaw1997}.
\begin{center}
\begin{tabular}{llllllll}
\multicolumn{8}{c}{Table I. Anti-Miura transformation: KP $\rightarrow$ modified KP}\\
\hline\hline
\ \ \ $L\rightarrow \mathcal{L} $ & \ \ \ \ $Z=$ & \ \ \ \ $\hat\Phi_1=$& \ \ \ $\hat\Psi_1=$ & $w=$ & $w^*=$ & $\tau_0=$& $\tau_1=$\\
\hline
 \ $T_m(\Phi)=\Phi^{-1}$ & \ \ \ \ $\Phi^{-1}W $ &\ \ \ $\Phi^{-1}\Phi_1$ & \ $-\int\Phi\Psi_{1x}dx$& $\Phi^{-1}\psi$ &$-\int\Phi\psi^*_xdx$ & $\tau$ & $\Phi\tau$\\
 \ $T_n(\Psi)=\pa^{-1}\Psi$ & \ \ \ $\pa^{-1}\Psi W \pa $&\ \ \ $\int\Psi\Phi_1 dx$ & \ $\Psi^{-1}\Psi_{1x}$&$\lambda\int\Psi\psi dx$ &$\Psi^{-1}\psi^*_x/\lambda$ & $\Psi\tau$ & $\tau$\\
\hline\hline
\end{tabular}
\end{center}
\begin{center}
\begin{tabular}{lllllll}
\multicolumn{6}{c}{Table II. Miura transformation: modified KP $\rightarrow$ KP}\\
\hline\hline
\ \ \ $\mathcal{L}\rightarrow L$ & \ \ \ \ $W=$ & \ \ \ \ $\Phi=$\ \ \ & \ \ $\Psi=$&$\psi=$ & $\psi^*=$ & $\tau=$\\
\hline
 \  $T_\mu=z_0^{-1}$ & \ \ \ \ $z_0^{-1}Z $ &\ \ \ $z_0^{-1}\hat \Phi$ & \ $-z_0\hat\Psi_{x}$& $z_0^{-1}w$ &$-z_0w^*_x$ &$\tau_0$ \\
 \ $T_\nu=z_0^{-1}\pa$ & \ \ \ $z_0^{-1}\pa Z\pa^{-1}$ &\ $z_0^{-1}\hat\Phi_x$ & \ $z_0\hat\Psi$& $z_0^{-1}w_x/\lambda$ & $\lambda z_0w^*$& $\tau_1$ \\
\hline\hline
\end{tabular}
\end{center}
Here $\Phi_1\neq c\Phi$ and $\Psi_1\neq c\Psi$ are the eigenfunction and the adjoint eigenfunction of the KP hierarchy with respect to the Lax operator $L$, while $\hat \Phi_1$ and $\hat \Psi_1$ mean the eigenfunction and the adjoint eigenfunction of the modified KP hierarchy with respect to $\mathcal{L}$. $w=w(t,\lambda)$ and $w^*=w(t,\lambda)$ are the wave and the adjoint wave functions of the modified KP hierarchy respectively, while $\psi=\psi(t,\lambda)$ and $\psi^*=\psi^*(t,\lambda)$ are the ones of the KP hierarchy. The results for tau functions can be obtained by comparing the coefficients of $\partial^0$ and $\partial^{-1}$ in the dressing operators $Z$ and $W$, which are not considered in \cite{Shaw1997}.

\noindent{\bf Remark:} In the anti-Miura transformations, $(\tau_0,\tau_1)=(\tau,\Phi\tau)$ or $(\tau_0,\tau_1)=(\Psi\tau,\tau)$ still satisfies the bilinear equation (\ref{mkptaubilinear}) of the modified KP hierarchy, according to the results in (\ref{tauhattildetaubilinear}), while $\tau=\tau_0$ or $\tau=\tau_1$ satisfies the bilinear equation of the KP hierarchy (\ref{kptaubilinear}) by Proposition \ref{mkp2tau}.

After the preparation above, now we can discuss the spectral representations of the modified KP hierarchy, by using the Miura links from the results for the KP hierarchy. For this, define the squared eigenfunction potential $\Omega(\hat\Phi,\hat\Psi_x)$ and $\hat \Omega(\hat\Phi_x,\hat\Psi)$ for the eigenfunction $\hat\Phi$ and the adjoint eigenfunction $\hat\Psi$ of the modified KP hierarchy in the way below\cite{Oevel1998,Chengjgp2018}.
\begin{align*}
\Omega(\hat\Phi,\hat\Psi_x)_x=&\hat\Phi\hat\Psi_x,\quad \Omega(\hat\Phi,\hat\Psi_x)_{t_n}={\rm Res}_\pa(\pa^{-1}\hat\Psi_x(L^n)_{\geq 1}\hat\Phi\pa^{-1}),\\
\hat \Omega(\hat\Phi_x,\hat\Psi)_x=&\hat\Phi_x\hat\Psi,\quad \hat \Omega(\hat\Phi_x,\hat\Psi)_{t_n}={\rm Res}_\pa(\pa^{-1}\hat\Psi\pa(L^n)_{\geq 1}\pa^{-1}\hat\Phi_x\pa^{-1}),
\end{align*}
The relation of $\Omega(\hat\Phi,\hat\Psi_x)$ and $\hat \Omega(\hat\Phi_x,\hat\Psi)$ is as follows.
\begin{eqnarray}
\hat \Omega(\hat\Phi_x,\hat\Psi)=-\Omega(\hat\Phi,\hat\Psi_x)+\hat\Phi\hat\Psi.\label{hatsep}
\end{eqnarray}
Then one can have the following spectral representations of the eigenfunction and the adjoint eigenfunction for the modified KP hierarchy in the proposition below.
\begin{proposition}\label{propmkpsp}
Let $w(t,\lambda)$ and $w^*(t,\lambda)$ be the wave function and the adjoint wave function for the modified KP hierarchy respectively, then for the eigenfunction $\hat\Phi(t)$ and the adjoint eigenfunction $\hat\Psi(t)$ of the modified KP hierarchy,
\begin{align}
&{\rm Res}_{\lambda}w(t',\lambda)\hat \Omega(\hat\Phi(t)_x,w^*(t,\lambda))
=\hat\Phi(t)-\hat\Phi(t'),\label{mkpphmphsp}\\
&{\rm Res}_{\lambda}w^*(t',\lambda) \Omega(w(t,\lambda),\hat\Psi(t)_x)=\hat\Psi(t)-\hat\Psi(t').\label{mkppsmpssp}
\end{align}
\end{proposition}
\begin{proof}
Firstly, if denote $\Phi(t)=z_0^{-1}\hat\Phi(t)$ and $\psi^*(t,\lambda)=-z_0w^*(t,\lambda)_x$, then $\Phi$ and $\psi^*(t,\lambda)$ are the eigenfunction and the adjoint wave function of the KP hierarchy, according to the results of the Miura transformations $T_\mu=z_0^{-1}$, showed in Table II. Then one can compute
$\hat \Omega(\hat\Phi(t)_x,w^*(t,\lambda))$ in the way below,
\begin{align}
\hat \Omega(\hat\Phi(t)_x,w^*(t,\lambda))=\int\hat\Phi(t)_xw^*(t,\lambda)dx
=\hat\Phi(t)w^*(t,\lambda)+\Omega(\Phi(t),\psi^*(t,\lambda)).\label{mkpsepkpsep}
\end{align}
Further by considering $z_0^{-1}w(t,\lambda)$ is the wave function of the KP hierarchy (see Table I), one can prove (\ref{mkpphmphsp}) according to (\ref{mkpwavebilinear}) and (\ref{spectralkp}). (\ref{mkppsmpssp}) can be proved by the similar method.
\end{proof}

By considering the relation of $\Omega(\Phi,\Psi_x)$ and $\hat \Omega(\Phi_x,\Psi)$ in (\ref{hatsep}), one can obtain the corollary \cite{Chengjgp2018} below.
\begin{corollary}\label{corspectralmkp}
Under the same condition in Proposition \ref{propmkpsp},
\begin{align}
\hat\Phi(t)={\rm Res}_{\lambda}w(t,\lambda)\hat\rho(\lambda),\quad
\hat\Psi(t)={\rm Res}_{\lambda}w^*(t,\lambda)\hat\rho^*(\lambda),\label{mkpphpssp}
\end{align}
with $\hat\rho(\lambda)=\Omega(\hat\Phi(t'),w^*(t',\lambda)_{x'})$ and $\hat\rho^*(\lambda)=\hat \Omega(w(t',\lambda)_{x'},\hat\Psi(t'))$.
\end{corollary}
According to (\ref{mkpsepkpsep}), Proposition \ref{sepexpprop} and Lemma \ref{mkpfaylemma}, one can obtain the expressions of SEPs in the modified KP hierarchy.
\begin{proposition}\label{sepexpressionmkp}
Given the eigenfunction $\hat \Phi(t)$ and $\hat \Psi(t)$, the expressions of SEPs for the modified KP hierarchy are listed below.
\begin{eqnarray*}
\Omega(\hat\Phi(t),w^*(t,\lambda)_x)&=&w^*(t,\lambda)\hat\Phi(t+\varepsilon(\lambda^{-1})),\\
\hat \Omega(\hat\Phi(t)_x,w^*(t,\lambda))&=&w^*(t,\lambda)
\left(\hat\Phi(t)-\hat\Phi(t+\varepsilon(\lambda^{-1}))\right),\\
\hat \Omega(w(t,\mu)_x,\hat\Psi(t))&=&w(t,\mu)\hat\Psi(t-\varepsilon(\mu^{-1})),\\
\Omega(w(t,\mu),\hat\Psi(t)_x)&=&w(t,\mu)\left(\hat\Psi(t)-\hat\Psi(t-\varepsilon(\mu^{-1}))\right).
\end{eqnarray*}
Particularly,
\begin{align*}
\Omega(w(t,\mu),w^*(t,\lambda)_x)=w^*(t,\lambda)w(t+\varepsilon(\lambda^{-1}),\mu),\\
\hat \Omega(w(t,\mu)_x,w^*(t,\lambda))=w(t,\mu)w^*(t-\varepsilon(\mu^{-1}),\lambda).
\end{align*}
\end{proposition}

\subsection{The neutral free Fermions and the BKP hierarchy}
The BKP hierarchy can be seen as the sub-hierarchy of the KP hierarchy or the modified KP hierarchy \cite{Date1983,Yang2020}, which can be expressed by the neutral free Fermions\cite{Date1983,Jimbo1983,Kac2019,Kac1998,vandeleur2015,You1988,Wang2019}
\begin{align*}
\phi_n=\frac{\psi_n+(-1)^n\psi_{-n}^*}{\sqrt{2}},\quad n\in \mathbb{Z},
\end{align*}
satisfying
\begin{equation}
  [\phi_m,\phi_n]_+=(-1)^m\delta_{m,-n}.\label{neufermirelation}
\end{equation}
Denote the $\mathcal{A}_B$ as the Clifford algebra generated by $\{\phi_n\}_{n\in \mathbb{Z}}$. Then
define the vacuum $|0\rangle$ and $\langle 0|$ in the way below
\begin{align}
&\phi_{-j}|0\rangle=0 \quad (j>0),\quad \langle 0|\phi_{j}=0 \quad(j>0),\nonumber\\
&\phi_{0}|0\rangle=\frac{1}{\sqrt{2}}|0\rangle,\quad \langle 0|\phi_{0}=\frac{1}{\sqrt{2}}\langle 0|.
 \label{neufermivacuum}
\end{align}
One can thus obtain the Fock space for neutral Fermions $\mathcal{F}_B=\mathcal{A}_B|0\rangle$ and its dual space $\mathcal{F}_B^*=\langle 0|\mathcal{A}_B$. Also there is a bilinear pairing $\mathcal{F}_B^*\times \mathcal{F}_B\rightarrow \mathbb{C}$, denoted by $\langle 0|ab|0\rangle$ and computed according to (\ref{neufermirelation}), (\ref{neufermivacuum}) and $\langle 0|0\rangle=1$.

If define $S_B=\sum_{j\in \mathbb{Z}}(-1)^j\phi_j\otimes\phi_{-j}$
 and
 $$g\in O_\infty=\Big\{\exp\Big(\sum_{m,n\in\mathbb{Z}}b_{mn}:\phi_m\phi_n:\Big)|\text{$\exists$ $N$, $b_{mn}=0$, if $|m+n|>N$}\Big\},$$
  then one can obtain the bilinear equations of the BKP hierarchy in the Fermionic version \cite{Kac1998,Kac2019,vandeleur2015}
\begin{align}
S_B(g|0\rangle\otimes g|0\rangle)=\frac{1}{2}g|0\rangle\otimes g|0\rangle.\label{bkpbilinearfermi}
\end{align}
By introducing the following two generating series:
\begin{align*}
\phi(z)=\sum_{j\in\mathbb{Z}}\phi_jz^{j},\quad \sum_{k\in 2\mathbb{Z}+1}H_{B,k}z^{-k-1}=:\frac{1}{2z}\phi(z)\phi(-z):
\end{align*}
and defining $H_B(t)=\sum_{k=0}^\infty t_{2k+1}H_{B,2k+1}$, then we can define
the isomorphism between $\mathcal{F}_B$ and $\mathbb{C}[t_1,t_3,t_5\cdots]$ as follows\cite{Kac2019,You1988,Wang2019},
\begin{align*}
\sigma_{B,t}:&\mathcal{F}_B\rightarrow\mathbb{C}[t_1,t_3,t_5\cdots]\nonumber\\
&a|0\rangle\mapsto\langle 0|e^{H_B(t)}a|0\rangle,\quad\text{$a\in\mathcal{A}_B$}.
\end{align*}
Here the normal order of neutral Fermions $:\cdot:$ is the same as before. Then the neutral Boson-Fermion correspondence can be introduced in the way below by using the relation \cite{Wang2019} $\sqrt{2}\langle 0|\phi(z)=\langle 0|e^{-H_B(2\widetilde{\varepsilon}(\lambda^{-1}))}$,
\begin{align*}
\sigma_{B,t}\cdot\phi(\lambda)\cdot\sigma_{B,t}^{-1}=\frac{1}{\sqrt{2}}
e^{\widetilde{\xi}(t,\lambda)}e^{-2\widetilde{\xi}(\hat{\partial},\lambda^{-1})}
\end{align*}
where $\widetilde{\xi}(t,\lambda)=\sum_{k=0}t_{2k+1}\lambda^{2k+1}$ and $\hat{\partial}=\Big(\frac{\partial}{\partial x},\frac{1}{3}\frac{\partial}{\partial t_3},\frac{1}{5}\frac{\partial}{\partial t_5},\cdots\Big)$. Now if denote $\tau_B(t)=\langle 0|e^{H_B(t)}g|0\rangle$, we can rewritten (\ref{bkpbilinearfermi}) into
\begin{align}
{\rm Res}_{\lambda}\lambda^{-1}\tau_B(t-2\widetilde{\varepsilon}(\lambda^{-1}))
\tau_B(t'+2\widetilde{\varepsilon}(\lambda^{-1}))
e^{\widetilde{\xi}(t-t',\lambda)}=\tau_B(t)\tau_B(t'),\label{bkpbilinearboson}
\end{align}
where $\widetilde{\varepsilon}(\lambda^{-1})=(\lambda^{-1},\lambda^{-3}/3,
\lambda^{-5}/5,\cdots)$.
By introducing the wave functions of the BKP hierarchy in the way below\cite{Date1983}£º
\begin{align}
\psi_B(t,\lambda)=\frac{\tau_{B}(t-2\widetilde
{\varepsilon}(\lambda^{-1}))}{\tau_{B}(t)}e^{\widetilde{\xi}(t,\lambda)}
=\frac{\sqrt{2}\langle 0|e^{H_B(t)}\phi(\lambda)g|0\rangle}{\langle 0|e^{H_B(t)}g|0\rangle}\label{bkpwavefunction},
\end{align}
then (\ref{bkpbilinearboson}) will become into
\begin{eqnarray}
{\rm Res}_{\lambda}\lambda^{-1}\psi_B(t,\lambda)\psi_B(t',-\lambda)=1.\label{bkpwavebilinear}
\end{eqnarray}
So if let $W_B=1+\sum_{j=1}v_{j}\pa^{-j}$ satisfying $\psi_B(t,\lambda)=W_B(e^{\widetilde{\xi}(t,\lambda)})$, then it can be proved that
\begin{align}
W_B\pa^{-1}W_B^*=\pa^{-1},\quad (W_B)_{t_{2k+1}}=-(W_B\pa^{2k+1}W_B^{-1})_{<0}W_B.
\end{align}
Further define $L_B=W_B\pa W_B^{-1}$, one can obtain the Lax equation of the BKP hierarchy
\begin{align}
L_B^{*}=-\pa L_B \pa^{-1},\quad (L_B)_{t_{2k+1}}=[(L_B^{2k+1})_{\geq 0},L_B].
\end{align}
\noindent{\bf Remark:}
Given a fixed Lax operator $L_B$ of the BKP hierarchy, the tau function $\tau_B$ is determined up to a multiplication of $c\exp(\sum_{i\ {\rm odd}}a_it_i)$. $\tau'_B\sim\tau_B$ and $\tau'_B\approx\tau_B$ have similar meanings to the KP and modified KP cases.

In \cite{Yang2020}, the BKP hierarchy is viewed as the Kupershmidt
reduction of the modified KP hierarchy, i.e., $\mathcal{L}^*=\partial\mathcal{L}\partial^{-1}$. As a sub-hierarchy of the modified KP hierarchy, the wave function $\psi_B(t,\lambda)$ of the BKP hierarchy is related with the adjoint wave function $w^*(t,\lambda)$ of the mKP hierarchy in the way below\cite{Yang2020},
\begin{align}
w^*(t,\lambda)=\frac{1}{\lambda}\psi_B(t,-\lambda).\label{wspsib}
\end{align}
Therefore the bilinear equation \eqref{mkpwavebilinear} of the modified KP hierarchy can be naturally reduced into the one of the BKP hierarchy given in \eqref{bkpwavebilinear}.
What's more, any eigenfunction $\Phi_B(t)$ of the BKP hierarchy (satisfying $\Phi_{B,t_{2k+1}}=(L_B^{2k+1})_{\geq 1}(\Phi_B)$) can also be viewed as the adjoint eigenfunction of the BKP hierarchy seen as the Kupershmidt
reduction of the modified KP hierarchy\cite{Yang2020}, that is,
\begin{align*}
\Phi_B(t)_{t_{2k+1}}=-(\partial (L_B^{2k+1})_{\geq 1}\partial^{-1})^*(\Phi_B(t)).
\end{align*}
By using these facts, the corresponding spectral representation for the BKP hierarchy can be derived naturally from the one of the modified KP hierarchy (see Proposition \ref{propmkpsp}), which is given in the proposition below\cite{Loris1999,Cheng2010}.

\begin{proposition}\label{propspecbkp}
For the eigenfunction $\Phi_B(t)$ of the BKP hierarchy,
\begin{align}
{\rm Res}_{\lambda}\frac{1}{\lambda}\Omega(\Phi_B(t)_x,
\psi_{B}(t,-\lambda))\psi_{B}(t',\lambda)
=\Phi_B(t)-\Phi_B(t'),\label{bkpspphimphi}
\end{align}
Further,
\begin{align}
\Phi_B(t)={\rm Res}_{\lambda}\frac{1}{\lambda}\rho_B(\lambda)\psi_{B}(t,\lambda),\quad \rho_B(\lambda)=\Omega(\Phi_B(t'),\psi_{B}(t',-\lambda)_x).\label{bkpspphi}
\end{align}
\end{proposition}

If assume $\Omega(\psi_B(t,\lambda),\Phi_B(t)_x)=K(t,\lambda)e^{\xi(t,\lambda)}$ with $K(t,\lambda)\in \mathbb{C}((\lambda^{-1}))$, then by choosing $t-t'=2\widetilde{\varepsilon}(\lambda^{-1})$ in \eqref{bkpspphimphi} one can obtain the following expressions of SEPs for the BKP hierarchy\cite{Loris1999,Cheng2010}.
\begin{proposition}
Assume $\Phi_B(t)$ to be the eigenfunction of the BKP hierarchy, then the expression of SEPs for the BKP hierarchy are given below.
\begin{align*}
&\Omega(\psi_B(t,\lambda),\Phi_B(t)_x)=\frac{1}{2}\psi_B(t,\lambda)
\Big(\Phi_B(t)-\Phi_B(t-2\widetilde{\varepsilon}(\lambda^{-1}))\Big)\\
&\Omega(\Phi_B(t),\psi_B(t,\lambda)_x)=\frac{1}{2}\psi_B(t,\lambda)
\Big(\Phi_B(t)+\Phi_B(t-2\widetilde{\varepsilon}(\lambda^{-1}))\Big).
\end{align*}
In particular,
\begin{align*}
\Omega(&\psi_B(t,\mu),\psi_B(t,-\lambda)_x)=\frac{1}{2}\psi_B(t,\mu)
\Big(\psi_B(t,-\lambda)-\psi_B(t-2\widetilde{\varepsilon}(\mu^{-1}),-\lambda)\Big)\nonumber\\
&=\frac{1}{2}\psi_B(t,-\lambda)
\Big(\psi_B(t,\mu)+\psi_B(t+2\widetilde{\varepsilon}(\lambda^{-1}),\mu)\Big)
+\frac{1}{2}(\lambda+\mu)\delta(\lambda,\mu).
\end{align*}
\end{proposition}
Further if define the SEP of B type (BSEP)\cite{Loris1999,Cheng2010}
$\Omega_B(\Phi_{B1},\Phi_{B2})\triangleq
\Omega(\Phi_{B2},\Phi_{B1,x})-\Omega(\Phi_{B1},\Phi_{B2,x})$, then one can find that
\begin{align}
\Omega_B(\Phi_B(t),\psi_B(t,\lambda))=
-\psi_B(t,\lambda)\Phi_B(t-2\widetilde{\varepsilon}(\lambda^{-1})).\label{omegab}
\end{align}

\subsection{Important relations on free Fermions}\label{subsectionrelationfreeferm}

In this section, we will give  important relations about free Fermions used to construct the bilinear equations. Firstly denote
\begin{align*}
&H_{n,a}=\{(j_a,j_{a-1},\cdots,j_1)|n\geq j_a>\cdots>j_1\geq 1\},\\
&m+H_{n,a}=\{(m+j_a,m+j_{a-1},\cdots,m+j_1)|n\geq j_a>\cdots>j_1\geq 1\}.
\end{align*}
Note that $H_{n,n}=\vec{\bf n}$, $H_{n,0}=\{0\}$ and $H_{n,a}=\emptyset$ with $a>n$ or $a<0$. For $\vec{\bf n}=(n,n-1,\cdots,2,1)$, set $\beta_{\vec{\bf n}}=\beta_n\beta_{n-1}\cdots\beta_2\beta_1$ and $\vec{\beta}_{\vec{\bf n}}=(\beta_n,\beta_{n-1},\cdots,\beta_1)$. Define
\begin{align*}
\vec{\alpha}_{\vec{\bf m}}\cup\vec{\beta}_{\vec{\bf n}}\triangleq
(\alpha_m,\cdots,\alpha_1,\beta_n,\cdots,\beta_1),\quad
\vec{\alpha}_{\vec{\bf m}}\setminus\{\alpha_i\}\triangleq(\alpha_m,\cdots,\alpha_{i+1},\alpha_{i-1},
\cdots,\alpha_1).
\end{align*}
and $|\vec{\alpha}|\triangleq\alpha_m+\cdots+\alpha_1$.
.

Next, let us discuss the relations on the free Fermions $\psi_i$ and $\psi_j^*$, which will be used to the bilinear equations of the Darboux transformations of the KP and modified KP hierarchies.

\begin{lemma}\label{s4beta}
For $\beta_i\in V=\oplus_{l}\mathbb{C}\psi_l$ and $\beta_j^*\in V^*=\oplus_{l}\mathbb{C}\psi^*_l$, one has the following relations
\begin{align*}
&S(\beta_{\vec{\bf n}}\otimes 1)=(-1)^n (\beta_{\vec{\bf n}}\otimes 1) S,\quad S(1\otimes\beta_{\vec{\bf k}}^*)=(-1)^k(1\otimes\beta_{\vec{\bf k}}^*)S,\\
&S(1\otimes\beta_{\vec{\bf n}})=\sum_{l=1}^n(-1)^{n-l}\beta_l\otimes\beta_{\vec{\bf n}\setminus\{l\}}+(-1)^n (1\otimes\beta_{\vec{\bf n}}) S,\\
&S(\beta_{\vec{\bf k}}^*\otimes 1)=\sum_{l=1}^k(-1)^{k-l}\beta_{\vec{\bf k}\setminus\{l\}}^*\otimes \beta_l^*+(-1)^k (\beta_{\vec{\bf k}}^*\otimes 1)S.
\end{align*}
\end{lemma}
\begin{proof}
Note that $S(\beta_{\vec{\bf n}}\otimes 1)$ and $S(1\otimes\beta_{\vec{\bf k}}^*)$ are obvious. Next we only discuss $S(1\otimes\beta_{\vec{\bf n}})$, since $S(\beta_{\vec{\bf k}}^*\otimes 1)$ is almost the same.

Actually by using the fact $\sum_j[\psi_j^*,\beta_l]_+\psi_j=\beta_l$,
\begin{align*}
&S(1\otimes\beta_{\vec{\bf n}})=
\sum_j(\psi_j\otimes\psi_j^*)\cdot(1\otimes\beta_{\vec{\bf n}})\\
=&\sum_j\sum_{l=1}^n(-1)^{n-l}[\psi_j^*,\beta_l]_+\psi_j\otimes
\beta_{\vec{\bf n}\setminus\{l\}}
+(-1)^n\sum_j(\psi_j\otimes\beta_{\vec{\bf n}}\psi_j^*)\\
=&\sum_{l=1}^n(-1)^{n-l}\beta_l\otimes\beta_{\vec{\bf n}\setminus\{l\}}+(-1)^n (1\otimes\beta_{\vec{\bf n}}) S.
\end{align*}
\end{proof}
\begin{lemma}\label{sl4beta}
For $\beta_i\in V=\oplus_{l}\mathbb{C}\psi_l$ and $\beta_j^*\in V^*=\oplus_{l}\mathbb{C}\psi^*_l$,
\begin{align*}
&S^l(\beta_{\vec{\bf n}}\otimes 1)=(-1)^{nl} (\beta_{\vec{\bf n}}\otimes 1) S^l,\quad S^l(1\otimes\beta_{\vec{\bf k}}^*)=(-1)^{kl}(1\otimes\beta_{\vec{\bf k}}^*)S^l,\nonumber\\
&S^l(1\otimes\beta_{\vec{\bf n}})=(-1)^{nl}\sum_{j=0}^lC_l^jj!A_{n,j}^+S^{l-j},\quad
S^l(\beta_{\vec{\bf k}}^*\otimes 1)=(-1)^{kl}\sum_{j=0}^lC_l^jj!A_{k,j}^-S^{l-j},
\end{align*}
where for $n\geq j\geq 0$,
\begin{align*}
A_{n,j}^+=\sum_{\vec{\gamma}\in H_{n,j}}(-1)^{-|\vec{\gamma}|}\beta_{\vec{\gamma}}
\otimes\beta_{\vec{\bf n}\setminus\vec{\gamma}},\quad
A_{n,j}^-=\sum_{\vec{\gamma}\in H_{n,j}}(-1)^{-|\vec{\gamma}|}\beta^*_{\vec{\bf n}\setminus\vec{\gamma}}\otimes\beta^*_{\vec{\gamma}}.
\end{align*}
In particular, $A_{n,0}^+=1\otimes\beta_{\vec{\bf n}}$, $A_{n,0}^-=\beta_{\vec{\bf n}}^*\otimes1$ and $A_{n,j}^\pm=0$ for $n<j$ or $j<0$.
\end{lemma}
\begin{proof}
The results for $S^l(\beta_{\vec{\bf n}}\otimes 1)$ and $S^l(1\otimes\beta_{\vec{\bf k}}^*)$ are obviously by using Lemma \ref{s4beta}. Next we try to prove $S^l(1\otimes\beta_{\vec{\bf n}})$, where the key is to compute $SA_{n,j}^+$. Firstly note that $A_{n,j}^+=\sum_{l=j}^n\Big(\beta_{l}\otimes\beta_{\vec{\bf n}\setminus\vec{ l}}\Big)A_{l-1,j-1}^+$,
then by Lemma \ref{s4beta},
\begin{align*}
SA_{n,j}^+=(-1)^nA_{n,j+1}^++\sum_{l=j}^n(-1)^{n-1}
\Big(\beta_{l}\otimes\beta_{\vec{\bf n}\setminus\vec{l}}\Big)SA_{l-1,j-1}^+.
\end{align*}
Repeat the procedure above, one at last obtains
\begin{align*}
SA_{n,j}^+=&(-1)^njA_{n,j+1}^++\sum_{\vec{\gamma}\in H_{n,j}}(-1)^{n-1-|\vec{\gamma}|+\gamma_1}
\Big(\beta_{\gamma}\otimes\beta_{(\vec{\bf n}\setminus\vec{\gamma})\setminus\overrightarrow{\gamma_1-1}}\Big)
SA_{\gamma_1-1,0}^+\\
=&(-1)^n\Big((j+1)A_{n,j+1}^++A_{n,j}^+S\Big),
\end{align*}
where $\vec{\gamma}=(\gamma_j,\cdots,\gamma_1)\in H_{n,j}$.
After the preparation above, one can easily prove the results for $S^l(1\otimes\beta_{\vec{\bf n}})$ by induction on $l$. Similar discussion can be done for $S^l(\beta_{\vec{\bf k}}^*\otimes 1)$.
\end{proof}
By Lemma \ref{sl4beta}, one can obtain the lemma below
\begin{lemma}\label{slbetaotimesbeta}
For $\tilde{\beta}_i,\bar\beta_i,\beta_i\in V$ and $\tilde{\beta}^*_j,\bar\beta^*_j,\beta^*_j\in V^*$, the actions of $S^l$ on $\tilde{\beta}^*_{\overrightarrow{\bf {k_1}}}\beta^*_{\vec{\bf {k}}}\tilde{\beta}_{\overrightarrow{\bf {n_1}}}\beta_{\vec{\bf {n}}}\otimes\bar {\beta}^*_{\overrightarrow{\bf {k_2}}}\beta^*_{\vec{\bf {k}}}\bar \beta_{\overrightarrow{\bf {n_2}}}\beta_{\vec{\bf {n}}}$ and $\tilde{\beta}_{\overrightarrow{\bf {n_1}}}\beta_{\vec{\bf {n}}}\tilde{\beta}^*_{\overrightarrow{\bf {k_1}}}\beta^*_{\vec{\bf {k}}}\otimes\bar \beta_{\overrightarrow{\bf {n_2}}}\beta_{\vec{\bf {n}}}\bar {\beta}^*_{\overrightarrow{\bf {k_2}}}\beta^*_{\vec{\bf {k}}}$ are given by
\begin{align}
&S^{l}\Big(\tilde{\beta}^*_{\overrightarrow{\bf {k_1}}}\beta^*_{\vec{\bf {k}}}\tilde{\beta}_{\overrightarrow{\bf {n_1}}}\beta_{\vec{\bf {n}}}\otimes\bar {\beta}^*_{\overrightarrow{\bf {k_2}}}\beta^*_{\vec{\bf {k}}}\bar \beta_{\overrightarrow{\bf {n_2}}}\beta_{\vec{\bf {n}}}\Big)=\sum_{j=0}^lC_{l}^j
(-1)^{(l-j)(n_1+n_2+k_1+k_2)}j!B_{k_1,n_2,j}^+S^{l-j},\quad l\geq 0, \label{slak1n2}\\
&S^{l}\Big(\tilde{\beta}_{\overrightarrow{\bf {n_1}}}\beta_{\vec{\bf {n}}}\tilde{\beta}^*_{\overrightarrow{\bf {k_1}}}\beta^*_{\vec{\bf {k}}}\otimes\bar \beta_{\overrightarrow{\bf {n_2}}}\beta_{\vec{\bf {n}}}\bar {\beta}^*_{\overrightarrow{\bf {k_2}}}\beta^*_{\vec{\bf {k}}}\Big)=\sum_{j=0}^lC_{l}^j
(-1)^{(l-j)(n_1+n_2+k_1+k_2)}j!B_{k_1,n_2,j}^-S^{l-j},\quad l\geq 0.
\end{align}
Here $B_{k_1,n_2,j}^{\pm}$ is defined in the way below,
\begin{align*}
B_{k_1,n_2,j}^+=&
\sum_{a=0}^j\sum_{\vec{\gamma}\in H_{k_1,j-a}}\sum_{\vec{\delta}\in H_{n_2,a}}(-1)^{a(k_2+n_2)+jk_1-|\vec{\gamma}|-|\vec{\delta}|}\tilde{\beta}^*_{\overrightarrow{\bf k_1}\setminus \vec{\gamma}}\beta^*_{\vec{\bf k}}\bar \beta_{\vec{\delta}}\tilde{\beta}_{\overrightarrow{\bf {n_1}}}\beta_{\vec{\bf n}}\otimes \tilde{\beta}^*_{\vec{\gamma}}\bar{\beta}^*_{\overrightarrow{\bf k_2}}\beta^*_{\vec{\bf k}}\bar{\beta}_{\overrightarrow{\bf n_2}\setminus\vec{\delta}}\beta_{\vec{\bf n}},\\
B_{k_1,n_2,j}^-=&
\sum_{a=0}^j\sum_{\vec{\gamma}\in H_{n_2,j-a}}\sum_{\vec{\delta}\in H_{k_1,a}}(-1)^{a(k_1+n_1)+jn_2-|\vec{\gamma}|-|\vec{\delta}|}\bar \beta_{\vec{\gamma}}\tilde{\beta}_{\overrightarrow{\bf {n_1}}}\beta_{\vec{\bf n}}\tilde{\beta}^*_{\overrightarrow{\bf k_1}\setminus \vec{\delta}}\beta^*_{\vec{\bf k}}\otimes\bar{\beta}_{\overrightarrow{\bf n_2}\setminus\vec{\gamma}}\beta_{\vec{\bf n}} \tilde{\beta}^*_{\vec{\delta}}\bar{\beta}^*_{\overrightarrow{\bf k_2}}\beta^*_{\vec{\bf k}}.
\end{align*}
 Particularly, $B_{k_1,n_2,j}^\pm=0$ for $k_1+n_2<j$ and
\begin{align*}
&B_{k_1,n_2,k_1+n_2}^+=(-1)^{n_2(k_1+k_2)+\frac{k_1(k_1-1)}{2}
+\frac{n_2(n_2-1)}{2}}\beta^*_{\vec{\bf k}}\bar \beta_{\overrightarrow{\bf n_2}}\tilde{\beta}_{\overrightarrow{\bf {n_1}}}\beta_{\vec{\bf n}}\otimes \tilde{\beta}^*_{\overrightarrow{\bf k_1}}\bar{\beta}^*_{\overrightarrow{\bf k_2}}\beta^*_{\vec{\bf k}}\beta_{\vec{\bf n}},\\
&B_{k_1,n_2,k_1+n_2}^-=(-1)^{k_1(n_1+n_2)+\frac{k_1(k_1-1)}{2}
+\frac{n_2(n_2-1)}{2}}\bar \beta_{\overrightarrow{\bf n_2}}\tilde{\beta}_{\overrightarrow{\bf {n_1}}}\beta_{\vec{\bf n}}\beta^*_{\vec{\bf k}}\otimes\beta_{\vec{\bf n}}\tilde{\beta}^*_{\overrightarrow{\bf k_1}}\bar{\beta}^*_{\overrightarrow{\bf k_2}}\beta^*_{\vec{\bf k}}.
\end{align*}
\end{lemma}
Now let us switch to discussing the relations of the neutral free Fermions, which will be needed in the constructions of the bilinear equations in the BKP Darboux transformations.
\begin{lemma}\label{sbbeta}
For $\beta_{i}\in V_B=\oplus_l\mathbb{C}\phi_l$, one has the following relations
\begin{align*}
&S_B(\beta_{\vec{\bf n}}\otimes 1)=\sum_{l=1}^n(-1)^{n-l}\beta_{\vec{\bf n}\setminus\{l\}}\otimes\beta_l+(-1)^n (\beta_{\vec{\bf n}}\otimes 1) S_B,\\
&S_B(1\otimes\beta_{\vec{\bf n}})=\sum_{l=1}^n(-1)^{n-l}\beta_l\otimes\beta_{\vec{\bf n}\setminus\{l\}}+(-1)^n (1\otimes\beta_{\vec{\bf n}}) S_B.
\end{align*}
\end{lemma}
\begin{proof}
This lemma can be proved by the similar method in Lemma \ref{s4beta} with $\beta_l=\sum_j(-1)^j[\beta_l,\phi_{-j}]_+\phi_j$.
\end{proof}
For $\beta_{i}\in V_B$ with $i=1,2,\cdots,n$ and $n\geq j\geq 0$, denote \begin{align*}
\mathfrak{A}_{n,j}^+=\sum_{\vec{\gamma}\in H_{n,j}}(-1)^{-|\vec{\gamma}|}\beta_{\vec{\bf n}\setminus\vec{\gamma}}
\otimes\beta_{\vec{\gamma}},\quad
\mathfrak{A}_{k,j}^-=\sum_{\vec{\gamma}\in H_{k,j}}(-1)^{-|\vec{\gamma}|}\beta_{\vec{\gamma}}\otimes\beta_{\vec{\bf k}\setminus\vec{\gamma}}.
\end{align*}
It is obviously that $\mathfrak{A}^{\pm}_{n,j}=0$ for $j>n$ or $j<0$ and $\mathfrak{A}_{n,0}^+=\beta_{\vec{\bf n}}\otimes 1$, $\mathfrak{A}_{n,0}^-=1\otimes \beta_{\vec{\bf n}}$.
\begin{lemma}\label{sblbeta}
For $\beta_{i}\in V_B=\oplus_l\mathbb{C}\phi_l$, $S_B^l(\beta_{\vec{\bf n}}\otimes 1)$ and $S_B^l(1\otimes \beta_{\vec{\bf n}})$ can be computed in the way below.
\begin{align*}
S_B^l(\mathfrak{A}_{n,0}^{\pm})=(-1)^{nl}\sum_{j=0}^lC_l^j
\sum_{k=0}^{[j/2]}a_{n,j,k}\mathfrak{A}_{n,j-2k}^{\pm}S^{l-j},
\end{align*}
where $[j/2]$ means the integer that is no greater than $j/2$, and $a_{n,j,k}$ is some constant satisfying the following recursion relation, that is, $a_{n,j+1,k}=na_{n,j,k-1}+(a_{n,j,k}-a_{n,j,k-1})(j-2k+1)$ with $a_{n,j,0}=j!$ and $a_{n,j,k}=0$ if $k<0$ or $k>[j/2]$.
\end{lemma}
\begin{proof}
Let us firstly compute $S_B\mathfrak{A}_{n,j}^+$. By using  $\mathfrak{A}_{n,j}^+=\sum_{l=j}^n\Big(\beta_{\vec{\bf n}\setminus\vec{l}}\otimes\beta_{l}\Big)\mathfrak{A}_{l-1,j-1}^+$ and the formula $[S_B(\beta_{\vec{\bf n}}\otimes 1),1\otimes\beta_l]_+=\beta_{l}\beta_{\vec{\bf n}}\otimes 1$ derived by Lemma \ref{sbbeta},
\begin{align*}
S_B\mathfrak{A}_{n,j}^+
=&(-1)^n\mathfrak{A}^+_{n,j+1}+\sum_{l=j}^n(-1)^{n-1}(\beta_{\vec{\bf n}\setminus\vec{l}}\otimes\beta_{\vec{l}})S_B\mathfrak{A}^+_{l-1,j-1}
+\sum_{\vec{\gamma}\in H_{n,j-1}}(-1)^{n-|\vec{\gamma}|}(n-\gamma_{j-1})(\beta_{\vec{\bf n}\setminus\vec{\gamma}}
\otimes\beta_{\vec{\gamma}}).
\end{align*}
Further one can obtain
\begin{align*}
S_B\mathfrak{A}_{n,j}^+=&(-1)^n\Big((j+1)\mathfrak{A}_{n,j+1}^+
+\mathfrak{A}_{n,j}^+S_B\Big)\\
&+\sum_{\vec{\gamma}\in H_{n,j-1}}\sum_{l=1}^{j}(-1)^{n-|\vec{\gamma}|}(\gamma_{l}-\gamma_{l-1}-1)\Big(\beta_{\vec{\bf n}\setminus\vec{\gamma}}\otimes \beta_{\vec{\gamma}}\Big)\\
=& (-1)^n\Big((j+1)\mathfrak{A}_{n,j+1}^++\mathfrak{A}_{n,j}^+S_B
+(n-j+1)\mathfrak{A}_{n,j-1}^+\Big),
\end{align*}
where $\gamma_j=n$ and $\gamma_0=0$. Similarly, one can prove
\begin{align*}
S_B\mathfrak{A}_{n,j}^-
= (-1)^n\Big((j+1)\mathfrak{A}_{n,j+1}^-+\mathfrak{A}_{n,j}^-S_B
+(n-j+1)\mathfrak{A}_{n,j-1}^-\Big).
\end{align*}
After the preparation above, one can easily obtain by induction on $l$,
\begin{align*}
S_B^l(\mathfrak{A}_{n,0}^\pm)=(-1)^{nl}\sum_{j=0}^l
\mathfrak{A}^\pm_{n,j;l}S_B^{l-j},
\end{align*}
where $\mathfrak{A}^\pm_{n,j;l}=\sum_{k=0}^{[(j+1)/2]}a_{n,j,k;l}\mathfrak{A}_{n,j-2k}^{\pm}$ and $a_{n,j,k;l}$ is some constant satisfying the following recursion relations
\begin{align*}
&a_{n,0,0;l+1}=a_{n,0,0;l},\quad a_{n,l+1,k;l+1}=na_{n,l,k-1;l}+(l-2k+1)(a_{n,l,k;l}-a_{n,l,k-1;l}),\\
&a_{n,j,k;l+1}=a_{n,j,k;l}+na_{n,j-1,k-1;l}+(j-2k)(a_{n,j-1,k;l}-a_{j-1,k-1;l}),\quad k=0,1,\cdots,[(j+1)/2].
\end{align*}
Here $a_{n,0,0;0}=1$ and if $j,k$ and $l$ do not satisfy $0\leq j\leq l$ and $0\leq k\leq [j/2]$, we will set $a_{n,j,k;l}=0$. If set $a_{n,j,k}=a_{n,j,k;j}$, then $a_{n,j,0}=j!$ can be proved by induction on $j$.
$\mathfrak{A}^\pm_{n,j;l}$ satisfies the following recursion relations
\begin{align*}
\mathfrak{A}^\pm_{n,j;l+1}=&(-1)^{n(l+1)}\Big( S_B^{l+1}(\mathfrak{A}^\pm_{n,0})S_B^{-l-1+j}\Big)_{[0]}=(-1)^{n}\Big( S_B\sum_{j_1=0}^l\mathfrak{A}^\pm_{n,j_1;l}S_B^{j-j_1-1}\Big)_{[0]}\\
=&(-1)^{n}\Big(\sum_{j_1=0}^l\big(\mathfrak{A}^\pm_{n,j_1;l}S_B
+(S_B\mathfrak{A}^\pm_{n,j_1;l})_{[0]}\big)S_B^{j-j_1-1}\Big)_{[0]}\\
=&(-1)^n(S_B\mathfrak{A}^\pm_{n,j-1;l})_{[0]}
+\mathfrak{A}^\pm_{n,j;l},
\end{align*}
where $(\sum_ib_iS_B^i)_{[0]}=b_0$. In particular, $(-1)^n(S_B\mathfrak{A}^\pm_{n,j-1;j-1})_{[0]}=\mathfrak{A}^\pm_{n,j;j}$ since $\mathfrak{A}^\pm_{n,j;j-1}=0$ by noting $a_{n,j,k;j-1}=0$.

At last this lemma can be proved if $a_{n,j,k;l}=C_l^ja_{n,j,k}$, which is equivalent to $\mathfrak{A}^\pm_{n,j;l}=C_l^j\mathfrak{A}^\pm_{n,j;j}$. In fact, if it holds for $l$, then by the recursion relation of $\mathfrak{A}^\pm_{n,j;l}$, one can find
\begin{align*}
\mathfrak{A}^\pm_{n,j;l+1}
=&(-1)^n(S_B\mathfrak{A}^\pm_{n,j-1;l})_{[0]}
+\mathfrak{A}^\pm_{n,j;l}=(-1)^nC_l^{j-1}(S_B\mathfrak{A}^\pm_{n,j-1;j-1})_{[0]}
+C_l^j\mathfrak{A}^\pm_{n,j;j}\\
=&(C_l^{j-1}+C_l^{j})\mathfrak{A}^\pm_{n,j;j}
=C_{l+1}^j\mathfrak{A}^\pm_{n,j;j}
\end{align*}
still holds for $l+1$.
\end{proof}

\section{The Darboux transformations of the KP hierarchy}
In this section, we firstly review some basic facts on the Darboux transformations of the KP hierarchy, especially the fermionic forms. Then we discuss the transformed tau functions in the fermionic forms under the successive applications of the Darboux transformations. At last, the bilinear equations in the Darboux chains are investigated and some examples are given. Here we would like to point out that the result in Subsection \ref{subsectdti} for the Fermionic tau functions under the mixed using $T_D$ and $T_I$ is the key to find the bilinear equations.
\subsection{Review on the Darboux transformations of the KP hierarchy}\label{subsectionkptdti}
Just as we know, there are two basic Darboux transformations \cite{Chau1992,Oevelpa1993,Oevelrmp1993,He2002} of the KP hierarchy, that is,
\begin{align*}
T_d(\Phi)=\Phi\partial\Phi^{-1},\quad T_i(\Psi)=\Psi^{-1}\partial^{-1}\Psi,
\end{align*}
where $\Phi$ and $\Psi$ are the eigenfunction and the adjoint eigenfunction of the KP hierarchy respectively. $T_d$ and $T_i$ can be derived \cite{Shaw1997} by the Miura links between the KP and modified KP hierarchies through $L\xrightarrow{\text{anti-Miura}}\mathcal{L}\xrightarrow{\text{Miura}}L^{[1]}$.
Under the Darboux transformation, the corresponding changes of the dressing operator $W$, eigenfunction $\Phi_1$, the adjoint eigenfunction $\Psi_1$ and the tau function $\tau$ are listed in Table III\cite{Chau1992,He2002,Oevelpa1993,Oevelrmp1993}.
\begin{center}
\begin{tabular}{lllll}
\multicolumn{5}{c}{Table III. Darboux transformation: KP $\rightarrow$ KP}\\
\hline\hline
\ \ \ $L\rightarrow L^{[1]} $ & \ \ $W^{[1]}=$ & \ \ \ \ $\Phi_1^{[1]}=$& \ \ \ $\Psi_1^{[1]}=$ & $\tau^{[1]}$\\
\hline
 \ $T_d(\Phi)=\Phi\partial\Phi^{-1}$ & \ $\Phi\partial\Phi^{-1}W\partial^{-1} $ &\ \ \ $\Phi(\Phi_1/\Phi)_x$ & \ $-\int\Phi\Psi_{1}dx/\Phi$& $\Phi\tau$\\
 $T_i(\Psi)=\Psi^{-1}\pa^{-1}\Psi$ & \ $\Psi^{-1}\pa^{-1}\Psi W \pa $&\ \ $\int\Psi\Phi_1 dx/\Psi$ & \ $-\Psi(\Psi_{1}/\Psi)_x$& $\Psi\tau$\\
\hline\hline
\end{tabular}
\end{center}
Here we use $A^{[n]}$ to denote the transformed object $A$ under $n$-step Darboux transformations. Sometimes in order to distinguish the actions of $T_d$ and $T_i$, $A^{+[n]}$ is used to indicate the transformed $A$ under $n$-step $T_d$, while $A^{-[n]}$ is for $T_i$.

From (\ref{tauhattildetaubilinear}) and Proposition \ref{mkp2tau}, one can obtain the corollary \cite{Willox1998,Kac2018} below.
\begin{corollary}\label{cortautrankp}
If $\tau(t)$ is a tau function of the KP hierarchy, $\Phi(t)$ is any eigenfunction and $\Psi(t)$ is any adjoint eigenfunction, then $\tau^{+[1]}(t)=\Phi(t)\tau(t)$ and $\tau^{-[1]}(t)=\Psi(t)\tau(t)$ are still the tau functions of the KP hierarchy, that is
\begin{align*}
&{\rm Res}_{\lambda}\tau^{\pm[1]}(t-\varepsilon(\lambda^{-1}))
\tau^{\pm[1]}(t'+\varepsilon(\lambda^{-1}))
e^{\xi(t-t',\lambda)}=0,
\end{align*}
\end{corollary}

By direct computation, one can obtain the lemma \cite{Oevelrmp1993,Willox2004} below
\begin{lemma}\label{commkpdt}
The Bianchi diagram for Darboux transformations $T_d$ and $T_i$ can commute,
$$
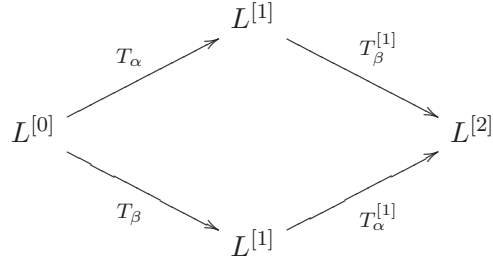
\xymatrix{&&L^{[1]}\ar[drr]^{T_\beta^{[1]}}&&\\
L^{[0]}\ar[urr]^{T_\alpha}\ar[drr]_{T_\beta}& &&& L^{[2]}\\
&&L^{[1]}\ar[urr]_{T_\alpha^{[1]}}&&}$$
that is,
\begin{align*}
T_d(\Phi_1^{+[1]})T_d(\Phi_2)=&T_d(\Phi_2^{+[1]})T_d(\Phi_1),\quad
T_i(\Psi_1^{-[1]})T_i(\Psi_2)=T_i(\Psi_2^{-[1]})T_i(\Psi_1),\\
T_d(\Phi^{-[1]})T_i(\Psi)=&T_i(\Psi^{+[1]})T_d(\Phi).
\end{align*}
Here $T_\beta^{[1]}$ means the Darboux transformation generated by the (adjoint) eigenfunction of $L^{[1]}$.
\end{lemma}
Though almost everything in the KP hierarchy can be kept under the Darboux transformation, there is still some facts that are not clear after the Darboux transformation, for example, the relations between the tau functions and the wave functions, which is given in the lemma below.
\begin{lemma}
Under the Darboux transformation $T_d(\Phi)$ and $T_i(\Psi)$,
\begin{align*}
T_d(\Phi)(\psi(t,\lambda))=\frac{\lambda \tau^{+[1]}(t-\varepsilon(\lambda^{-1}))}{\tau^{+[1]}(t)}e^{\xi(t,\lambda)},\quad &(T_d^*(\Phi))^{-1}(\psi^*(t,\lambda))=\frac{\tau^{+[1]}(t+\varepsilon(\lambda^{-1}))}{\lambda \tau^{+[1]}(t)}e^{-\xi(t,\lambda)};\\
T_i(\Psi)(\psi(t,\lambda))=\frac{\tau^{-[1]}(t-\varepsilon(\lambda^{-1}))}
{\lambda\tau^{-[1]}(t)}e^{\xi(t,\lambda)},\quad &(T_i^*(\Psi))^{-1}(\psi^*(t,\lambda))
=\frac{\lambda\tau^{-[1]}(t+\varepsilon(\lambda^{-1}))}{\tau^{-[1]}(t)}e^{-\xi(t,\lambda)}&
\end{align*}
\end{lemma}
\begin{proof}
Here we will only discuss the case $T_d$, since it is almost the same in the case $T_i$. Firstly by using Lemma \ref{mkpfaylemma}, one can obtain the result for $T_d(\Phi)(\psi(t,\lambda))$. As for $(T_d^*(\Phi))^{-1}(\psi^*(t,\lambda))$, one can use the expression of $\Omega(\Phi(t),\psi^*(t,\lambda))$ in Proposition \ref{sepexpprop}.
\end{proof}

Next we will express the Darboux transformation in the Fermionic language. For this, if introduce
\begin{align*}
\alpha={\rm Res}_{\lambda}\rho(\lambda)\psi(\lambda),\quad \alpha^*={\rm Res}_{\lambda}\frac{1}{\lambda}\rho^*(\lambda)\psi^*(\lambda),
\end{align*}
with $\rho(\lambda)$ and $\rho^*(\lambda)$ given in Subsection \ref{subseckp},
then by (\ref{kpwavefunction}), (\ref{kpadwavefunction}) and \eqref{spectralkp},
\begin{eqnarray}
\Phi(t)=\frac{\langle 1|e^{H(t)}\alpha g|0\rangle}{\langle 0|e^{H(t)}g|0\rangle}, && \Psi(t)=\frac{\langle -1|e^{H(t)}\alpha^*g|0\rangle}{\langle 0|e^{H(t)}g|0\rangle}.\label{eigenfermi}
\end{eqnarray}
Note that $\alpha\in V=\oplus_{i\in\mathbb{Z}}\mathbb{C}\psi_i$ and $\alpha^*\in V^*=\oplus_{i\in\mathbb{Z}}\mathbb{C}\psi_i^*$.
Therefore one can obtain the following proposition \cite{Willox1998,Chau1992} by Table III.
\begin{proposition}\label{tdtifermiversion}
Assume $\tau(t)=\sigma(g|0\rangle)=\langle 0| e^{H(t)} g|0\rangle$ with $g\in GL_\infty$, then
\begin{itemize}
  \item Under $T_d(\Phi)$, $g|0\rangle\rightarrow \alpha g|0\rangle$, i.e., $\tau(t)\rightarrow \tau^{+[1]}(t)=\langle 1| e^{H(t)} \alpha g|0\rangle$.
  \item Under $T_i(\Psi)$, $g|0\rangle\rightarrow \alpha g|0\rangle$, i.e., $\tau(t)\rightarrow \tau^{-[1]}(t)=\langle -1|e^{H(t)}\alpha^*g|0\rangle$.
\end{itemize}
\end{proposition}
\noindent {\bf Remark}: Slightly abusing the notation, we will use $\tau$ for $g|0\rangle$ if $\tau(t)=\sigma(g|0\rangle)$ and let $z=1$ in the isomorphism $\sigma_t$ between $\mathcal{F}$ and $\mathcal{B}$. In these notations, $\tau\xrightarrow{T_d(\Phi)}\tau^{+[1]}=\alpha\tau$ and $\tau\xrightarrow{T_i(\Psi)}\tau^{-[1]}=\alpha^*\tau$. In what follows, we always believe the initial tau function is $\tau^{[0]}=g|0\rangle$, since other cases are almost the same. Under the Darboux transformations, one can easily find
$\mathcal{F}_l\xrightarrow{T_d}\mathcal{F}_{l+1}$ and $\mathcal{F}_l\xrightarrow{T_i}\mathcal{F}_{l-1}$.

\subsection{Successive applications of $T_d$ in the Fermionic picture}
Before further discussion, let us see the changes of SEPs for the KP hierarchy, which can be computed by direct computations according to Table III.
\begin{lemma}\label{transep}
Under the Darboux transformation $T_d(\Phi)$,
\begin{align*}
&\Omega(\Phi_1(t)^{+[1]},\psi^*(t,\lambda)^{+[1]})=
\lambda\Big(\Omega(\Phi_1(t),\psi^*(t,\lambda))
-\frac{\Phi_1(t)}{\Phi(t)}\Omega(\Phi(t),\psi^*(t,\lambda))\Big),\\
&\Omega(\psi(t,\lambda)^{+[1]},\Psi_1(t)^{+[1]})
=\frac{1}{\lambda}\Big(\Omega(\psi(t,\lambda),\Psi_1(t))
-\frac{\Omega(\Phi(t),\Psi_1(t))}{\Phi(t)}\psi(t,\lambda)
\Big).
\end{align*}
Under the Darboux transformation $T_i(\Psi)$,
\begin{align*}
&\Omega(\Phi_1(t)^{-[1]},\psi^*(t,\lambda)^{-[1]})=\frac{1}{\lambda}
\Big(
\Omega(\Phi_1(t),\psi(t,\lambda))
-\frac{\Omega(\Phi_1(t),\Psi(t))}{\Psi(t)}\psi^*(t,\lambda)\Big),\\
&\Omega(\psi(t,\lambda)^{-[1]},\Psi_1(t)^{-[1]})
=\lambda\Big(\Omega(\psi(t,\lambda),\Psi_1(t))
-\frac{\Psi_1(t)}{\Psi(t)}\Omega(\psi(t,\lambda),\Psi(t))
\Big).
\end{align*}
\end{lemma}
Firstly,
consider the change of $\tau(t)$ under successive applications of $T_d$. For this, let us discuss the following Darboux transformation $T_d$ chain.
\begin{eqnarray}
&&L\xrightarrow{T_d(\Phi_1(t))}L^{[1]}\xrightarrow{T_d(\Phi^{[1]}_2(t))}L^{[2]}
  \rightarrow\cdots\rightarrow L^{[n-1]}\xrightarrow{T_d(\Phi_{n}^{[n-1]}(t))}L^{[n]},\label{tdphi}
\end{eqnarray}
where $\Phi_1,\cdots,\Phi_n$ are independent eigenfunctions of the KP hierarchy. Further we assume $\alpha_i\in V$ such that $\Phi_i(t)=\frac{\langle 1|e^{H(t)}\alpha_i g|0\rangle}{\langle 0|e^{H(t)}g|0\rangle}$ and denote $T_d^{[\vec{\bf n}]}=T_d(\Phi_{n}^{[n-1]}(t))\cdots T_d(\Phi_1(t))$. In what follows, $A^{+[\vec{\bf n}]}$ (for example $\tau^{+[\vec{\bf n}]}$) is also used as the transformed research objects $A$ under $T_d^{[\vec{\bf n}]}$.

For the Darboux chain (\ref{tdphi}), denote $\rho_i^{[j]}(\lambda)$ and $\alpha_i^{[j]}\in V$ in the way below,
\begin{align*}
\rho_i^{[j]}(\lambda)=-\Omega(\Phi_i^{[j]}(t'),\psi^{*}(t',\lambda)^{[j]}),
\quad\alpha_i^{[j]}={\rm Res}_{\lambda}\frac{1}{\lambda^j}\rho_i^{[j]}(\lambda)
\psi(\lambda).
\end{align*}
By Lemma \ref{transep}, one can find for $i>j$
\begin{align}
&\rho_i^{[j]}(\lambda)=\lambda\big(
\rho_i^{[j-1]}(\lambda)+c^{[j-1]}_{i}\rho_{j}^{[j-1]}(\lambda)\big),\quad
\alpha_i^{[j]}=\alpha_i^{[j-1]}+c^{[j-1]}_{i}\alpha_{j}^{[j-1]},\quad \rho_i^{[i]}(\lambda)=0,\quad \alpha_i^{[i]}=0.
\label{alpharelation}
\end{align}
where $c^{[j-1]}_{i}=-\frac{\Phi_{i}^{[j-1]}(t')}{\Phi_{j}^{[j-1]}(t')}$ can be viewed as a constant independent of $t$. Then one has following lemma.
\begin{proposition}\label{ndttd}
Under the $n$-step Darboux transformation $T_d$ (see chain (\ref{tdphi})),
\begin{align*}
\tau^{+[\vec{\bf n}]}=\alpha_n^{[n-1]}\cdots \alpha_1^{[0]}\tau^{[0]}=\alpha_n\cdots \alpha_1\tau^{[0]},
\end{align*}
that is,
\begin{align*}
\tau^{+[\vec{\bf n}]}(t)=\langle n|e^{H(t)}\alpha_n^{[n-1]}\cdots \alpha_1^{[0]}g|0\rangle=\langle n|e^{H(t)}\alpha_n\cdots \alpha_1g|0\rangle,
\end{align*}
where we have set $\tau^{[0]}(t)=\langle 0|e^{H(t)}g|0\rangle$.
\end{proposition}
\begin{proof}
Obviously, it is correct for $n=0$ and $n=1$. Next assume $\tau^{[n]}(t)=\langle n|e^{H(t)}\alpha_n^{[n-1]}\cdots \alpha_1^{[0]}g|0\rangle$ holds $n$. Then
\begin{align*}
\tau^{[n+1]}(t)=&\Phi_{n+1}(t)^{[n]}\tau^{[n]}(t)
={\rm Res}_\lambda\rho_{n+1}^{[n]}(\lambda)\psi(t,\lambda)^{[n]}\tau^{[n]}(t)
={\rm Res}_\lambda\rho_{n+1}^{[n]}(\lambda)\tau^{[n]}(t-\varepsilon(\lambda^{-1}))
e^{\xi(t,\lambda)}\\
=&{\rm Res}_\lambda\rho_{n+1}^{[n]}(\lambda)\langle n|e^{H(t-\varepsilon(\lambda^{-1}))}\alpha_n^{[n-1]}\cdots \alpha_1^{[0]}g|0\rangle e^{\xi(t,\lambda)}\\
=&{\rm Res}_\lambda\lambda^{-n}\rho_{n+1}^{[n]}(\lambda)\langle n+1|e^{H(t)}\psi(\lambda)\alpha_n^{[n-1]}\cdots \alpha_1^{[0]}g|0\rangle\\
=&\langle n+1|e^{H(t)}\alpha_{n+1}^{[n]}\alpha_n^{[n-1]}\cdots \alpha_1^{[0]}g|0\rangle.
\end{align*}
Based upon this, $\tau^{[n]}(t)=\langle n|e^{H(t)}\alpha_n\cdots \alpha_1g|0\rangle$ can be obtained by (\ref{alpharelation}) and the fact $\alpha_i^{[j]}$ anti-commutes with each other.
\end{proof}
\noindent{\bf Remark:} For $\vec{\bf n}=(n,n-1,\cdots,2,1)$, denote  $\vec{\alpha}_{\vec{\bf n}}=(\alpha_n,\cdots,\alpha_1)$ and $\tau_{\vec{\alpha}_{\vec{\bf n}}}=\alpha_n\cdots\alpha_1\tau$. According to Proposition \ref{ndttd}, one can know $\tau^{[\vec{\bf n}]}=\tau_{\vec{\alpha}_{\vec{\bf n}}}$.

\subsection{Successive applications of $T_i$ in the Fermionic picture}
Next consider the following chain of Darboux transformation $T_i$,
\begin{eqnarray}
&&L\xrightarrow{T_i(\Psi_1(t))}L^{[1]}\xrightarrow{T_i(\Psi_2^{[1]}(t))}L^{[2]}
  \rightarrow\cdots\rightarrow L^{[n-1]}\xrightarrow{T_i(\Psi_n^{[n-1]}(t))}L^{[n]}.\label{tipsi}
\end{eqnarray}
Here $\Psi_1,\cdots,\Psi_n$ are the independent adjoint eigenfunctions of the KP hierarchy. Denote $\alpha_j^*\in V^*$ such that $\Psi_j(t)=\frac{\langle -1|e^{H(t)}\alpha_j^*g|0\rangle}{\langle 0|e^{H(t)}g|0\rangle}$ and $T_i^{[\vec{\bf n}]}=T_i(\Psi_n^{[n-1]}(t))\cdots T_i(\Psi_1(t))$. The transformed research object $A$ of the KP hierarchy under $T_i^{[\vec{\bf n}]}$ will be denoted as $A^{-[\vec{\bf n}]}$. Let
\begin{align*}
\rho_i^{*[j]}(\lambda)=\Omega(\psi(t',\lambda)^{[j]},\Psi_i^{[j]}(t')),\quad
\alpha_i^{*[j]}={\rm Res}_{\lambda}\frac{1}{\lambda^{j+1}}\rho_i^{*[j]}(\lambda)
\psi^*(\lambda),\quad i>j
\end{align*}
satisfying
\begin{align*}
\rho_i^{*[j]}(\lambda)=\lambda\big(b^{[j-1]}_{i}
\rho_{j}^{*[j-1]}(\lambda)+\rho_i^{*[j-1]}(\lambda)\big),\quad
\alpha_i^{*[j]}=b^{[j-1]}_{i}\alpha_{j}^{*[j-1]}+\alpha_i^{*[j-1]},
\end{align*}
where $b^{[j-1]}_{i}=-\frac{\Psi_i^{[j-1]}(t')}{\Psi_{j}^{[j-1]}(t')}$ is also believed as a constant independent of $t$. And also $\rho_i^{*[i]}=0$ and $\alpha_i^{*[i]}=0$. Similarly, one can obtain the following proposition.
\begin{proposition}\label{ndtti}
For the chain (\ref{tipsi}) of the Darboux transformation $T_i(\Psi)$,
\begin{eqnarray}
\tau^{-[\vec{\bf n}]}=\alpha^{[n-1]*}_n\cdots\alpha^*_1 \tau^{[0]}= \alpha^*_n\cdots\alpha^*_1 \tau^{[0]}.
\end{eqnarray}
\end{proposition}
\noindent{\bf Remark:} By Lemma \ref{commkpdt}, $\alpha_1^*\alpha_2^*\tau\approx\alpha_2^*\alpha_1^*\tau$. If denote $\vec{\alpha}^*_{\vec{\bf n}}=(\alpha^*_n,\cdots,\alpha^*_1)$ and $\tau_{\vec{\alpha}^*_{\vec{\bf n}}}=\alpha^*_n\cdots\alpha^*_1\tau$ with $\tau=g|0\rangle$, then $\tau^{-[\vec{\bf n}]}=\tau_{\vec{\alpha}^*_{\vec{\bf n}}}$.

\subsection{Successive applications of the mixed using $T_d$ and $T_i$ in the Fermionic picture}\label{subsectdti}
Now let us discuss the mixed using $T_d(\Phi)$ and $T_i(\Psi)$. It is more complicated in this case. The possible reason is that $\psi_i$ can not always anticommute or commute with $\psi_j^*$, by noting that $[\psi_i,\psi_j^*]_+=\delta_{ij}$.

Consider the following Darboux chain of $T_d$ and $T_i$,
\begin{eqnarray}
&&L\xrightarrow{T_d(\Phi_1(t))}L^{[1]}
\xrightarrow{T_d(\Phi_2^{[1]}(t))}L^{[1]}
  \rightarrow\cdots\rightarrow L^{[n-1]}\xrightarrow
  {T_d(\Phi_n^{[n-1]}(t))}L^{[n]}\nonumber\\
&&\xrightarrow{T_i(\Psi_{1}^{[n]}(t))}L^{[n+1]}
\xrightarrow{T_i(\Psi_{2}^{[n+1]}(t))}\cdots\rightarrow L^{[n+k-1]}\xrightarrow{T_i(\Psi_{n+k}^{[n+k-1]}(t))}
L^{[n+k]}.\label{tdti}
\end{eqnarray}
Denote $T^{+[\vec{\bf n},\vec{\bf k}]}=T_i(\Psi_k^{[n+k-1]})\cdots T_i(\Psi_1^{[n]})T_d^{[\vec{\bf n}]}$ and $T^{-[\vec{\bf n},\vec{\bf k}]}=T_d(\Phi_n^{[n+k-1]})\cdots T_i(\Phi_1^{[k]})T_i^{[\vec{\bf k}]}$, and let $\tau^{\pm[\vec{\bf n},\vec{\bf k}]}$ be the transformed tau functions under $T^{\pm[\vec{\bf n},\vec{\bf k}]}$. Then it is obviously that $T^{+[\vec{\bf n},\vec{\bf k}]}=T^{-[\vec{\bf n},\vec{\bf k}]}$ and $\tau^{+[\vec{\bf n},\vec{\bf k}]}\sim \tau^{-[\vec{\bf n},\vec{\bf k}]}$.

Let $\rho^{*+[\vec{\bf n}]}_{j}$ and $\rho^{-[\vec{\bf k}]}_{j}$ be the transformed results of $\rho^{*}_{j}$ and $\rho_{j}$ under $T_d^{[\vec{\bf n}]}$ and $T_i^{[\vec{\bf k}]}$ respectively. Set $\alpha_j^{*+[\vec{\bf n}]}={\rm Res}_{\lambda}\lambda^{n-1}\rho^{*+[\vec{\bf n}]}_{j}(\lambda)\psi^*(\lambda)$ and $\alpha_j^{-[\vec{\bf k}]}={\rm Res}_{\lambda}\lambda^{k}\rho^{-[\vec{\bf k}]}_{j}(\lambda)
\psi(\lambda)$, which can also be written as $\alpha_j^{*[n]}$ and $\alpha_j^{[k]}$ for short.

\begin{proposition}\label{kptdtitau}
$\tau^{\pm [\vec{\bf n},\vec{\bf k}]}$ are expressed in the following forms
\begin{align}
\tau^{+[\vec{\bf n},\vec{\bf k}]}=\alpha_k^{*+[\vec{\bf n}]}\cdots\alpha_1^{*+[\vec{\bf n}]}
\alpha_n\cdots\alpha_1\tau^{[0]},\quad \tau^{-[\vec{\bf n},\vec{\bf k}]}=\alpha_n^{-[\vec{\bf k}]}\cdots\alpha_1^{-[\vec{\bf k}]}
\alpha_k^*\cdots\alpha_1^*\tau^{[0]}.
\label{taunpkfrelation}
\end{align}
Further
\begin{align}
\tau^{+[\vec{\bf n},\vec{\bf k}]}=\sum_{a=0}^k\sum_{\vec{\gamma}\in H_{k,a}}\sum_{\vec{\delta}\in H_{n,n-a}}C^{+[n,k]}_{a,\vec{\gamma},\vec{\delta}}
\vec{\alpha}^*_{\vec{\bf k}\setminus\vec{\gamma}}\vec{\alpha}_{\vec{\delta}}\tau^{[0]},\quad \tau^{-[\vec{\bf n},\vec{\bf k}]}=\sum_{a=0}^n\sum_{\vec{\gamma}\in H_{n,a}}\sum_{\vec{\delta}\in H_{k,k-a}}C^{-[n,k]}_{a,\vec{\gamma},\vec{\delta}}
\vec{\alpha}_{\vec{\bf n}\setminus\vec{\gamma}}\vec{\alpha}^*_{\vec{\delta}}\tau^{[0]}.\label{taunpkf2cases}
\end{align}
Here $C^{\pm[n,k]}_{a,\vec{\gamma},\vec{\delta}}$ is some constant independent of $t$, which can be determined by $C^{+[n,k]}_{0,0,\vec{\bf n}}=C^{-[n,k]}_{0,0,\vec{\bf k}}=1$,  and the following recursion relations,
\begin{align}
C^{+[n,k+1]}_{a,\vec{\gamma},\vec{\delta}}=C^{+[n,k]}_{a,\vec{\gamma},\vec{\delta}},\quad
C^{+[n,k+1]}_{a+1,\{k+1\}\cup\vec{\gamma},\vec{\delta}\setminus\{\delta_l\}}=
C^{+[n,k]}_{a,\vec{\gamma},\vec{\delta}}\sum_{i=1}^n(-1)^{n+k-l}
\Psi_{k+1}^{[i]}(t')\frac{\tau_{\vec{\alpha}_{\delta_l}\cup\vec{\alpha}_{\overrightarrow{\bf i-1}}(t')}}
{\tau_{\vec{\alpha}_{\overrightarrow{\bf i-1}}}(t')},
\label{crecursion}\\
C^{-[n+1,k]}_{a,\vec{\xi},\vec{\eta}}=C^{-[n,k]}_{a,\vec{\xi},\vec{\eta}},\quad C^{-[n+1,k]}_{a+1,\{n+1\}\cup\vec{\xi},\vec{\eta}\setminus\{\eta_l\}}=
C^{-[n,k]}_{a,\vec{\xi},\vec{\eta}}\sum_{i=1}^k(-1)^{n+k-l}
\Phi_{n+1}^{[i]}(t')
\frac{\tau_{\vec{\alpha}^*_{\eta_l}\cup\vec{\alpha}^*_{\overrightarrow{\bf i-1}}(t')}}
{\tau_{\vec{\alpha}^*_{\overrightarrow{\bf i-1}}}(t')},
\end{align}
with $\vec{\gamma}\in H_{k,a}$, $\vec{\delta}\in H_{n,n-a}$, $\xi\in H_{n,a}$, $\vec{\eta}\in H_{k,k-a}$.
\end{proposition}
\begin{proof}
By the similar way in Proposition \ref{ndtti}, one can obtain
\begin{align*}
\tau^{+[\vec{\bf n},\vec{\bf k}]}=\alpha_k^{*+[\vec{\bf n}]}\cdots\alpha_1^{*+[\vec{\bf n}]}\alpha_n\cdots\alpha_1\tau^{[0]}.
\end{align*}
Next we will try to prove the result for $\tau^{+[\vec{\bf n},\vec{\bf k}]}$ by induction on $k$. Firstly by Lemma \ref{transep} one can obtain
\begin{align*}
\rho_j^{*+[\vec{\bf n}]}(\lambda)=\frac{1}{\lambda^n}\rho_j^{*[0]}(\lambda)+
\sum_{i=1}^n\frac{1}{\lambda^{n+1-i}}\Psi_j^{[i]}(t')\psi^{[i-1]}(t',\lambda),
\end{align*}
which implies that
\begin{align}
\alpha_j^{*+[\vec{\bf n}]}=\alpha_j^*+\sum_{l\in\mathbb{Z}}\sum_{i=1}^n
\frac{\Psi_j^{[i]}(t')}{\tau_{\vec{\alpha}_{\overrightarrow{\bf i-1}}}(t')}\sigma_{t'}(\psi_l
\alpha_{\overrightarrow{\bf i-1}}
\tau^{[0]})\cdot\psi^*_l.\label{alphastarn}
\end{align}

When $k=1$, one can find
\begin{align*}
\tau^{+[\vec{\bf n},1]}=\alpha_1^{*+[\vec{\bf n}]}\alpha_{\vec{\bf n}}\tau^{[0]}
=\alpha_1^*\alpha_{\vec{n}}\tau^{[0]}+\sum_{i=1}^n
\frac{\Psi_1^{[i]}(t')}{\tau_{\vec{\alpha}_{\overrightarrow{\bf i-1}}}(t')}(\sigma_{t'}\otimes 1)\cdot S(\alpha_{\overrightarrow{\bf i-1}}\tau^{[0]}\otimes \alpha_{\vec{\bf n}}\tau^{[0]}).
\end{align*}
Next compute $S(\alpha_{\overrightarrow{\bf i-1}}\tau^{[0]}\otimes \alpha_{\vec{\bf n}}\tau^{[0]})$ by Lemma \ref{s4beta} and the fact $S(\tau^{[0]}\otimes\tau^{[0]})=0$, that is,
\begin{align*}
S(\alpha_{\overrightarrow{\bf i-1}}\tau^{[0]}\otimes \alpha_{\vec{\bf n}}\tau^{[0]})=S(1\otimes \alpha_{\vec{\bf n}})\cdot(\alpha_{\overrightarrow{\bf i-1}}\tau^{[0]}\otimes \tau^{[0]})=\sum_{l=1}^n(-1)^{n-l}\alpha_{\{l\}\cup\overrightarrow{\bf i-1}}\tau^{[0]}
\otimes \alpha_{\vec{\bf n}\setminus\{l\}}\tau^{[0]}.
\end{align*}
Therefore by setting $C_{1,1,\vec{\bf n}\setminus\{l\}}^{+[n,1]}
=(-1)^{n-l}\sum_{i=1}^l\Psi_1^{[i]}(t')
\tau_{\vec{\alpha}_l\cup\vec{\alpha}_{\overrightarrow{\bf i-1}}}(t')/\tau_{\vec{\alpha}_{\overrightarrow{\bf i-1}}}(t')$,
\begin{align*}
\tau^{+[\vec{\bf n},1]}=\alpha_1^*\alpha_{\vec{\bf n}}\tau^{[0]}+\sum_{l=1}^nC^{+[n,1]}_{1,1,\vec{\bf n}\setminus\{l\}}
\alpha_{\vec{\bf n}\setminus\{l\}}\tau^{[0]},
\end{align*}
which has the form of (\ref{taunpkf2cases}) with $C^{+[n,1]}_{0,0,\vec{\bf n}}=1$.

If $\tau^{+[\vec{\bf n}, \vec{\bf k}]}$ in (\ref{taunpkf2cases}) is correct for $k$, then apply $\alpha_{k+1}^{*+[\vec{\bf n}]}$ on $\tau^{+[\vec{\bf n}, \vec{\bf k}]}$,
\begin{align*}
\tau^{+[\vec{\bf n}, \overrightarrow{\bf k+1}]}=\sum_{a=0}^k\sum_{\vec{\gamma}\in H_{k,a}}\sum_{\vec{\delta}\in H_{n,n-a}}C^{+[n,k]}_{a,\vec{\gamma},\vec{\delta}}
\Big(\alpha_{k+1}^{*[n]}\alpha^*_{\vec{\bf k}\setminus\vec{\gamma}}
\alpha_{\vec{\delta}}\tau^{[0]}\Big).
\end{align*}
Further by Lemma \ref{s4beta},
\begin{align*}
&\alpha_{k+1}^{*[n]}\alpha^*_{\vec{\bf k}\setminus\vec{\gamma}}
\alpha_{\vec{\delta}}\tau^{[0]}=\alpha_{k+1}^*\alpha^*_{\vec{\bf k}\setminus\vec{\gamma}}
\alpha_{\vec{\delta}}\tau^{[0]}+\sum_{i=1}^n
\frac{\Psi_{k+1}^{[i]}(t')}{\tau_{\vec{\alpha}_{\overrightarrow{\bf i-1}}}(t')}(\sigma_{t'}\otimes 1)\cdot S(\alpha_{\overrightarrow{\bf i-1}}\tau^{[0]}\otimes \alpha^*_{\vec{\bf k}\setminus\vec{\gamma}}
\alpha_{\vec{\delta}}\tau^{[0]})\\
=&\alpha_{k+1}^*\alpha^*_{\vec{\bf k}\setminus\vec{\gamma}}
\alpha_{\vec{\delta}}\tau^{[0]}+\sum_{i=1}^n
\sum_{l=1}^{n-a}(-1)^{n+k-l}
\frac{\Psi_{k+1}^{[i]}(t')}{\tau_{\vec{\alpha}_{\overrightarrow{\bf i-1}}}(t')}
\tau_{\vec{\alpha}_{\{\delta_l\}\cup\overrightarrow{i-1}}}(t') \alpha^*_{\vec{\bf k}\setminus\vec{\gamma}}
\alpha_{\vec{\delta}\setminus\{\delta_l\}}\tau^{[0]}
\end{align*}
After inserting the expression of $\alpha_{k+1}^{*[n]}\alpha^*_{\vec{\bf k}\setminus\vec{\gamma}}
\alpha_{\vec{\delta}}\tau^{[0]}$ into $\tau^{+[\vec{\bf n}, \overrightarrow{\bf k+1}]}$ and noting that $H_{k+1,a}$ can be divided into two groups: $H_{k,a}$ and $\{k+1\}\cup H_{k,a-1}$, one can obtain that (\ref{taunpkf2cases}) for $\tau^{+[\vec{\bf n}, \vec{\bf k}]}$ is also correct for $k+1$ with the recursion relation (\ref{crecursion}).

The results of $\tau^{-[\vec{\bf n}, \vec{\bf k}]}$ can be proved in the similar method.
\end{proof}
\noindent{\bf Remark}: Here we would like to point out that the structures of $\tau^{\pm[\vec{\bf n},\vec{\bf k}]}$ should be connected with the product of $\beta_n\cdots\beta_1\beta_k^*\cdots\beta_1^*$ for $\beta_i\in V$ and $\beta_j^*\in V^*$, that is,
\begin{align}
\beta_n\cdots\beta_1\beta_k^*\cdots\beta_1^*
=&\sum_{a=0}^n\sum_{\vec{\gamma}\in H_{k,a}}\sum_{\eta\in S_n}(-1)^{sgn\eta}[\beta_{\eta(a)},\beta^*_{\gamma_a}]_+\cdots
[\beta_{\eta(1)},\beta^*_{\gamma_1}]_+\nonumber\\
&\beta^*_{\vec{n}\setminus\vec{\gamma}}\beta_{\eta(n)}\cdots
\beta_{\eta(a+1)}(-1)^{ak-|\gamma|+(n-a)k}/(n-a)!.
\end{align}
Here $S_n$ is the $n$-th permutation group.
\begin{corollary}\label{coralpdown}
For any $\vec{\gamma}\in H_{k,a}$ and $\vec{\delta}\in H_{n,b}$, one has the following relations
\begin{align*}
\alpha_{\vec{\gamma}}^{*+[\vec{\bf n}]}\tau_{\vec{\alpha}_{\vec{\delta}}}
=\alpha_{\vec{\gamma}}^{*+[\vec{\delta}]}
\tau_{\vec{\alpha}_{\vec{\delta}}},\quad  \alpha_{\vec{\gamma}}^{-[\vec{\bf n}]}\tau_{\vec{\alpha}^*_{\vec{\delta}}}
=\alpha_{\vec{\gamma}}^{-[\vec{\delta}]}
\tau_{\vec{\alpha}^*_{\vec{\delta}}}.
\end{align*}
\end{corollary}
\begin{proof}
We firstly prove the case $\vec{\gamma}=\{j\}\in H_{k,1}$.
Assume $\vec{\delta}=\vec{\bf b}=(b,b-1,\cdots,1)$, then by using (\ref{alphastarn}),
\begin{align*}
\alpha_j^{*+[\vec{\bf n}]}\tau_{\alpha_{\vec{\bf b}}}=
\alpha_j^*\alpha_{\vec{\bf b}}\tau^{[0]}+\sum_{l=1}^n
\frac{\Psi_j^{[l]}(t')}{\tau^{[l-1]}(t')}\left(\langle l-1|e^{H(t')}\otimes 1\right)\cdot S(\alpha_{\overrightarrow{l-1}}\tau^{[0]}\otimes \alpha_{\vec{\bf b}}\tau^{[0]}).
\end{align*}
Note that $S(\alpha_{\overrightarrow{l-1}}\tau^{[0]}\otimes \alpha_{\vec{\bf b}}\tau^{[0]})=0$ for $l\geq b+1$. So the sum $\sum_{l=1}^n$ will become into $\sum_{l=1}^b$, which can lead to $\alpha_j^{*+[\vec{\bf n}]}\tau_{\vec{\alpha}_{\vec{\bf b}}}=\alpha_j^{*+[\vec{\bf b}]}
\tau_{\vec{\alpha}_{\vec{\bf b}}}$. Therefore for general $\vec{\delta}\in H_{n,b}$, one has
\begin{align*}
\alpha_{j}^{*+[\vec{\bf n}]}\tau_{\vec{\alpha}_{\vec{\bf n}\setminus\vec{\delta}}}=\alpha_{j}^{*+[(\vec{\bf n}\setminus\vec{\delta})\cup\vec{\delta}]}
\tau_{\vec{\alpha}_{\vec{\delta}}}
=\alpha_{j}^{*+[\vec{\delta}]}
\tau_{\vec{\alpha}_{\vec{\delta}}}.
\end{align*}
Here we have used $\alpha_{j}^{*+[\vec{\bf n}]}=\alpha_{j}^{*+[(\vec{\bf n}\setminus\vec{\delta})\cup\vec{\delta}]}$ derived from $\rho_{j}^{*+[\vec{\bf n}]}=\rho_{j}^{*+[(\vec{\bf n}\setminus\vec{\delta})\cup\vec{\delta}]}$ by Lemma \ref{commkpdt} and $\tau^{+[\vec{\bf n}]}\approx\tau^{+[(\vec{\bf n}\setminus\vec{\delta})\cup\vec{\delta}]}$. Next assume we have proved $\alpha_{\vec{\gamma}}^{*+[\vec{\bf n}]}\tau_{\vec{\alpha}_{\vec{\delta}}}
=\alpha_{\vec{\gamma}}^{*+[\vec{\delta}]}
\tau_{\vec{\alpha}_{\vec{\delta}}}$ for any $\vec{\gamma}\in H_{k,a}$. Then if $\vec{\gamma}=(\gamma_{a+1},\gamma_a,\cdots,\gamma_1)\in H_{k,a+1}$,
\begin{align*}
\alpha_{\vec{\gamma}}^{*+[\vec{\bf n}]}\tau_{\vec{\alpha}_{\vec{\delta}}}
=&\alpha_{\gamma_{a+1}}^{*+[\vec{\bf n}]}\alpha_{\vec{\gamma}\setminus\{\gamma_{a+1}\}}
^{*+[\vec{\bf n}]}\tau_{\vec{\alpha}_{\vec{\delta}}}
=\alpha_{\gamma_{a+1}}^{*+[\vec{\bf n}]}\alpha_{\vec{\gamma}\setminus\{\gamma_{a+1}\}}^{*+[\vec{\delta}]}
\tau_{\vec{\alpha}_{\vec{\delta}}}=(-1)^a\alpha_{\vec{\gamma}\setminus\{\gamma_{a+1}\}}^{*+[\vec{\delta}]}
\alpha_{\gamma_{a+1}}^{*+[\vec{\bf n}]}\tau_{\vec{\alpha}_{\vec{\delta}}}\\
=&(-1)^a\alpha_{\vec{\gamma}\setminus\{\gamma_{a+1}\}}^{*+[\vec{\delta}]}
\alpha_{\gamma_{a+1}}^{*+[\vec{\delta}]}\tau_{\vec{\alpha}_{\vec{\delta}}}=\alpha_{\vec{\gamma}}^{*+[\vec{\delta}]}
\tau_{\vec{\alpha}_{\vec{\delta}}}
\end{align*}

Similarly, one can prove $ \alpha_{\vec{\gamma}}^{-[\vec{\bf n}]}\tau_{\vec{\alpha}^*_{\vec{\delta}}}
=\alpha_{\vec{\gamma}}^{-[\vec{\delta}]}
\tau_{\vec{\alpha}^*_{\vec{\delta}}}$.
\end{proof}
\begin{corollary}\label{coralpup}
$\alpha_{\vec{\bf k}}^{*+[\vec{\bf n}]}{\alpha}_{\overrightarrow{\bf n_1+n}}$ and $\alpha_{\vec{\bf n}}^{-[\vec{\bf k}]}{\alpha}^*_{\overrightarrow{\bf k_1+k}}$ act on $\tau^{[0]}$ in the way below
\begin{align*}
&\alpha_{\vec{\bf k}}^{*+[\vec{\bf n}]}{\alpha}_{\overrightarrow{\bf n_1+n}}\tau^{[0]}
=\sum_{a=0}^k\sum_{\vec{\gamma}\in H_{k,a}}\sum_{\vec{\delta}\in n+H_{n_1,n_1-a}}\mathcal{C}^{+[k,n_1,n]}_{a,\vec{\gamma},\vec{\delta}}\alpha
^{*+[\vec{\delta}\cup\vec{\bf n}]}_{\vec{\bf k}\setminus\vec{\gamma}}\alpha_{\vec{\delta}\cup \vec{\bf n}}
\tau^{[0]},\\
&\alpha_{\vec{\bf n}}^{-[\vec{\bf k}]}{\alpha}^*_{\overrightarrow{\bf k_1+k}}\tau^{[0]}
=\sum_{a=0}^n\sum_{\vec{\gamma}\in H_{n,a}}\sum_{\vec{\delta}\in k+H_{k_1,k_1-a}}\mathcal{C}^{-[n,k_1,k]}_{a,\vec{\gamma},\vec{\delta}}\alpha
^{-[\vec{\delta}\cup\vec{\bf k}]}_{\vec{\bf n}\setminus\vec{\gamma}}\alpha^*_{\vec{\delta}\cup \vec{\bf k}}
\tau^{[0]}.
\end{align*}
Here $\mathcal{C}^{+[k,n_1,n]}_{a,\vec{\gamma},\vec{\delta}}$ and $\mathcal{C}^{-[n,k_1,k]}_{a,\vec{\gamma},\vec{\delta}}$ are some constant independent of $t$, which can be fixed by with $\mathcal{C}^{+[k,n_1,n]}_{0,0,\overrightarrow{\bf n_1}}=\mathcal{C}^{-[n,k_1,k]}_{0,0,\overrightarrow{\bf k_1}}=1$ and the recursion relations below
\begin{align*}
\mathcal{C}^{+[k+1,n_1,n]}_{a,\vec{\gamma},\vec{\delta}}
=\mathcal{C}^{+[k,n_1,n]}_{a,\vec{\gamma},\vec{\delta}},\quad\mathcal{C}^{+[k+1,n_1,n]}_{a+1,\{k+1\}\cup\vec{\gamma},\vec{\delta}\setminus\{\delta_l\}}
=\mathcal{C}^{+[k,n_1,n]}_{a,\vec{\gamma},\vec{\delta}}\sum_{i=1}^l(-1)^{n_1+k-l-1}
\Psi_{k+1}^{+[\vec{\delta}_{\vec{\bf i}}\cup\vec{\bf n}]}(t')\frac{\tau^{+[\delta_l\cup\vec{\delta}_{\overrightarrow{\bf i-1}}\cup \vec{\bf n}]}(t')}{\tau^{+[\vec{\delta}_{\overrightarrow{\bf i-1}}\cup\vec{\bf n}]}(t')},\\
\mathcal{C}^{-[n+1,k_1,k]}_{a,\vec{\xi},\vec{\eta}}
=\mathcal{C}^{-[n,k_1,k]}_{a,\vec{\xi},\vec{\eta}},\quad\mathcal{C}^{-[n+1,k_1,k]}_{a+1,\{n+1\}\cup\vec{\xi},\vec{\eta}\setminus\{\eta_l\}}
=\mathcal{C}^{-[n,k_1,k]}_{a,\vec{\xi},\vec{\eta}}\sum_{i=1}^l(-1)^{k_1+n-l-1}
\Phi_{n+1}^{-[\vec{\eta}_{\vec{\bf i}}\cup\vec{\bf n}]}(t')\frac{\tau^{-[\eta_l\cup\vec{\eta}_{\overrightarrow{\bf i-1}}\cup \vec{\bf n}]}(t')}{\tau^{-[\vec{\eta}_{\overrightarrow{\bf i-1}}\cup\vec{\bf n}]}(t')},
\end{align*}
with $\vec{\gamma}\in H_{k,a}$, $\vec{\delta}\in n+H_{n_1,n_1-a}$, $\vec{\xi}\in H_{n,a}$, $\vec{\eta}\in k+H_{k_1,k_1-a}$.
\end{corollary}
\begin{proof}
Here we only prove the first one, since the second is almost the same. Firstly, the result of $\alpha_{\vec{\bf k}}^{*+[\vec{\bf n}]}{\alpha}_{\overrightarrow{\bf n_1+n}}\tau^{[0]}$ for $k=0$ is correct. So we can assume it holds for $k$. Then by the similar way in Proposition \ref{kptdtitau} for $\vec{\delta}\in n+H_{n_1,n_1-a}$,
\begin{align*}
&\rho_{k+1}^{*+[\vec{\delta}\cup\vec{\bf n}]}=
\frac{1}{\lambda^{n_1-a}}\rho_{k+1}^{*+[\vec{\bf n}]}+\sum_{i=1}^{n_1-a}\frac{1}{\lambda^{n_1-a+1-i}}\Psi_{k+1}^{+[\vec{\delta}_{\vec{\bf i}}\cup\vec{\bf n}]}(t')\psi^{+[\vec{\delta}_{\overrightarrow{\bf i-1}}\cup\vec{\bf n}]}(t',\lambda),\\
&\alpha_{k+1}^{*+[\vec{\delta}\cup\vec{\bf n}]}=
\alpha_{k+1}^{*+[\vec{\bf n}]}+\sum_{l\in\mathbb{Z}}\sum_{i=1}^{n_1-a}
\frac{\Psi_{k+1}^{+[\vec{\delta}_{\vec{\bf i}}\cup\vec{\bf n}]}(t')}{\tau_{\vec{\alpha}_{\delta_{\overrightarrow{\bf i-1}}}\cup\alpha_{\vec{\bf n}}}(t')}\sigma_{t'}(\psi_l\alpha_{\delta_{\overrightarrow{\bf i-1}}}\cup\alpha_{\vec{\bf n}}\tau^{[0]})\psi_l^*,
\end{align*}
and
\begin{align*}
\alpha_{k+1}^{*+[\vec{\bf n}]}\alpha_{\vec{\delta}\cup \vec{\bf n}}\tau^{[0]}=\alpha_{k+1}^{*+[\vec{\delta}\cup\vec{\bf n}]}\alpha_{\vec{\delta}\cup \vec{\bf n}}\tau^{[0]}+\sum_{l=1}^{n_1-a}\tilde{C}_{l}\alpha
_{\vec{\delta}\setminus\{\delta_l\}}\alpha_{\vec{\bf n}}\tau^{[0]},
\end{align*}
with $\tilde{C}_l=\sum_{i=1}^l(-1)^{n_1-a-l-1}\Psi_{k+1}^{+[\delta_{\vec{\bf i}}\cup\vec{\bf n}]}(t')\frac{\tau_{\vec{\alpha}_{\delta_l}\cup\vec{\alpha}_{\delta_{\overrightarrow{\bf i-1}}}\cup \vec{\alpha}_{\vec{\bf n}}}(t')}{\tau_{\vec{\alpha}_{\delta_{\overrightarrow{\bf i-1}}}\cup\vec{\alpha}_{\vec{\bf n}}}(t')}$ and $\vec{\delta}=(\delta_{n_1-a},\delta_{n_1-a-1},\cdots,\delta_1)\in n+H_{n_1,n_1-a}$. Therefore
\begin{align*}
&\alpha_{\overrightarrow{\bf k+1}}^{*+[\vec{\bf n}]}{\alpha}_{\overrightarrow{\bf n_1+n}}\tau^{[0]}
=\sum_{a=0}^k\sum_{\vec{\gamma}\in H_{k,a}}\sum_{\vec{\delta}\in n+H_{n_1,n_1-a}}\mathcal{C}^{+[k,n_1,n]}_{a,\vec{\gamma},\vec{\delta}}\alpha_{k+1}^{*+[\vec{\bf n}]}\alpha
^{*+[\vec{\delta}\cup\vec{\bf n}]}_{\vec{\bf k}\setminus\vec{\gamma}}\alpha_{\vec{\delta}\cup \vec{\bf n}}
\tau^{[0]}\nonumber\\
=&\sum_{a,\vec{\gamma},\vec{\delta}}\mathcal{C}^{+[k,n_1,n]}_{a,\vec{\gamma},\vec{\delta}}
\Big(\alpha_{k+1}^{*+[\vec{\delta}\cup\vec{\bf n}]}\alpha
^{*+[\vec{\delta}\cup\vec{\bf n}]}_{\vec{\bf k}\setminus\vec{\gamma}}\alpha_{\vec{\delta}\cup \vec{\bf n}}+(-1)^{k-a}\sum_{l=1}^{n_1-a}\tilde{C}_{l}\alpha
^{*+[(\vec{\delta}\setminus{\{\delta_l\}})\cup\vec{\bf n}]}_{\vec{\bf k}\setminus\vec{\gamma}}\alpha
_{\vec{\delta}\setminus\{\delta_l\}}\alpha_{\vec{\bf n}}\Big)\tau^{[0]},
\end{align*}
which is just the result for $k+1$.
Here we have used $\alpha
^{*+[\vec{\delta}\cup\vec{\bf n}]}_{\vec{\bf k}\setminus\vec{\gamma}}\alpha
_{\vec{\delta}\setminus\{\delta_l\}}\alpha_{\vec{\bf n}}\tau^{[0]}=\alpha
^{*+[(\vec{\delta}\setminus{\{\delta_l\}})\cup\vec{\bf n}]}_{\vec{\bf k}\setminus\vec{\gamma}}\alpha
_{\vec{\delta}\setminus\{\delta_l\}}\alpha_{\vec{\bf n}}\tau^{[0]}$ derived by Corollary \ref{coralpdown}.
\end{proof}
\noindent {\bf Remark:} Denote $\mathcal{O}_{n,k}^{\pm}$ as the linear combination of the transformed tau functions $\tau^{\pm[\overrightarrow{\bf n_1},\overrightarrow{\bf k_1}]}$ for $n_1\leq n$ and $k_1\leq k$, under no more than $n$-step $T_d$ and no more than $k$-step $T_i$ with the the generating eigenfunctions and adjoint eigenfunctions in $T^{\pm[\vec{\bf n},\vec{\bf k}]}$, then
\begin{align*}
&\alpha_{\vec{\bf k}}^{*+[\vec{\bf n}]}{\alpha}_{\overrightarrow{\bf n_1+n}}\tau^{[0]}
=\alpha_{\vec{\bf k}}^{*+[\overrightarrow{\bf n+n_1}]}{\alpha}_{\overrightarrow{\bf n_1+n}}\tau^{[0]}+\mathcal{O}_{n_1+n-1,k-1}^+,\\
&\alpha_{\vec{\bf n}}^{-[\vec{\bf k}]}{\alpha}^*_{\overrightarrow{\bf k_1+k}}\tau^{[0]}
=\alpha_{\vec{\bf n}}^{-[\overrightarrow{\bf k+k_1}]}{\alpha}^*_{\overrightarrow{\bf k_1+k}}\tau^{[0]}+\mathcal{O}_{n-1,k+k_1-1}^-.
\end{align*}
Further it is important to note that by Lemma \ref{s4beta},
\begin{align}
S(\mathcal{O}_{n',k'}^\pm\otimes\mathcal{O}_{n,k}^\pm)
=\mathcal{O}_{n',k'-1}^\pm\otimes\mathcal{O}_{n,k+1}^\pm
+\mathcal{O}_{n'+1,k'}^\pm\otimes\mathcal{O}_{n-1,k}^\pm.\label{sootimeso}
\end{align}
\subsection{The bilinear equations in the Darboux chains of the KP hierarchy}\label{subsectbilinkp}
After the preparation above, now we can consider the generalization of the results about $n$-step $T_d$ showed in the bilinear equations (\ref{bilinnmkp}) of the $(l-l')$-th modified KP hierarchy. Firstly by Lemma \ref{slbetaotimesbeta}, Corollary \ref{coralpdown} and Corollary \ref{coralpup}, one can easily obtain the following theorem.
\begin{theorem}\label{stnnkkotnk}
Given $n'\geq n\geq 0$, $k'\geq k\geq 0$,
\begin{align}
S^l\Big(\tau^{+[\overrightarrow{\bf n'},\overrightarrow{\bf k'}]}\otimes\tau^{+[\vec{\bf n},\vec{\bf k}]}\Big)=l!\sum_{\vec{\gamma}\in k+H_{k'-k,l}}(-1)^{lk'-|\vec{\gamma}|}\tau^{+[\overrightarrow{\bf n'},\overrightarrow{\bf k'}\setminus\vec{\gamma}]}\otimes\tau^{+[\vec{\bf n},\vec{\gamma}\cup\vec{\bf k}]},\\
S^l\Big(\tau^{-[\vec{\bf n},\vec{\bf k}]}\otimes\tau^{-[\overrightarrow{\bf n'},\overrightarrow{\bf k'}]}\Big)=l!\sum_{\vec{\gamma}\in n+H_{n'-n,l}}(-1)^{ln'-|\vec{\gamma}|}\tau^{-[\vec{\gamma}\cup\vec{\bf n},\vec{\bf k}]}\otimes\tau^{-[\overrightarrow{\bf n'}\setminus\vec{\gamma},\overrightarrow{\bf k'}]}.
\end{align}
In particular,
\begin{align}
&S^{k'-k}\Big(\tau^{+[\overrightarrow{\bf n'},\overrightarrow{\bf k'}]}\otimes\tau^{+[\vec{\bf n},\vec{\bf k}]}\Big)=(-1)^{\frac{(k'-k)(k'-k-1)}{2}}(k'-k)!\tau^{+[\overrightarrow{\bf n'},\vec{\bf k}]}\otimes\tau^{+[\vec{\bf n},\overrightarrow{\bf k'}]},\label{skk}\\
&S^{n'-n}\Big(\tau^{-[\vec{\bf n},\vec{\bf k}]}\otimes\tau^{-[\overrightarrow{\bf n'},\overrightarrow{\bf k'}]}\Big)=(-1)^{\frac{(n'-n)(n'-n-1)}{2}}(n'-n)!\tau^{-[\overrightarrow{\bf n'},\vec{\bf k}]}\otimes\tau^{-[\vec{\bf n},\overrightarrow{\bf k'}]},\label{snn}\\
&S^{k'-k+1}\Big(\tau^{+[\overrightarrow{\bf n'},\overrightarrow{\bf k'}]}\otimes\tau^{+[\vec{\bf n},\vec{\bf k}]}\Big)=S^{n-n'+1}\Big(\tau^{-[\vec{\bf n},\vec{\bf k}]}\otimes\tau^{-[\overrightarrow{\bf n'},\overrightarrow{\bf k'}]}\Big)=0\label{skk+1}.
\end{align}
\end{theorem}
\noindent{\bf Remark:} Here we do not discuss $S^l\Big(\tau^{-[\overrightarrow{\bf n'},\overrightarrow{\bf k'}]}\otimes\tau^{-[\vec{\bf n},\vec{\bf k}]}\Big)$ and $S^l\Big(\tau^{+[\vec{\bf n},\vec{\bf k}]}\otimes\tau^{+[\overrightarrow{\bf n'},\overrightarrow{\bf k'}]}\Big)$. Actually according to Corollary \ref{coralpup}, there will be many additional terms in them, just as Theorem \ref{stnnkotnkk} below.\\
\noindent{\bf Remark:} (\ref{skk}) and (\ref{snn}) are the generalization of (\ref{klcmkp}). What's more,
we have
\begin{align*}
S\Big(\tau\otimes\tau_{\alpha_{\vec{\bf n}}}\Big)=
\sum_{l=1}^n(-1)^{n-l}\tau_{\alpha_l}\otimes \tau_{\alpha_{\vec{\bf n}}\setminus \alpha_l},\quad S(\tau_{\alpha^*_{\vec{\bf n}}}\otimes\tau)=
\sum_{l=1}^n(-1)^{n-l}\tau_{\alpha^*_{\vec{\bf n}}\setminus \alpha^*_l}\otimes \tau_{\alpha_l^*},
\end{align*}
which implies that
\begin{align*}
&{\rm Res}_\lambda\lambda^{-n}\tau(t'-\varepsilon(\lambda^{-1}))
\tau_{\alpha_{\vec{\bf n}}}(t+\varepsilon(\lambda^{-1}))=
\sum_{l=1}^n(-1)^{n-l}\tau_{\alpha_l}(t')\tau_{\alpha_{\vec{\bf n}}\backslash \alpha_l}(t),\nonumber\\
&{\rm Res}_\lambda\lambda^{-n}\tau_{\alpha^*_{\vec{\bf n}}}(t'-\varepsilon(\lambda^{-1}))
\tau(t+\varepsilon(\lambda^{-1}))=
\sum_{l=1}^n(-1)^{n-l}\tau_{\alpha^*_{\vec{\bf n}}\backslash \alpha^*_l}(t')\tau_{\alpha_l^*}(t).
\end{align*}
In particularly when $n=1$, the relations above are just the bilinear equations (\ref{mkptaubilinear}) of the modified KP hierarchy of Kupershmidt-Kiso version.
\begin{corollary}
$\tau^{\pm[\vec{\bf n}]}$ satisfies the following bilinear equations
\begin{align}
S(\tau^{+[\overrightarrow{\bf n'}]}\otimes \tau^{+[\vec{\bf n}]})=0,\quad S(\tau^{-[\vec{\bf n}]}\otimes \tau^{-[\overrightarrow{\bf n'}]})=0,\quad n'\geq n.
\end{align}
The corresponding Bosonic forms are
\begin{align}
&{\rm Res}_{\lambda}\lambda^{n'-n}\tau^{+[\overrightarrow{\bf n'}]}(t-\varepsilon(\lambda^{-1}))
\tau^{+[\vec{\bf n}]}(t'+\varepsilon(\lambda^{-1}))e^{\xi(t-t',\lambda)}=0,
\label{staupntaupn}\\
&{\rm Res}_{\lambda}\lambda^{n'-n}\tau^{-[\vec{\bf n}]}(t-\varepsilon(\lambda^{-1}))
\tau^{-[\overrightarrow{\bf n'}]}(t'+\varepsilon(\lambda^{-1}))e^{\xi(t-t',\lambda)}=0, \quad n'\geq n.\label{staumntaumn}
\end{align}
\end{corollary}

\noindent {\bf Remark:} Note that relations (\ref{staumntaumn}) describe the transformed tau functions in the chain of the Darboux transformation $T_i$. In fact, (\ref{staumntaumn}) is also the bilinear equations of $(n'-n)$-th modified KP hierarchy, just by setting $\tau_{-n}=\tau^{-[\vec{\bf n}]}$. What's more, (\ref{staupntaupn}) and (\ref{staumntaumn}) constitute the whole discrete KP hierarchy \cite{Adler1999,Haine2000,Dickey1999}, since the discrete variables in (\ref{staupntaupn}) only takes the non-negative values, while the ones (\ref{staumntaumn}) are all non-positive.

Now we have discussed the bilinear equations between the transformed tau functions under $n'$-step $T_d$ and $k'$-step $T_i$, and the ones under $n$-step $T_d$ and $k$-step $T_i$ for $n'\geq n$, $k'\geq k$. In fact, these relations are not the whole, there should be other bilinear equations involving $\Big(\tau^{\pm[\overrightarrow{\bf n'},\vec{\bf k}]},\tau^{\pm[\vec{\bf n},\overrightarrow{\bf k'}]}\Big)$ and $\Big(\tau^{\pm[\vec{\bf n},\overrightarrow{\bf k'}]},\tau^{\pm[\overrightarrow{\bf n'},\vec{\bf k}]}\Big)$ with $n'\geq n$, $k'\geq k$, which are given in the theorem below.
\begin{theorem}\label{stnnkotnkk}
For $n'\geq n$, $k'\geq k$,
\begin{align}
S\Big(\tau^{\pm[\overrightarrow{\bf n'},\vec{\bf k}]}\otimes\tau^{\pm[\vec{\bf n},\overrightarrow{\bf k'}]}\Big)=&0,\\
S^l(\tau^{\pm[\vec{\bf n},\overrightarrow{\bf k'}]}\otimes\tau^{\pm[\overrightarrow{\bf n'},\vec{\bf k}]})
=&l!\sum_{a=0}^l\sum_{\vec{\gamma}\in k+H_{k'-k,l-a}}
\sum_{\vec{\delta}\in n+H_{n'-n,a}}(-1)^{a(n'-k)+lk'+\frac{1\mp1}{2}\Big((k-k')a+(n'-n)(l-a)\Big)}\nonumber\\
&\times(-1)^{|\vec{\gamma}|+|\vec{\delta}|}\Big(\tau^{\pm[\vec{\delta}\cup\vec{\bf n},\overrightarrow{\bf k'}\setminus\vec{\gamma}]}\otimes\tau^{\pm[\overrightarrow{\bf n'}\setminus\vec{\delta},\vec{\gamma}\cup\vec{\bf k}]}+\sum_{\mu=\pm}M^\pm_{a+n,k'-l+a,n'-a,l-a+k;\mu}\Big).
\end{align}
Particularly,
\begin{align}
S^{k'-k+n'-n}&(\tau^{\pm[\vec{\bf n},\overrightarrow{\bf k'}]}\otimes\tau^{\pm[\overrightarrow{\bf n'},\vec{\bf k}]})
=(-1)^{(n'-n)(k'-k)+\frac{(n'-n)(n'-n-1)}{2}
+\frac{(k'-k)(k'-k-1)}{2}}\nonumber\\
&\times(n'-n+k'-k)!\Big(\tau^{\pm[\overrightarrow{\bf n'},\vec{\bf k}]}\otimes\tau^{\pm[\vec{\bf n},\overrightarrow{\bf k'}]}+M_{n' knk';\pm}^{\pm}
\Big),\label{skknntt}\\
S^{k'-k+n'-n+1}&(\tau^{\pm[\vec{\bf n},\overrightarrow{\bf k'}]}\otimes\tau^{\pm[\overrightarrow{\bf n'},\vec{\bf k}]})
=0,\label{skknn1tt}
\end{align}
where $M_{n' knk';+}^\pm=\mathcal{O}_{n'-1,k-1}^{\pm}\otimes\mathcal{O}_{n,k'}^{\pm}$ and $M_{n' knk';-}^\pm=\mathcal{O}_{n',k}^{\pm}\otimes\mathcal{O}_{n-1,k'-1}^{\pm}$.
\end{theorem}
\begin{proof}
In fact, this theorem can be proved by using Lemma \ref{slbetaotimesbeta} and Corollary \ref{coralpup}. Here we should point out that there is no $\mathcal{O}_{n',k}^+\otimes\mathcal{O}_{n-1,k'-1}^+$-term in $S^{k'-k+n'-n}(\tau^{+[\vec{\bf n},\overrightarrow{\bf k'}]}\otimes\tau^{+[\overrightarrow{\bf n'},\vec{\bf k}]})$, since it is proportional to $\alpha_{\vec{\bf k}}^{*[\vec{\bf n}]}\alpha_{\overrightarrow{\bf n'}}\tau^{[0]}\otimes\alpha_{\overrightarrow{\bf k'}}^{*[\vec{\bf n}]}\alpha_{\vec{\bf n}}\tau^{[0]}$. Similar discussion can be done for $S^{k'-k+n'-n}(\tau^{-[\vec{\bf n},\overrightarrow{\bf k'}]}\otimes\tau^{-[\overrightarrow{\bf n'},\vec{\bf k}]})$ and $S^{k'-k+n'-n+1}(\tau^{\pm[\vec{\bf n},\overrightarrow{\bf k'}]}\otimes\tau^{\pm[\overrightarrow{\bf n'},\vec{\bf k}]})$.
\end{proof}

\noindent{\bf Remark:} Note that the terms $M_{n' knk';+}^+$ in (\ref{skknntt}) are usually the linear combination of the tau functions with some constants as the coefficients. If we set these constants to be zero, then these terms vanish. In fact, this can be realized. For example, we can consider the Darboux transformations generated by the (adjoint) wave functions, instead of the (adjoint) eigenfunctions.

Next let us consider the Darboux chain generated by the wave functions.
\begin{eqnarray*}
&&L\xrightarrow{T_d(\psi(t,\lambda_1))}L^{[1]}
\xrightarrow{T_d(\psi(t,\lambda_2)^{[1]})}L^{[2]}
  \rightarrow\cdots\rightarrow L^{[n-1]}\xrightarrow{T_d(\psi(t,\lambda_n)^{[n-1]})}L^{[n]}\nonumber\\
&&\xrightarrow{T_i(\psi^*(t,\mu_1)^{[n+1]})}L^{(n+1)}\xrightarrow
{T_i(\psi^*(t,\mu_2)^{[n+2]})}\cdots\rightarrow L^{[n+k-1]}\xrightarrow{T_i(\psi^*(t,\mu_k)^{[n+k-1]})}L^{[n+k]}.
\end{eqnarray*}
Denote $T_{\lambda\mu}^{[\vec{\bf n},\vec{\bf k}]}=T_i(\psi^*(t,\mu_k)^{[n+k-1]})\cdots T_i(\psi^*(t,\mu_1)^{[n]})T_d(\psi(t,\lambda_n)^{[n-1]})\cdots T_d(\psi(t,\lambda_1))$ as the successive applications of $n$-step $T_d$ and $k$-step $T_i$, and $\tau_{\lambda\mu}^{[\vec{\bf n},\vec{\bf k}]}$ as the transformed tau functions under $T_{\lambda\mu}^{[\vec{\bf n},\vec{\bf k}]}$.

By the similar method as before, one can obtain
\begin{align}
\tau_{\lambda\mu}^{[\vec{\bf n},\vec{\bf k}]}=\left(\prod_{j=2}^n\lambda_j^{-j+1}\prod_{l=1}^k\mu_{l}^{n-l}\right)
\psi_{\lambda_{\vec{\bf n}};\mu_{\vec{\bf k}}}\tau^{[0]},
\end{align}
where $\psi_{\lambda_{\vec{\bf n}};\mu_{\vec{\bf k}}}
\triangleq\psi^*(\mu_{k})\cdots\psi^*(\mu_{1})
\psi(\lambda_{n})\cdots\psi(\lambda_{1})$. Note that $\tau_{\lambda\mu}^{[\vec{\bf n},\vec{\bf k}]}\approx\psi_{\lambda_{\vec{\bf n}};\mu_{\vec{\bf k}}}\tau^{[0]}$, i.e.,they determine the same (adjoint) wave functions, which means that the corresponding $T_d$ and $T_i$ are the same in the Darboux chain above. So in what follows, we can always use $\tau_{\lambda\mu}^{[\vec{\bf n},\vec{\bf k}]}=\psi_{\lambda_{\vec{\bf n}};\mu_{\vec{\bf k}}}\tau^{[0]}$ instead.
\begin{proposition}
For $n'\geq n$, $k'\geq k$,
\begin{align*}
S^l\Big(\tau_{\lambda\mu}^{[\overrightarrow{\bf n'},\overrightarrow{\bf k'}]}\otimes\tau_{\lambda\mu}^{[\vec{\bf n},\vec{\bf k}]}\Big)=&l!\sum_{\vec{\gamma}\in k+H_{k'-k,l}}(-1)^{lk'-|\vec{\gamma}|}\tau_{\lambda\mu}^{[\overrightarrow{\bf n'},\overrightarrow{\bf k'}\setminus\vec{\gamma}]}\otimes\tau_{\lambda\mu}^{[\vec{\bf n},\vec{\gamma}\cup\vec{\bf k}]},\\
S^l\Big(\tau_{\lambda\mu}^{[\vec{\bf n},\vec{\bf k}]}\otimes\tau_{\lambda\mu}^{[\overrightarrow{\bf n'},\overrightarrow{\bf k'}]}\Big)=&l!\sum_{\vec{\gamma}\in n+H_{n'-n,l}}(-1)^{l(k'-k+n')-|\vec{\gamma}|}\tau_{\lambda\mu}^{[\vec{\gamma}\cup\vec{\bf n},\vec{\bf k}]}\otimes\tau_{\lambda\mu}^{[\overrightarrow{\bf n'}\setminus\vec{\gamma},\overrightarrow{\bf k'}]},\\
S\Big(\tau_{\lambda\mu}^{[\overrightarrow{\bf n'},\vec{\bf k}]}\otimes\tau_{\lambda\mu}^{[\vec{\bf n},\overrightarrow{\bf k'}]}\Big)=&0,\\
S^l\Big(\tau_{\lambda\mu}^{[\vec{\bf n},\overrightarrow{\bf k'}]}\otimes\tau_{\lambda\mu}^{[\overrightarrow{\bf n'},\vec{\bf k}]}\Big)
=&l!\sum_{a=0}^l\sum_{\vec{\gamma}\in k+H_{k'-k,l-a}}
\sum_{\vec{\delta}\in n+H_{n'-n,a}}(-1)^{a(n'-k)+lk'-|\vec{\gamma}|-|\vec{\delta}|}\Big(\tau_{\lambda\mu}^{[\vec{\delta}\cup\vec{\bf n},\overrightarrow{\bf k'}\setminus\vec{\gamma}]}\otimes\tau_{\lambda\mu}^{[\overrightarrow{\bf n'}\setminus\vec{\delta},\vec{\gamma}\cup\vec{\bf k}]}\Big).
\end{align*}
Particularly,
\begin{align*}
S^{k'-k}\Big(\tau_{\lambda\mu}^{[\overrightarrow{\bf n'},\overrightarrow{\bf k'}]}\otimes\tau_{\lambda\mu}^{[\vec{\bf n},\vec{\bf k}]}\Big)=&(-1)^{\frac{(k'-k)(k'-k-1)}{2}}(k'-k)!
\tau_{\lambda\mu}^{[\overrightarrow{\bf n'},\vec{\bf k}]}\otimes\tau_{\lambda\mu}^{[\vec{\bf n},\overrightarrow{\bf k'}]},\\
S^{n'-n}\Big(\tau_{\lambda\mu}^{[\vec{\bf n},\vec{\bf k}]}\otimes\tau_{\lambda\mu}^{[\overrightarrow{\bf n'},\overrightarrow{\bf k'}]}\Big)=&(-1)^{(n'-n)(k'-k)+\frac{(n'-n)(n'-n-1)}{2}}(n'-n)!
\tau_{\lambda\mu}^{[\overrightarrow{\bf n'},\vec{\bf k}]}\otimes\tau_{\lambda\mu}^{[\vec{\bf n},\overrightarrow{\bf k'}]},\\
S^{n'-n+k'-k}\Big(\tau_{\lambda\mu}^{[\vec{\bf n},\overrightarrow{\bf k'}]}\otimes\tau_{\lambda\mu}^{[\overrightarrow{\bf n'},\vec{\bf k}]}\Big)
=&(-1)^{(n'-n)(k'-k)+\frac{(n'-n)(n'-n-1)}{2}+\frac{(k'-k)(k'-k-1)}{2}}\\
&\times(n'-n+k'-k)!\tau_{\lambda\mu}^{[\overrightarrow{\bf n'},\vec{\bf k}]}\otimes\tau_{\lambda\mu}^{[\vec{\bf n},\overrightarrow{\bf k'}]}.
\end{align*}
\end{proposition}
\noindent{\bf Remark:} Note that after several applications of the operator $S$, $\tau_{\lambda\mu}^{[\overrightarrow{\bf n'},\overrightarrow{\bf k'}]}\otimes\tau_{\lambda\mu}^{[\vec{\bf n},\vec{\bf k}]}$, $\tau_{\lambda\mu}^{[\vec{\bf n},\vec{\bf k}]}\otimes\tau_{\lambda\mu}^{[\overrightarrow{\bf n'},\overrightarrow{\bf k'}]}$ and $\tau_{\lambda\mu}^{[\vec{\bf n},\overrightarrow{\bf k'}]}\otimes\tau_{\lambda\mu}^{[\overrightarrow{\bf n'},\vec{\bf k}]}$ will become into $\tau_{\lambda\mu}^{[\overrightarrow{\bf n'},\vec{\bf k}]}\otimes\tau_{\lambda\mu}^{[\vec{\bf n},\overrightarrow{\bf k'}]}$ up to a multiplication of a constant. Different from the usual (adjoint) eigenfunctions, we can believe the corresponding fields $\psi(\lambda)$ and $\psi^*(\mu)$ for the wave function and the adjoint wave function respectively can anticommute with each other, since $[\psi(\lambda),\psi^*(\mu)]_+=\mu\delta(\lambda,\mu)$ can be viewed as zero when $\lambda\neq\mu$.
\subsection{Examples of the bilinear equations in the KP Darboux transformations}
\label{subsectkpdtexp}
In this section, we will give some examples of the bilinear equations. Firstly for $(n',k')=(1,1)$ and $(n,k)=(0,0)$, the various bilinear equations are given as follows,
\begin{align}
&S(\tau^{\pm[1,1]}\otimes\tau^{\pm[0,0]})=\pm\tau^{\pm[1,0]}\otimes\tau^{\pm[0,1]},\quad S(\tau^{\pm[0,0]}\otimes\tau^{\pm[1,1]})=\mp\tau^{\pm[1,0]}\otimes\tau^{\pm[0,1]},
\label{stau11tau00}\\
&S(\tau^{\pm[0,1]}\otimes\tau^{\pm[1,0]})=
\pm\tau^{\pm[0,0]}\otimes\tau^{\pm[1,1]}\mp
\tau^{\pm[1,1]}\otimes\tau^{\pm[0,0]},\label{stau01tau10}\\
&S(\tau^{\pm[1,0]}\otimes\tau^{\pm[0,1]})=0,\label{stau10tau01}
\end{align}
where in \eqref{stau01tau10} we have used Corollary \ref{coralpup}.
 The corresponding Bosonic forms are listed below
\begin{align}
{\rm Res}_\lambda \tau^{\pm[1,1]}\Big(t-\varepsilon(\lambda^{-1})\Big)&\tau^{\pm[0,0]}\Big(t'+\varepsilon(\lambda^{-1})\Big)
e^{\xi(t-t',\lambda)}=\pm\tau^{\pm[1,0]}(t)\tau^{\pm[0,1]}(t'),\label{restau11tau00}\\
{\rm Res}_\lambda \tau^{\pm[0,0]}\Big(t-\varepsilon(\lambda^{-1})\Big)&\tau^{\pm[1,1]}\Big(t'+\varepsilon(\lambda^{-1})\Big)
e^{\xi(t-t',\lambda)}
=\mp\tau^{\pm[1,0]}(t)\tau^{\pm[0,1]}(t'),\label{restau00tau11}\\
{\rm Res}_\lambda \tau^{\pm[0,1]}\Big(t-\varepsilon(\lambda^{-1})\Big)&\tau^{\pm[1,0]}\Big(t'+\varepsilon(\lambda^{-1})\Big)
\lambda^{-2}e^{\xi(t-t',\lambda)}
\nonumber\\
&=\pm\tau^{\pm[0,0]}(t)\tau^{\pm[1,1]}(t')\mp\tau^{\pm[1,1]}(t)\tau^{\pm[0,0]}(t'),\label{restau01tau10}\\
{\rm Res}_\lambda \tau^{\pm[1,0]}\Big(t-\varepsilon(\lambda^{-1})\Big)&\tau^{\pm[0,1]}\Big(t'+\varepsilon(\lambda^{-1})\Big)
\lambda^2e^{\xi(t-t',\lambda)}
=0.\label{restau10tau01}
\end{align}
It should be noted that \eqref{stau10tau01} or \eqref{restau10tau01} is the bilinear equation of $2$-nd modified KP hierarchy hierarchies, which is equivalent to the following equations\cite{Kac1998}
\begin{align*}
S^2(\tau^{\pm[0,1]}\otimes\tau^{\pm[1,0]})=
-2\tau^{\pm[1,0]}\otimes\tau^{\pm[0,1]},
\end{align*}
with the corresponding Bosonic form
\begin{align*}
{\rm Res}_{\lambda_1}{\rm Res}_{\lambda_2}&\Big(\lambda_1^{-1}-\lambda_2^{-1}\Big)^2 \tau^{\pm[0,1]}\Big(t-[\lambda_1^{-1}]-[\lambda_2^{-1}]\Big)
\tau^{\pm[1,0]}\Big(t'+[\lambda_1^{-1}]+[\lambda_2^{-1}]\Big)\\
&\times e^{\xi(t-t',\lambda_1)+\xi(t-t',\lambda_2)}
=-2\tau^{\pm[1,0]}(t)\tau^{\pm[0,1]}(t').
\end{align*}

Next we will try to understand \eqref{restau11tau00}, \eqref{restau00tau11} and  \eqref{restau01tau10}. Actually according to Table III and Proposition \ref{sepexpprop},
\begin{align*}
\tau^{\pm[1,1]}(t)=\mp \Omega(\Phi(t),\Psi(t))\tau(t),\quad \tau^{\pm [1,0]}(t)=\Phi(t)\tau(t),\quad \tau^{\pm [0,1]}(t)=\Psi(t)\tau(t),\quad \tau^{[0,0]}=\tau(t).
\end{align*}
Thus \eqref{restau11tau00}, \eqref{restau00tau11} and  \eqref{restau01tau10} will become into
\begin{align}
&{\rm Res}_\lambda \Omega\Big(\Phi(t-\varepsilon(\lambda^{-1})),\Psi(t-\varepsilon(\lambda^{-1}))\Big)
\psi(t,\lambda)\psi^*(t',\lambda)
=-\Phi(t)\Psi(t'),\label{resoppmpp}\\
&{\rm Res}_\lambda
\psi(t,\lambda)\psi^*(t',\lambda)
\Omega\Big(\Phi(t'+\varepsilon(\lambda^{-1})),\Psi(t'+\varepsilon(\lambda^{-1}))\Big)
=\Phi(t)\Psi(t').\label{smpphipsi}\\
&{\rm Res}_\lambda
\Omega\big(\psi(t,\lambda),\Psi(t)\big)\Omega\big(\Phi(t'),\psi^*(t',\lambda)\big)
=\Omega\big(\Phi(t'),\Psi(t')\big)-\Omega\big(\Phi(t),\Psi(t)\big).\label{resooomo}
\end{align}
Note that (\ref{resooomo}) is also obtained in \cite{Loris1997}, but the method is very complicated. If further set $t-t'=[\mu^{-1}]$ in the three relations above (\ref{resoppmpp})-(\ref{resooomo}), one can obtain
\begin{align}
&\Omega\Big(\Phi(t-\varepsilon(\lambda^{-1})),\Psi(t-\varepsilon(\lambda^{-1}))\Big)
-\Omega\Big(\Phi(t),\Psi(t)\Big)=-\frac{1}{\lambda}\Phi(t)
\Psi(t-\varepsilon(\lambda^{-1})),\label{stminusadd}
\end{align}
which is an important property also given in \cite{Aratyn1998,Willox1998}. Conversely by considering (\ref{kpwavebilinear}), the relations (\ref{resoppmpp})-(\ref{resooomo}) can also be derived from (\ref{stminusadd}). In what follows, we will point out that (\ref{stminusadd}) is very crucial in the proof of ASvM formula\cite{Adler1995,Dickey1995}, which connects the actions of the additional symmetries on the wave functions and the tau functions. Also (\ref{stminusadd}) is very useful in the derivations of the bilinear equations of the constrained KP hierarchy \cite{Cheng1994,Loris1997}.

We will end this section with the example of Theorem \ref{stnnkotnkk}. let $(n',k')=(2,2)$ and $(n,k)=(0,1)$, and denote $\tau^{[n,k]}=\tau^{+[\vec{\bf n},\vec{\bf k}]}$, then we will have
\begin{align}
S^3(\tau^{[0,2]}\otimes\tau^{[2,1]})=-6\tau^{[2,1]}\otimes\tau^{[0,2]}
+(c_1^+\tau^{[1,0]}+c_0^+\tau^{[0,0]})\otimes (c_2^-\tau^{[0,2]}+c_1^-\tau^{[0,1]}+c_0^-\tau^{[0,0]}),
\end{align}
where $c_i^\pm$ are some constants. If apply another $S$ on the above relation, one can find
\begin{align}
S^4(\tau^{[0,2]}\otimes\tau^{[2,1]})=0.
\end{align}

\section{The Darboux transformations of the modified KP hierarchy}
In this section, the Darboux transformations of the modified KP hierarchy are discussed in the Fermionic pictures. The transformed tau functions in the Fermionic forms under successive applications of the Darboux transformations are given. Then based upon these results, the bilinear equations in the Darboux chain are obtained. Also some examples are listed in the last subsection.
\subsection{Reviews on some facts of the Darboux transformations of the modified KP hierarchy}
In the modified KP hierarchy, there are three elementary Darboux transformation operators\cite{Oevelrmp1993,Shaw1997,Chengjnmp2018}, which are
\begin{align*}
T_1(\hat\Phi)\triangleq\hat\Phi^{-1},\quad
T_2(\hat\Phi)\triangleq\hat\Phi^{-1}_x\pa,\quad
T_3(\hat\Psi)\triangleq\pa^{-1}\hat\Psi_x,
\end{align*}
which can also be obtained by the Miura links \cite{Shaw1997} between the KP and modified KP hierarchies through $\mathcal{L}\xrightarrow{\text{Miura}}L\xrightarrow{\text{anti-Miura}}
\mathcal{L}^{[1]}$.
Here the meaning of $A^{[n]}$ is the same as the KP case, i.e., the transformed object $A$ under $n$-step Darboux transformation.
It is noted that $T_1$, $T_2$ and $T_3$ can not commute with each other, i.e., $T_jT_l\neq T_lT_j$ with $j,l=1,2,3$.
 So another two types of Darboux transformation operators are introduced \cite{Oevelrmp1993,Shaw1997,Chengjnmp2018}, that is,
\begin{align*}
T_D(\hat\Phi)\triangleq T_2(1^{[1]})T_1(\hat\Phi)=(\hat\Phi^{-1})_x^{-1}\pa\hat\Phi^{-1},\quad
T_I(\hat\Psi)\triangleq T_1(1^{[1]})T_3(\hat\Psi)=\hat\Psi^{-1}\pa^{-1}\hat\Psi_x,
\end{align*}
which are showed that they can commute with each other, i.e., $T_\mu T_\nu=T_\nu T_\mu$ with $\mu,\nu=D, I$,  and therefore they are more applicable. One can see \cite{Chengjnmp2018} for more details.

Under the Darboux transformations of the modified KP hierarchy, the objects in the modified KP hierarchy are transformed in the way shown in Table IV \cite{Chengjnmp2018}.
\begin{center}
\begin{tabular}{llllll}
\multicolumn{6}{c}{Table IV. Elementary Darboux transformations: modified KP $\rightarrow$ modified KP}\\
\hline\hline
$L_{mKP}\rightarrow L_{mKP}^{[1]}$ &$Z^{[1]}=$ & $\hat\Phi_1^{[1]}=$ &$\hat\Psi_1^{[1]}=$&$\tau_0^{[1]}=$&$\tau_1^{[1]}=$\\
\hline
$T_1=\hat\Phi^{-1}$ &$\hat\Phi^{-1}Z$ &$\hat\Phi^{-1}\hat\Phi_1$ & $\int\hat\Phi\hat\Psi_{1x}dx$&$\tau_0$&$\hat\Phi\tau_1$\\
$T_2=\hat\Phi_x^{-1}\pa$  &  $\hat\Phi_x^{-1}\pa Z\pa^{-1}$ & $\hat\Phi_x^{-1}\hat\Phi_{1x}$ &$-\int\hat\Phi_x\Psi_{1}dx$& $\tau_1$&$\hat\Phi_x\tau_1^2/\tau_0$\\
$T_3=\pa^{-1}\hat\Psi_x$  &$\pa^{-1}\hat\Psi_x Z\pa$ & $\int\hat\Psi_x\hat\Phi_1 dx$ & $-\hat\Psi_x^{-1}\hat\Psi_{1x}$&  $\hat\Psi_x\tau_0^2/\tau_1$&$\tau_0$\\
$T_D=(\hat\Phi^{-1})_x^{-1}\pa\hat\Phi^{-1}$ &  $T_D(\hat\Phi)Z\pa^{-1}$ & $(\hat\Phi_1/\hat\Phi)_x/(\hat\Phi^{-1})_x$ &  $\int\hat\Phi_x\hat\Psi_{1}dx/\hat\Phi$&
$\hat\Phi\tau_1$& $-\hat\Phi_x\tau_1^2/\tau_0$ \\
$T_I=\hat\Psi^{-1}\pa^{-1}\hat\Psi_x$  &  $T_I(\hat\Psi) Z\pa$ &  $\int\hat\Psi_x\hat\Phi_1dx/\hat\Psi$ &$(\hat\Psi_1/\hat\Psi)_x/(\hat\Psi^{-1})_x$& $\hat\Psi_x\tau^2_0/\tau_1$ &$\hat\Psi\tau_0$\\
\hline\hline
\end{tabular}
\end{center}
{\bf Remark:} The changes of the tau functions under the Darboux transformation presented in the table above, are obtained by comparing the changes of the dressing operator. These proofs are not strict. In fact, we need to further show that the new tau functions satisfy the bilinear equations (\ref{mkptaubilinear}) of the modified KP hierarchy. We will discuss this questions in the next subsection.

Before further discussion, the lemma below are needed.
\begin{lemma}
Assume $\hat\Phi$ and $\hat\Psi$ are the eigenfunction and the adjoint eigenfunction of the modified KP hierarchy respectively, $w(t,\lambda)$ and $w^*(t,\lambda)$ are the wave and adjoint wave functions, and $\tau_0(t)$ and $\tau_1(t)$ are the tau functions. Then
\begin{eqnarray*}
&&T_1(\hat\Phi)(w(t,\lambda))=\frac{\tau_0^{[1]}(t-\varepsilon(\lambda^{-1}))}{ \tau^{[1]}_1(t)}e^{\xi(t,\lambda)},
\pa^{-1}(T_1^*(\hat\Phi))^{-1}\pa(w^*(t,\lambda))=\frac{ \tau^{[1]}_1(t+\varepsilon(\lambda^{-1}))}{\lambda\tau^{[1]}_0(t)}e^{-\xi(t,\lambda)};\nonumber\\
&&T_2(\hat\Phi)(w(t,\lambda))=\frac{\lambda\tau_0^{[1]}(t-\varepsilon(\lambda^{-1}))}{ \tau^{[1]}_1(t)}e^{\xi(t,\lambda)},
\pa^{-1}(T_2^*(\hat\Phi))^{-1}\pa(w^*(t,\lambda))=\frac{ \tau^{[1]}_1(t+\varepsilon(\lambda^{-1}))}{\lambda^2\tau^{[1]}_0(t)}e^{-\xi(t,\lambda)};\nonumber\\
&&T_3(\hat\Psi)(w(t,\lambda))=\frac{\tau_0^{[1]}(t-\varepsilon(\lambda^{-1}))}{\lambda \tau^{[1]}_1(t)}e^{\xi(t,\lambda)},
\pa^{-1}(T_3^*(\hat\Psi))^{-1}\pa(w^*(t,\lambda))=\frac{ \tau^{[1]}_1(t+\varepsilon(\lambda^{-1}))}{\tau^{[1]}_0(t)}e^{-\xi(t,\lambda)}
\end{eqnarray*}
and
\begin{eqnarray*}
&&T_D(\hat\Phi)(w(t,\lambda))=\frac{\lambda\tau_0^{[1]}(t-\varepsilon(\lambda^{-1}))}{ \tau^{[1]}_1(t)}e^{\xi(t,\lambda)},
\Big(\pa^{-1}(T_D^*(\hat\Phi))^{-1}\pa\Big)(w^*(t,\lambda))=\frac{ \tau^{[1]}_1(t+\varepsilon(\lambda^{-1}))}{\lambda^2\tau^{[1]}_0(t)}e^{-\xi(t,\lambda)};\nonumber\\
&&T_I(\hat\Psi)(w(t,\lambda))=\frac{\tau_0^{[1]}(t-\varepsilon(\lambda^{-1}))}{\lambda \tau^{[1]}_1(t)}e^{\xi(t,\lambda)},
\Big(\pa^{-1}(T_I^*(\hat\Psi))^{-1}\pa\Big)(w^*(t,\lambda))=\frac{ \tau^{[1]}_1(t+\varepsilon(\lambda^{-1}))}{\tau^{[1]}_0(t)}e^{-\xi(t,\lambda)}.
\end{eqnarray*}
\end{lemma}
\begin{proof}
Note that the results for $T_D$ and $T_I$ can be derived from the ones for $T_1$, $T_2$ and $T_3$. So we here only need to discuss the case of $T_1$, $T_2$ and $T_3$.

Firstly $T_1(\hat\Phi)(w)$ can be obtained by direct computation. The results for $\Big(\pa^{-1}(T_1^*(\hat\Phi))^{-1}\pa\Big)(w^*)$ can be proved by using Proposition \ref{sepexpressionmkp}. As for $T_2(\hat\Phi)(w)$ and $\Big(\pa^{-1}(T_3^*(\hat\Psi))^{-1}\pa\Big)(w^*)$, they can be derived by Lemma \ref{mkpfaylemma}. Next, we will concentrate on the proofs of $\Big(\pa^{-1}(T_2^*(\hat\Phi))^{-1}\pa\Big)(w^*)$ and $T_3(\hat\Psi)(w)$. Actually according to Proposition \ref{sepexpressionmkp} and Lemma \ref{mkpfaylemma},
\begin{eqnarray*}
\int w(t,z)_x w^*(t,\lambda)dx=-\frac{\tau^2_1(t+\varepsilon(\lambda^{-1}))}
{\tau_0(t+\varepsilon(\lambda^{-1})\tau_1(t)}
w(t+\varepsilon(\lambda^{-1}),z)_x\lambda^{-2}e^{-\xi(t,\lambda)},
\end{eqnarray*}
where we have used $\lambda^{-1}e^{\xi(\varepsilon(\lambda^{-1}),z)}
=-z^{-1}e^{\xi(\varepsilon(z^{-1}),\lambda)}$.
Therefore by the spectral representation of the eigenfunction $\hat\Phi(t)={\rm Res} _{\lambda}\hat\rho(\lambda)w(t,\lambda)$ (see Corollary \ref{corspectralmkp}).
\begin{eqnarray*}
\Big(\pa^{-1}(T_2^*(\Phi))^{-1}\pa\Big)(w^*(t,\lambda))=-{\rm Res}_{z}\hat\rho(z)\int w(t,z)_xw^*(t,\lambda)dx=\frac{\tau^{[1]}_1(t+\varepsilon(\lambda^{-1}))}
{\lambda^2\tau^{[1]}_0(t)}
e^{-\xi(t,\lambda)}.
\end{eqnarray*}
The result of $T_3(\hat\Psi)(w)$ can be proved in the similar way.
\end{proof}

\subsection{Bilinear relations for $\tau_0^{[1]}$ and $\tau_1^{[1]}$}
Just as showed in the last section, one can assume the tau pair of the modified KP hierarchy to be\footnote{In fact, other forms such as $(\tau_0,\tau_1)=\Big(\hat\alpha_0^*g|0\rangle, g|0\rangle\Big)$ with $\hat\alpha_0^*\in V^*$ can also be allowed and the corresponding results are the same.} $(\tau_0,\tau_1)=\Big(g|0\rangle,\hat\alpha_0g|0\rangle\Big)$ with
$g\in GL_\infty$ and $\hat\alpha_0\in V$. Then if denote
\begin{align*}
\hat\alpha={\rm Res}_{\lambda}\hat\rho(\lambda)\psi(\lambda),\quad \hat\alpha^*={\rm Res}_{\lambda}\lambda^{-1}\hat\rho^*(\lambda)\psi^*(\lambda),
\end{align*}
with $\hat\rho(\lambda)$ and $\hat\rho^*(\lambda)$ given in Corollary \ref{corspectralmkp}, then the eigenfunction $\hat\Phi$ and the adjoint eigenfunction $\hat\Psi$ can be written in the forms below.
\begin{eqnarray}
 \hat\Phi(t)=\frac{\sigma_t(\hat\alpha \tau_0)}{\sigma_t
 (\tau_1)},
 &&\hat\Psi(t)=\frac{\sigma_t(\hat\alpha^* \tau_1)}{\sigma_t
 (\tau_0)}.\label{mkpeigenfermionic}
\end{eqnarray}
Then one can rewrite the Darboux transformations of the modified KP hierarchy in the Fermionic forms by (\ref{mkpeigenfermionic}), which are given in the table below.
\begin{center}
\begin{tabular}{lll}
\multicolumn{3}{c}{Table V. Fermionic Darboux transformations: modified KP $\rightarrow$ modified KP}\\
\hline\hline
$L_{mKP}\rightarrow L_{mKP}^{[1]}$ & \ \ \ \ \ \ \ \ \ \ \ \ \    $\tau_0^{[1]}=$& \ \ \ \ \ \ \ \ \ \ \ \ \   $\tau_1^{[1]}=$\\
\hline
\ \ $T_1=\hat\Phi^{-1}$ & \ \ \ \ \ \ \ \ \ \ \ \ \  $\tau_0$& \ \ \ \ \ \ \ \ \ \ \ \ \  $\hat\alpha\tau_0$\\
\ \ $T_2=\hat\Phi_x^{-1}\pa$  &\ \ \ \ \ \ \ \ \ \ \ \ \ $\tau_1$& \ \ \ \ \ \ \ \ \ \ \ \ \ $\hat\alpha\tau_1$\\
\ \ $T_3=\pa^{-1}\hat\Psi_x$  &\ \ \ \ \ \ \ \ \ \ \ \ \  $-\hat\alpha^*\tau_0$& \ \ \ \ \ \ \ \ \ \ \ \ \  $\tau_0$\\
\ \ $T_D=(\hat\Phi^{-1})_x^{-1}\pa\hat\Phi^{-1}$  &\ \ \ \ \ \ \ \ \ \ \ \ \  $\hat\alpha\tau_0$& \ \ \ \ \ \ \ \ \ \ \ \ \  $-\hat\alpha\tau_1$\\
\ \ $T_I=\hat\Psi^{-1}\pa^{-1}\hat\Psi_x$  &\ \ \ \ \ \ \ \ \ \ \ \ \  $-\hat\alpha^*\tau_0$& \ \ \ \ \ \ \ \ \ \ \ \ \  $\hat\alpha^*\tau_1$\\
\hline\hline
\end{tabular}
\end{center}
Here we only illustrate how to convert $\hat\Phi_x\tau_1^2/\tau_0$ into the Fermionic form, and the case $\hat\Psi_x\tau_0^2/\tau_1$ is almost the same. In fact by Lemma \ref{mkpfaylemma} and the spectral representation,
\begin{align*}
\hat\Phi(t)_x\frac{\tau_1(t)^2}{\tau_0(t)}=&
{\rm Res}_\lambda\lambda\hat\rho(\lambda)\tau_1(t-\varepsilon(\lambda^{-1}))e^{\xi(t,\lambda)}
={\rm Res}_\lambda\langle 2|e^{H(t)}\psi(\lambda)\hat\alpha_0g|0\rangle\\
=&\langle 2|e^{H(t)}\hat\alpha\hat\alpha_0g|0\rangle=\sigma_t(\hat\alpha\tau_1).
\end{align*}

Then according to Theorem \ref{stnnkkotnk}, one can obtain
\begin{align}
S\Big(\tau_0^{[1]}\otimes\tau_1^{[1]}\Big)=\tau_1^{[1]}\otimes\tau_0^{[1]},
\end{align}
where for the case of $T_I$, it is proved by Lemma \ref{s4beta} and the fact $S(\tau_0\otimes\tau_1)=\tau_1\otimes\tau_0$.
The corresponding Bosonic forms are presented in the theorem below.
\begin{theorem}\label{mkptransformation}
Let $(\tau_0,\tau_1)$ be the tau functions of the modified KP hierarchy, and denote $\hat \Phi$ and $\hat\Psi$ to be the eigenfunction and the adjoint eigenfunction respectively. Then the following relations hold.
\begin{align*}
&{\rm Res}_{\lambda}\lambda^{-1}\hat\Phi(t'+\varepsilon(\lambda^{-1}))\tau_0(t-\varepsilon(\lambda^{-1}))
\tau_1(t'+\varepsilon(\lambda^{-1}))
e^{\xi(t-t',\lambda)}=\hat\Phi(t)\tau_1(t)\tau_0(t'),\\
&{\rm Res}_{\lambda}\lambda^{-1}\hat\Phi_x(t'+\varepsilon(\lambda^{-1}))\tau_1(t-\varepsilon(\lambda^{-1}))
\frac{\tau_1^2(t'+\varepsilon(\lambda^{-1}))}{\tau_0(t'+\varepsilon(\lambda^{-1}))}
e^{\xi(t-t',\lambda)}=\frac{\hat\Phi_x(t)\tau_1^2(t)}{\tau_0(t)}\tau_1(t'),\\
&{\rm Res}_{\lambda}\lambda^{-1}\hat\Psi_x(t-\varepsilon
(\lambda^{-1}))\frac{\tau_0^2(t-\varepsilon(\lambda^{-1}))}
{\tau_1(t-\varepsilon(\lambda^{-1}))}
\tau_0(t'+\varepsilon(\lambda^{-1}))
e^{\xi(t-t',\lambda)}=\tau_0(t)\frac{\hat\Psi_x(t')\tau_0^2(t')}{\tau_1(t')},
\end{align*}
and
\begin{align*}
{\rm Res}_{\lambda}&\lambda^{-1}\hat\Phi(t-\varepsilon(\lambda^{-1}))
\tau_1(t-\varepsilon(\lambda^{-1}))\hat\Phi_x(t'+\varepsilon(\lambda^{-1}))\frac{\tau_1^2(t'+\varepsilon(\lambda^{-1}))}{\tau_0(t'+\varepsilon
(\lambda^{-1}))}
e^{\xi(t-t',\lambda)}\nonumber\\
&
=\frac{\hat\Phi_x(t)\tau_1^2(t)}{\tau_0(t)}\hat\Phi(t')\tau_1(t'),\\
{\rm Res}_{\lambda}&\lambda^{-1}\hat\Psi_x(t-\varepsilon
(\lambda^{-1}))\frac{\tau_0^2(t-\varepsilon(\lambda^{-1}))}
{\tau_1(t-\varepsilon(\lambda^{-1}))}\hat\Psi(t'+\varepsilon(\lambda^{-1}))\tau_0(t'+\varepsilon(\lambda^{-1}))
e^{\xi(t-t',\lambda)}\nonumber\\
&
=\hat\Psi(t)\tau_0(t)\frac{\hat\Psi_x(t')\tau_0^2(t')}{\tau_1(t')}.
\end{align*}
\end{theorem}
\noindent{\bf Remark:} The relations in this theorem are usually very hard to prove in the Bosonic form, while they can be easily obtained in the Fermionic approach. By this theorem, it is again confirmed that the transformed tau functions under the Darboux transformations are correct.
\subsection{Successive applications of the Darboux transformations}
By using the results in Table IV, one can obtain the lemma below.
\begin{lemma}\label{sepmkpchange}
Under the Darboux transformation of the modified KP hierarchy, the SEPs: $\Omega(\hat\Phi_1,w^*(x,\lambda)_x)$ and $\hat \Omega(w(t,\lambda)_x,\hat\Psi_1)$ are transformed in the way below.
\begin{center}
\begin{tabular}{lll}
\multicolumn{3}{c}{Table VI. The transformed SEPs of the modified KP hierarchy}\\
\hline\hline
$L_{mKP}\rightarrow L_{mKP}^{[1]}$ &     $\Omega(\hat\Phi_1^{[1]},w_x^{*[1]})=$& $\hat \Omega(w_x^{[1]},\hat\Psi_1^{[1]})=$\\
\hline
\ \ $T_1=\hat\Phi^{-1}$ & $\Omega(\hat\Phi_1,w^*_x)$&$\hat \Omega(w_x,\hat\Psi_{1})-\frac{\hat \Omega(\hat\Phi_x,\hat\Psi_{1})}{\hat\Phi}w$\\
\ \ $T_2=\hat\Phi_x^{-1}\pa$  & $\lambda\Big(\Omega(\hat\Phi_1,w^*_x)-\hat\Phi_1w^*
\Big)$&$\lambda^{-1}\Big(\hat \Omega(w_x,\hat\Psi_1)-\frac{\hat \Omega(\hat\Phi_{x},\hat\Psi_1)}
{\hat\Phi_{x}}w_x\Big)$\\
\ \ $T_3=\pa^{-1}\hat\Psi_x$  &  $\lambda^{-1}\Big(\Omega(\hat\Phi_1,w^*_x)
-\frac{\Omega(\hat\Phi_1,\hat\Psi_{x})}{\hat\Psi_{x}}w^*_x\Big)$&$\lambda\Big(\hat \Omega(w_x,\hat\Psi_1)-\hat\Psi_1w
\Big)$\\
\ \ $T_D=(\hat\Phi^{-1})_x^{-1}\pa\hat\Phi^{-1}$  & $\lambda\Big(\Omega(\hat\Phi_1,w^*_x)
-\frac{\hat\Phi_1}{\hat\Phi}\Omega(\hat\Phi,w^*_x)
\Big)$& $\lambda^{-1}\Big(\hat \Omega(w_x,\hat\Psi_{1})-\frac{\hat \Omega(\hat\Phi_x,\hat\Psi_1)}{\hat\Phi_x}w_x\Big)$\\
\ \ $T_I=\hat\Psi^{-1}\pa^{-1}\hat\Psi_x$  & $\lambda^{-1}
\Big(\Omega(\hat\Phi_1,w^*)_x)
-\frac{w^*_x}{\hat\Psi_{x}}\Omega(\hat\Phi_1,\hat\Psi_{x})
\Big)$& $\lambda\Big(\hat \Omega(w_x,\hat\Psi_1)
-\frac{\hat\Psi_1}{\hat\Psi}\hat \Omega(w_x,\hat\Psi)
\Big)$\\
\hline\hline
\end{tabular}
\end{center}

\end{lemma}

Next the successive applications of the Darboux transformations are given in the proposition below.
\begin{proposition}\label{5DTsuccessive}
Given $\tau_0$ and $\tau_1$ for the modified KP hierarchy, the corresponding changes under the $n$-step Darboux transformation are given in the following table.
\begin{center}
\begin{tabular}{lll}
\multicolumn{3}{c}{Table VII. Fermionic Darboux transformations: mKP $\rightarrow$ mKP}\\
\hline\hline
$L_{mKP}\rightarrow L_{mKP}^{[\vec{\bf n}]}$ & \ \ \ \ \ \ \ \ \ \ \ \ \    $\tau_0^{[\vec{\bf n}]}=$& \ \ \ \ \ \ \ \ \ \ \ \ \   $\tau_1^{[\vec{\bf n}]}=$\\
\hline
\ \ $T_1=\hat\Phi^{-1}$ & \ \ \ \ \ \ \ \ \ \ \ \ \  $\tau_0$& \ \ \ \ \ \ \ \ \ \ \ \ \  $\hat\alpha_{n}\tau_0$\\
\ \ $T_2=\hat\Phi_x^{-1}\pa$  &\ \ \ \ \ \ \ \ \ \ \ \ \ $\hat\alpha_{\overrightarrow{\bf n-1}}\tau_1$& \ \ \ \ \ \ \ \ \ \ \ \ \ $\hat\alpha_{\vec{\bf n}}\tau_1$\\
\ \ $T_3=\pa^{-1}\hat\Psi_x$  &\ \ \ \ \ \ \ \ \ \ \ \ \  $(-1)^n\hat\alpha^*_{\vec{\bf n}}\tau_0$& \ \ \ \ \ \ \ \ \ \ \ \ \  $(-1)^{n-1}\hat\alpha^*_{\overrightarrow{\bf n-1}}\tau_0$\\
\ \ $T_D=(\hat\Phi^{-1})_x^{-1}\pa\hat\Phi^{-1}$  &\ \ \ \ \ \ \ \ \ \ \ \ \  $\hat\alpha_{\vec{\bf n}}\tau_0$& \ \ \ \ \ \ \ \ \ \ \ \ \  $(-1)^n\hat\alpha_{\vec{\bf n}}\tau_1$\\
\ \ $T_I=\hat\Psi^{-1}\pa^{-1}\hat\Psi_x$  &\ \ \ \ \ \ \ \ \ \ \ \ \  $(-1)^n\hat\alpha_{\vec{\bf n}}^*\tau_0$& \ \ \ \ \ \ \ \ \ \ \ \ \  $\hat\alpha_{\vec{\bf n}}^*\tau_1$\\
\hline\hline
\end{tabular}
\end{center}
Here $A^{[\vec{\bf n}]}$ is the same as the one in Section 4, $\hat\alpha_i\in V$  and $\hat\alpha_j^*\in V^*$ are corresponding to the eigenfunction $\Phi_i$ and the adjoint eigenfunction $\Psi_j$ respectively. Sometimes we use $A^{+[\vec{\bf n}]}$ and $A^{-[\vec{\bf n}]}$ to denote the transformed object $A$ under $T_D$ and $T_I$ respectively.
\end{proposition}
\begin{proof}
Firstly, the results in case $T_1$ are obviously by using Lemma \ref{sepmkpchange}. Next we will try to prove the cases of $T_2$ and $T_D$, and other two cases are similar. According to Table V, one can find   $\tau_0^{[n]}=\tau_1^{[n-1]}$ for $T_2$. So here we only need to investigate $\tau_1^{[n]}$. By using Lemma \ref{sepmkpchange},
\begin{align*}
&\hat\rho^{[n-1]}_n(\lambda)=\Omega(\hat\Phi^{(n-1)}_n(t'),w^{(n-1)*}(t',\lambda)_x)
=\lambda^{n-1}\hat\rho_n(\lambda)
-\sum_{i=1}^{n-1}\lambda^i\Phi_{n}^{[n-i-1]}(t')w^{[n-i-1]*}(t',\lambda).
\end{align*}
Further by the similar way in Proposition \ref{ndttd},
\begin{eqnarray*}
\hat\alpha^{[n-1]}_n={\rm Res}_{\lambda}\frac{1}{\lambda^{n-1}}
\hat\rho^{[n-1]}_n\psi(\lambda),\quad \tau_1^{[n]}=\alpha_n^{[n-1]}\tau_1^{[n-1]}.
\end{eqnarray*}
After inserting the expression of $\hat\rho^{[n-1]}_n(\lambda)$,
\begin{align*}
\hat\alpha^{[n-1]}_n
=\hat\alpha_n-\sum_{l\in\mathbb{Z}}\sum_{i=1}^{n-1}
\frac{\Phi_{n}^{[n-i-1]}(t')}{\tau_0^{[n-i-1]}(t')}\sigma_{t'}
(\psi_l^*\hat\alpha_{\overrightarrow{\bf n-i-1}}\tau_1)\cdot\psi_l
\end{align*}
Therefore if assume $\tau_1^{[n-1]}=\hat\alpha_{\overrightarrow{\bf n-1}}\tau_1$, then
\begin{align*}
&\tau_1^{[n]}=\hat\alpha_n^{[n-1]}\hat\alpha_{\overrightarrow{\bf n-1}}\tau_1=\hat\alpha_n\hat\alpha_{\overrightarrow{\bf n-1}}\tau_1-\sum_{i=1}^{n-1}
\frac{\Phi_{n}^{[n-i-1]}(t')}{\tau_0^{[n-i-1]}(t')}S(\hat\alpha_{\overrightarrow{\bf n-1}}\tau_1\otimes\hat\alpha_{\overrightarrow{\bf n-i-1}}\tau_1)\\
=&\hat\alpha_{\vec{\bf n}}\tau_1-\sum_{i=1}^{n-1}
\frac{\Phi_{n}^{[n-i-1]}(t')}{\tau_0^{[n-i-1]}(t')}(-1)^i\Big(\hat\alpha_
{\overrightarrow{\bf n-1}\setminus\overrightarrow{\bf n-i-1}} \otimes 1\Big)\cdot S(\hat\alpha_{\overrightarrow{\bf n-i-1}}\tau_1\otimes\hat\alpha_{\overrightarrow{\bf n-i-1}}\tau_1)=\hat\alpha_{\vec{\bf n}}\tau_1.
\end{align*}

As for $T_D$, we only need to show
\begin{align*}
\hat\alpha_n^{[n-1]}\cdots\hat\alpha_2^{[1]}\hat\alpha_1
=\hat\alpha_n\cdots\hat\alpha_2\hat\alpha_1,
\end{align*}
with $\hat\alpha^{[n-1]}_n={\rm Res}_{\lambda}\lambda^{1-n}
\hat\rho^{[n-1]}_n(\lambda)\psi(\lambda)$ and $\hat\rho^{[n-1]}_n(\lambda)=\Omega(\hat\Phi^{[n-1]}_n(t'),w^{[n-1]*}(t',\lambda)_x)$. In fact, this can be proved by using the relations below derived by Lemma \ref{sepmkpchange}
\begin{align*}
\hat\rho^{[n-1]}_n(\lambda)=\lambda\left(\hat\rho^{[n-2]}_n(\lambda)
-\frac{\hat\Phi^{[n-2]}_n(t')}{\hat\Phi^{[n-2]}_{n-1}(t')}
\hat\rho^{[n-2]}_{n-1}(\lambda)\right),\quad \hat\alpha^{[n-1]}_n=\hat\alpha^{[n-2]}_n
-\frac{\hat\Phi^{[n-2]}_n(t')}{\hat\Phi^{[n-2]}_{n-1}(t')}
\hat\alpha^{[n-2]}_{n-1}.
\end{align*}
\end{proof}
Next consider the successive mixed applications of $T_D$ and $T_I$. And let $\hat T^{\pm[\vec{\bf n}, \vec{\bf k}]}$ has the same meanings as  $T^{\pm[\vec{\bf n}, \vec{\bf k}]}$ in the KP case, just replacing $T_d$ and $T_i$ with $T_D$ and $T_I$ respectively. $A^{\pm[\vec{\bf n},\vec{\bf k}]}$ denote the transformed object $A$ under $\hat T^{\pm[\vec{\bf n}, \vec{\bf k}]}$ such as the transformed tau pair $\Big(\tau_0^{\pm[\vec{\bf n},\vec{\bf k}]}, \tau_1^{\pm[\vec{\bf n},\vec{\bf k}]}\Big)$.

\begin{proposition}\label{mkptdti}
The expressions of
$\Big(\tau_0^{\pm[\vec{\bf n},\vec{\bf k}]}, \tau_1^{\pm[\vec{\bf n},\vec{\bf k}]}\Big)$ are given as follows
\begin{align*}
&\Big(\tau_0^{+[\vec{\bf n},\vec{\bf k}]}, \tau_1^{+[\vec{\bf n},\vec{\bf k}]}\Big)=\hat\alpha_k^{*[n]}\cdots\hat\alpha_1^{*[n]}
\hat\alpha_n\cdots\hat\alpha_1\Big((-1)^k\tau_0^{[0]},(-1)^n\tau_1^{[0]}\Big),\\
&\Big(\tau_0^{-[\vec{\bf n},\vec{\bf k}]}, \tau_1^{-[\vec{\bf n},\vec{\bf k}]}\Big)=\hat\alpha_n^{[k]}\cdots\hat\alpha_1^{[k]}
\hat\alpha_k^*\cdots\hat\alpha_1^*\Big((-1)^k\tau_0^{[0]},(-1)^n\tau_1^{[0]}\Big),
\end{align*}
where $\hat\alpha_j^{*[n]}={\rm Res}_\lambda\lambda^{n-1}\hat\rho_j^{*[n]}(\lambda)\psi^*(\lambda)$ and $\hat\alpha_l^{[k]}={\rm Res}_\lambda\lambda^{k}\hat\rho_l^{[k]}(\lambda)\psi(\lambda)$.
Further
\begin{align*}
&\Big(\tau_0^{+[\vec{\bf n},\vec{\bf k}]}, \tau_1^{+[\vec{\bf n},\vec{\bf k}]}\Big)=\sum_{a=0}^k\sum_{\vec{\gamma}\in H_{k,a}}\sum_{\vec{\delta}\in H_{n,n-a}}\hat C^{+[n,k]}_{a,\vec{\gamma},\vec{\delta}}
\hat\alpha^*_{\vec{\bf k}\setminus\vec{\gamma}}\hat\alpha_{\vec{\delta}}
\Big((-1)^k\tau_0^{[0]},(-1)^n\tau_1^{[0]}\Big),\\
&\Big(\tau_0^{-[\vec{\bf n},\vec{\bf k}]}, \tau_1^{-[\vec{\bf n},\vec{\bf k}]}\Big)=\sum_{a=0}^n\sum_{\vec{\gamma}\in H_{n,a}}\sum_{\vec{\delta}\in H_{k,k-a}}\hat C^{-[n,k]}_{a,\vec{\gamma},\vec{\delta}}
\hat\alpha_{\vec{\bf n}\setminus\vec{\gamma}}\hat\alpha^*_{\vec{\delta}}
\Big((-1)^k\tau_0^{[0]},(-1)^n\tau_1^{[0]}\Big).
\end{align*}
Here $\hat C^{\pm[n,k]}_{a,\vec{\gamma},\vec{\delta}}$ with $\hat C^{+[n,k]}_{0,0,\vec{\bf n}}=\hat C^{-[n,k]}_{0,0,\vec{\bf k}}=1$ can be viewed as the constants independent of $t$, which satisfies the following recursion relations,
\begin{align*}
\hat C^{+[n,k+1]}_{a,\vec{\gamma},\vec{\delta}}=\hat C^{+[n,k]}_{a,\vec{\gamma},\vec{\delta}},\quad
\hat C^{+[n,k+1]}_{a+1,\{k+1\}\cup\vec{\gamma},\vec{\delta}\setminus\{\delta_l\}}=
\hat C^{+[n,k]}_{a,\vec{\gamma},\vec{\delta}}\sum_{i=1}^n(-1)^{n+k-l}\frac{\hat \Omega\Big(\hat \Phi_{i,x}^{[i-1]}(t'),\hat\Psi_{k+1}^{[i-1]}(t')\Big)
\tau_1^{+[\{\delta_l\}\cup\overrightarrow{\bf i-1}]}(t')}
{\hat\Phi_{i,x}^{[i-1]}(t')
\tau_1^{+[i-1]}(t')^2/\tau_0^{+[i-1]}(t')},\\
\hat C^{-[n+1,k]}_{a,\vec{\xi},\vec{\eta}}=\hat C^{-[n,k]}_{a,\vec{\xi},\vec{\eta}},\quad\hat C^{-[n+1,k]}_{a+1,\{n+1\}\cup\vec{\xi},\vec{\eta}\setminus\{\eta_l\}}=
\hat C^{-[n,k]}_{a,\vec{\xi},\vec{\eta}}\sum_{i=1}^k(-1)^{n+k-l}\frac{ \Omega\Big(\hat \Phi_{n+1}^{[i-1]}(t'),\hat\Psi_{i,x}^{[i-1]}(t')\Big)
\tau_0^{-[\{\eta_l\}\cup\overrightarrow{\bf i-1}]}(t')}
{\hat\Phi_{i,x}^{[i-1]}(t')
\tau_0^{-[i-1]}(t')^2/\tau_1^{-[i-1]}(t')},
\end{align*}
where $\vec{\gamma}\in H_{k,a}$, $\vec{\delta}\in n+H_{n_1,n_1-a}$, $\vec{\xi}\in H_{n,a}$, $\vec{\eta}\in k+H_{k_1,k_1-a}$.
\end{proposition}
\begin{proof}
Firstly according to Lemma \ref{mkpfaylemma} and Lemma \ref{sepmkpchange}, one can obtain
\begin{align*}
&\hat\rho_j^{*[n]}(\lambda)=\lambda^{-n}\hat\rho_j^{*[0]}(\lambda)-\sum_{i=1}^n\frac{\hat \Omega\Big(\hat \Phi_{i,x}^{[i-1]}(t'),\hat\Psi_j^{[i-1]}(t')\Big)}
{\lambda^{n+1-i}\hat\Phi_{i,x}^{[i-1]}(t')}w_x^{[i-1]}(t',\lambda),\\
&\hat\alpha_j^{*[n]}=\hat\alpha_j^{*[0]}
-\sum_{l\in\mathbb{Z}}
\sum_{i=1}^n(-1)^{i-1}\frac{\hat \Omega\Big(\hat \Phi_{i,x}^{[i-1]}(t'),\hat\Psi_j^{[i-1]}(t')\Big)
\tau_0^{+[i-1]}(t')}
{\hat\Phi_{i,x}^{[i-1]}(t')
\tau_1^{+[i-1]}(t')^2}\sigma_{t'}
(\psi_l\hat\alpha_{\overrightarrow{\bf i-1}}\tau_1^{[0]})\cdot\psi^*_l.
\end{align*}
Then by the similar methods in Proposition \ref{kptdtitau}, one can prove this proposition.
\end{proof}
\noindent{\bf Remark:} Note that in the modified KP case, the tau pairs $(\tau_0,\tau_1)$ and $(\tau_0',\tau_1')$ can determine the same dressing structures if and only if $(\tau_0,\tau_1)=c(\tau_0',\tau_1')$. Therefore the sign in Proposition
\ref{5DTsuccessive} and Proposition \ref{mkptdti} can not be moved to ensure that the transformed modified KP systems share the same Lax and dressing structures.

By similar methods as Corollary \ref{coralpdown} and \ref{coralpup}, one can obtain the two corollaries below.
\begin{corollary}
For any $\vec{\gamma}\in H_{k,a}$ and $\vec{\delta}\in H_{n,b}$, one has the following relation
\begin{align*}
\hat\alpha_{\vec{\gamma}}^{*+[\vec{\bf n}]}\hat\alpha_{\vec{\delta}}\Big(\tau_{0}^{[0]},\tau_1^{[0]}\Big)
=\hat\alpha_{\vec{\gamma}}^{*+[\vec{\delta}]}
\hat\alpha_{\vec{\delta}}\Big(\tau_{0}^{[0]},\tau_1^{[0]}\Big),\quad \hat \alpha_{\vec{\gamma}}^{-[\vec{\bf n}]}\hat\alpha^*_{\vec{\delta}}\Big(\tau_{0}^{[0]},\tau_1^{[0]}\Big)
=\hat\alpha_{\vec{\gamma}}^{-[\vec{\delta}]}
\hat\alpha^*_{\vec{\delta}}\Big(\tau_{0}^{[0]},\tau_1^{[0]}\Big).
\end{align*}
\end{corollary}
\begin{corollary}\label{corhatalpup}
$\hat\alpha_{\vec{\bf k}}^{*+[\vec{\bf n}]}{\hat\alpha}_{\overrightarrow{\bf n_1+n}}$ and $\hat\alpha_{\vec{\bf n}}^{-[\vec{\bf k}]}{\hat\alpha}^*_{\overrightarrow{\bf k_1+k}}$ act on $\Big(\tau_0^{[0]},\tau_1^{[0]}\Big)$ in the way below
\begin{align*}
&\hat\alpha_{\vec{\bf k}}^{*+[\vec{\bf n}]}\hat{\alpha}_{\overrightarrow{\bf n_1+n}}\Big(\tau_0^{[0]},\tau_1^{[0]}\Big)
=\sum_{a=0}^k\sum_{\vec{\gamma}\in H_{k,a}}\sum_{\vec{\delta}\in n+H_{n_1,n_1-a}}\hat{\mathcal{C}}^{+[k,n_1,n]}_{a,\vec{\gamma},\vec{\delta}}\hat\alpha
^{*+[\vec{\delta}\cup\vec{\bf n}]}_{\vec{\bf k}\setminus\vec{\gamma}}\hat\alpha_{\vec{\delta}\cup \vec{\bf n}}
\Big(\tau_0^{[0]},\tau_1^{[0]}\Big),\\
&\hat\alpha_{\vec{\bf n}}^{-[\vec{\bf k}]}{\hat\alpha}^*_{\overrightarrow{\bf k_1+k}}\Big(\tau_0^{[0]},\tau_1^{[0]}\Big)
=\sum_{a=0}^n\sum_{\vec{\gamma}\in H_{n,a}}\sum_{\vec{\delta}\in k+H_{k_1,k_1-a}}\hat{\mathcal{C}}^-_{a,\vec{\gamma},\vec{\delta}}\hat\alpha
^{-[\vec{\delta}\cup\vec{\bf k}]}_{\vec{\bf n}\setminus\vec{\gamma}}\hat\alpha^*_{\vec{\delta}\cup \vec{\bf k}}
\Big(\tau_0^{[0]},\tau_1^{[0]}\Big).
\end{align*}
Here $\hat{\mathcal{C}}^{+[k,n_1,n]}_{a,\vec{\gamma},\vec{\delta}}$ and $\hat{\mathcal{C}}^{-[n,k_1,k]}_{a,\vec{\gamma},\vec{\delta}}$ with $\hat{\mathcal{C}}^{+[k,n_1,n]}_{0,0,\overrightarrow{\bf n_1}}=\hat{\mathcal{C}}^{-[n,k_1,k]}_{0,0,\overrightarrow{\bf k_1}}=1$ are some constants independent of $t$, satisfying
$\hat{\mathcal{C}}^{+[k+1,n_1,n]}_{a,\vec{\gamma},\vec{\delta}}
=\hat{\mathcal{C}}^{+[k,n_1,n]}_{a,\vec{\gamma},\vec{\delta}}$, $\hat{\mathcal{C}}^{-[n+1,k_1,k]}_{a,\vec{\gamma},\vec{\delta}}
=\hat{\mathcal{C}}^{-[n,k_1,k]}_{a,\vec{\gamma},\vec{\delta}}$ and
\begin{align*}
&\hat{\mathcal{C}}^{+[k+1,n_1,n]}_{a+1,\{k+1\}\cup\vec{\gamma},\vec{\delta}\setminus\{\delta_l\}}=
\hat{\mathcal{C}}^{+[k,n_1,n]}_{a,\vec{\gamma},\vec{\delta}}\sum_{i=1}^l(-1)^{n_1+k-l-1}\frac{\hat \Omega\Big(\hat \Phi_{n+\delta_{i},x}^{+[\delta_{\overrightarrow{\bf i-1}}\cup\vec{\bf n}]}(t'),\hat\Psi_{k+1}^{+[\delta_{\overrightarrow{\bf i-1}}\cup\vec{\bf n}]}(t')\Big)
\tau_1^{+[\{\delta_l\}\cup\delta_{\overrightarrow{\bf i-1}}\cup\vec{\bf n}]}(t')}
{\hat\Phi_{n+\delta_{i},x}^{+[\delta_{\overrightarrow{\bf i-1}}\cup\vec{\bf n}]}(t')
\tau_1^{+[\delta_{\overrightarrow{\bf i-1}}\cup\vec{\bf n}]}(t')^2/\tau_0^{+[\delta_{\overrightarrow{\bf i-1}}\cup\vec{\bf n}]}(t')}.\\
&\hat{\mathcal{C}}^{-[n+1,k_1,k]}_{a+1,\{n+1\}\cup\vec{\xi},\vec{\eta}\setminus\{\eta_l\}}=
\hat{\mathcal{C}}^{-[n,k_1,k]}_{a,\vec{\xi},\vec{\eta}}
\sum_{i=1}^l(-1)^{k_1+n-l-1}\frac{ \Omega\Big(\hat \Phi_{n+1}^{-[\eta_{\overrightarrow{\bf i-1}}\cup\vec{\bf n}]}(t'),\hat\Psi_{k+\eta_{i},x}^{-[\eta_{\overrightarrow{\bf i-1}}\cup\vec{\bf n}]}(t')\Big)
\tau_0^{-[\{\eta_l\}\cup\eta_{\overrightarrow{\bf i-1}}\cup\vec{\bf n}]}(t')}
{\hat\Psi_{k+\eta_{i},x}^{-[\eta_{\overrightarrow{\bf i-1}}\cup\vec{\bf n}]}(t')
\tau_0^{-[\eta_{\overrightarrow{\bf i-1}}\cup\vec{\bf n}]}(t')^2/\tau_1^{-[\eta_{\overrightarrow{\bf i-1}}\cup\vec{\bf n}]}(t')},
\end{align*}
with $\vec{\gamma}\in H_{k,a}$, $\vec{\delta}\in n+H_{n_1,n_1-a}$, $\vec{\xi}\in H_{n,a}$, $\vec{\eta}\in k+H_{k_1,k_1-a}$.
\end{corollary}

\subsection{Bilinear equations in the Darboux chains of the modified KP hierarchy}
\label{subsectbilinmkp}
 The bilinear equations in the Darboux chains of the modified KP hierarchy are mainly the ones for $\tau_0^{\pm[i,j]}\otimes \tau_0^{\pm[l,m]}$, $\tau_1^{\pm[i,j]}\otimes \tau_1^{\pm[l,m]}$, $\tau_0^{\pm[i,j]}\otimes \tau_1^{\pm[l,m]}$ and $\tau_1^{\pm[i,j]}\otimes \tau_0^{\pm[l,m]}$ with $i,l\in\{n',n\}$ and $j,m\in\{k',k\}$, which are in fact all contained in the KP case. But there are still some particular bilinear relations for modified KP itself, that is the ones for $\tau_0^{\pm[i,j]}\otimes \tau_1^{\pm[l,m]}$ and $\tau_1^{\pm[i,j]}\otimes \tau_0^{\pm[l,m]}$. In fact by noting the facts $S(\tau_0\otimes\tau_1)=\tau_1\otimes\tau_0$ and $S(\tau_1\otimes\tau_0)=0$, one can get the following results by using similar methods in Theorem \ref{stnnkkotnk} and Theorem \ref{stnnkotnkk}.
 \begin{theorem}
Given $n'\geq n$, $k'\geq k$, $j=0,1$, $l\geq0$,
\begin{align*}
&S^l\Big(\tau_j^{+[\overrightarrow{\bf n'},\overrightarrow{\bf k'}]}\otimes\tau_{1-j}^{+[\vec{\bf n},\vec{\bf k}]}\Big)=l!\sum_{i=0}^{1-j}\sum_{\vec{\gamma}\in k+H_{k'-k,l-i}}(-1)^{(l-i)(k'-1)-|\vec{\gamma}|}\tau_{i+j}^{+[\overrightarrow{\bf n'},\overrightarrow{\bf k'}\setminus\vec{\gamma}]}\otimes\tau_{1-i-j}^{+[\vec{\bf n},\vec{\gamma}\cup\vec{\bf k}]},\\
&S^l\Big(\tau_j^{-[\vec{\bf n},\vec{\bf k}]}\otimes\tau_{1-j}^{-[\overrightarrow{\bf n'},\overrightarrow{\bf k'}]}\Big)=l!\sum_{i=0}^{1-j}\sum_{\vec{\gamma}\in n+H_{n'-n,l-i}}(-1)^{(l-i)(n'-1)-|\vec{\gamma}|}\tau_{i+j}^{[\vec{\gamma}\cup\vec{\bf n},\vec{\bf k}]}\otimes\tau_{1-i-j}^{[\overrightarrow{\bf n'}\setminus\vec{\gamma},\overrightarrow{\bf k'}]},\\
&S\Big(\tau_j^{\pm[\overrightarrow{\bf n'},\vec{\bf k}]}\otimes\tau_{1-j}^{\pm[\vec{\bf n},\overrightarrow{\bf k'}]}\Big)=(1-j)\tau_{1-j}^{\pm[\overrightarrow{\bf n'},\vec{\bf k}]}\otimes\tau_j^{\pm[\vec{\bf n},\overrightarrow{\bf k'}]},\\
&S^l(\tau_j^{\pm[\vec{\bf n},\overrightarrow{\bf k'}]}\otimes\tau_{1-j}^{\pm[\overrightarrow{\bf n'},\vec{\bf k}]})
=l!\sum_{i=0}^{1-j}\sum_{a=0}^{l-i}\sum_{\vec{\gamma}\in k+H_{k'-k,l-i-a}}
\sum_{\vec{\delta}\in n+H_{n'-n,a}}(-1)^{a(n'-k)+(l-i)(k'-1)-|\vec{\gamma}|-|\vec{\delta}|}\times\nonumber\\
&(-1)^{\frac{1\mp1}{2}\Big((k-k')a+(n'-n)(l-i-a)\Big)}
\Big(\tau_{i+j}^{\pm[\vec{\delta}\cup\vec{\bf n},\overrightarrow{\bf k'}\setminus\vec{\gamma}]}\otimes\tau_{1-i-j}^{\pm[\overrightarrow{\bf n'}\setminus\vec{\delta},\vec{\gamma}\cup\vec{\bf k}]}+\sum_{\mu=\pm}\mathcal{M}^{\pm(i+j)}_{a+n,k'-l+i+a,n'-a,l-i-a+k;\mu}\Big),
\end{align*}
In particular,
\begin{align*}
S^{k'-k+1-j}&\Big(\tau_j^{+[\overrightarrow{\bf n'},\overrightarrow{\bf k'}]}\otimes\tau_{1-j}^{+[\vec{\bf n},\vec{\bf k}]}\Big)=(-1)^{\frac{(k'-k)(k'-k-3)}{2}}(k'-k+1-j)!\tau_1^{+[\overrightarrow{\bf n'},\vec{\bf k}]}\otimes\tau_0^{+[\vec{\bf n},\overrightarrow{\bf k'}]},\\
S^{n'-n+1-j}&\Big(\tau_j^{-[\vec{\bf n},\vec{\bf k}]}\otimes\tau_{1-j}^{-[\overrightarrow{\bf n'},\overrightarrow{\bf k'}]}\Big)=(-1)^{\frac{(n'-n)(n'-n-3)}{2}}(n'-n+1-j)!\tau_1^{-[\overrightarrow{\bf n'},\vec{\bf k}]}\otimes\tau_0^{-[\vec{\bf n},\overrightarrow{\bf k'}]},\\
S^{k'-k+n'-n+1-j}&(\tau_j^{\pm[\vec{\bf n},\overrightarrow{\bf k'}]}\otimes\tau_{1-j}^{\pm[\overrightarrow{\bf n'},\vec{\bf k}]})
=(-1)^{(n'-n)(k'-k)+\frac{(n'-n)(n'-n-3)}{2}
+\frac{(k'-k)(k'-k-3)}{2}}\nonumber\\
&\times(n'-n+k'-k+1-j)!\Big(\tau_1^{\pm[\overrightarrow{\bf n'},\vec{\bf k}]}\otimes\tau_0^{\pm[\vec{\bf n},\overrightarrow{\bf k'}]}+\mathcal{M}_{n'knk';\pm}^{\pm}
\Big),\\
S^{k'-k+2-j}&\Big(\tau_j^{+[\overrightarrow{\bf n'},\overrightarrow{\bf k'}]}\otimes\tau_{1-j}^{+[\vec{\bf n},\vec{\bf k}]}\Big)=0,\quad S^{n'-n+2-j}\Big(\tau_j^{-[\overrightarrow{\bf n'},\overrightarrow{\bf k'}]}\otimes\tau_{1-j}^{-[\vec{\bf n},\vec{\bf k}]}\Big)=0,\\
&S^{k'-k+n'-n+2-j}(\tau_j^{\pm[\vec{\bf n},\overrightarrow{\bf k'}]}\otimes\tau_{1-j}^{\pm[\overrightarrow{\bf n'},\vec{\bf k}]})
=0,
\end{align*}
where $\mathcal{M}_{n'knk';+}^{\pm(i)}
=\mathcal{O}^{\pm(i)}_{n'-1,k-1}\otimes\mathcal{O}^{\pm(1-i)}_{n,k'}$, $\mathcal{M}_{n'knk';-}^{\pm(i)}=\mathcal{O}^{\pm (i)}_{n',k}\otimes
\mathcal{O}^{\pm(1-i)}_{n-1,k'-1}$ and $\mathcal{O}^{\pm(i)}_{n,k}$ means the linear combinations of the transformed tau functions in the form of $\tau_i^{\pm[\overrightarrow{\bf n_1},\overrightarrow{\bf k_1}]}$ for $n_1\leq n$ , $k_1\leq k$ and $i=0,1$, under no more than $n$-step $T_D$ and no more than $k$-step $T_I$ with the the generating eigenfunctions and adjoint eigenfunctions in $\hat T^{\pm[\vec{\bf n},\vec{\bf k}]}$.
\end{theorem}
We end this subsection with the bilinear equations in the Darboux chain generated by the wave functions.
Denote $\hat T_{\lambda\mu}^{[\vec{\bf n},\vec{\bf k}]}$ as the successive applications of $n$-step $T_D$ and $k$-step $T_I$ in the following way,
\begin{align*}
&\hat T_{\lambda\mu}^{[\vec{\bf n},\vec{\bf k}]}=T_I(w^*(t,\mu_k)^{[n+k-1]})\cdots T_I(w^*(t,\mu_1)^{[n]})T_D(w(t,\lambda_n)^{[n-1]})\cdots T_D(w(t,\lambda_1)^{[0]}).
\end{align*}
And let $\Big(\tau_{0,\lambda\mu}^{\pm[\vec{\bf n},\vec{\bf k}]},\tau_{1,\lambda\mu}^{\pm[\vec{\bf n},\vec{\bf k}]}\Big)$ be the transformed tau functions under $T_{\lambda\mu}^{\pm[\vec{\bf n},\vec{\bf k}]}$ starting from $\Big(\tau_0,\tau_1\Big)$. By the similar method as before, one can obtain
\begin{align*}
\Big(\tau_{0,\lambda\mu}^{[\vec{\bf n},\vec{\bf k}]},\tau_{1,\lambda\mu}^{[\vec{\bf n},\vec{\bf k}]}\Big)=\left(\prod_{j=2}^n\lambda_j^{-j+1}\prod_{l=1}^k\mu_{l}^{n-l}\right)
\psi_{\lambda_{\vec{\bf n}};\mu_{\vec{\bf k}}}\Big((-1)^k\tau_0^{[0]},(-1)^n\tau_1^{[0]}\Big),
\end{align*}
where $\psi_{\lambda_{\vec{\bf n}};\mu_{\vec{\bf k}}}$ is the same as the KP case. Note that $$\Big(\tau_{0,\lambda\mu}^{[\vec{\bf n},\vec{\bf k}]},\tau_{1,\lambda\mu}^{[\vec{\bf n},\vec{\bf k}]}\Big)\approx \psi_{\lambda_{\vec{\bf n}};\mu_{\vec{\bf k}}}\Big((-1)^k\tau_0^{[0]},(-1)^n\tau_1^{[0]}\Big),$$
Therefore they determine the same (adjoint) wave functions. So in what follows, we can always believe $\Big(\tau_{0,\lambda\mu}^{+[\vec{\bf n},\vec{\bf k}]},\tau_{1,\lambda\mu}^{+[\vec{\bf n},\vec{\bf k}]}\Big)= \psi_{\lambda_{\vec{\bf n}};\mu_{\vec{\bf k}}}\Big((-1)^k\tau_0^{[0]},(-1)^n\tau_1^{[0]}\Big)$ instead. Similarly as before, one can obtain the proposition below.
\begin{proposition}
For $n'\geq n$, $k'\geq k$, $j=0,1$,
\begin{align*}
S^l\Big(\tau_{j,\lambda\mu}^{[\overrightarrow{\bf n'},\overrightarrow{\bf k'}]}\otimes\tau_{1-j,\lambda\mu}^{[\vec{\bf n},\vec{\bf k}]}\Big)=&l!\sum_{i=0}^{1-j}\sum_{\vec{\gamma}\in k+H_{k'-k,l-i}}(-1)^{(l-i)(k'-1)-|\vec{\gamma}|}\tau_{i+j,\lambda\mu}^{[\overrightarrow{\bf n'},\overrightarrow{\bf k'}\setminus\vec{\gamma}]}\otimes\tau_{1-i-j,\lambda\mu}^{[\vec{\bf n},\vec{\gamma}\cup\vec{\bf k}]},\\
S^l\Big(\tau_{j,\lambda\mu}^{[\vec{\bf n},\vec{\bf k}]}\otimes\tau_{1-j,\lambda\mu}^{[\overrightarrow{\bf n'},\overrightarrow{\bf k'}]}\Big)=&l!\sum_{i=0}^{1-j}\sum_{\vec{\gamma}\in n+H_{n'-n,l-i}}(-1)^{(l-i)(k'-k+n'-1)-|\vec{\gamma}|}\tau_{i+j,\lambda\mu}^{[\vec{\bf n}\cup\vec{\gamma},\vec{\bf k}]}\otimes\tau_{1-i-j,\lambda\mu}^{[\overrightarrow{\bf n'}\setminus\vec{\gamma},\overrightarrow{\bf k'}]},\\
S\Big(\tau_{j,\lambda\mu}^{[\overrightarrow{\bf n'},\vec{\bf k}]}\otimes\tau_{1-j,\lambda\mu}^{[\vec{\bf n},\overrightarrow{\bf k'}]}\Big)=&(1-j)\tau_{1-j,\lambda\mu}^{[\overrightarrow{\bf n'},\vec{\bf k}]}\otimes\tau_{j,\lambda\mu}^{[\vec{\bf n},\overrightarrow{\bf k'}]},\\
S^l\Big(\tau_{j,\lambda\mu}^{[\vec{\bf n},\overrightarrow{\bf k'}]}\otimes\tau_{1-j,\lambda\mu}^{[\overrightarrow{\bf n'},\vec{\bf k}]}\Big)
=&l!\sum_{i=0}^{1-j}\sum_{a=0}^l\sum_{\vec{\gamma}\in k+H_{k'-k,l-i-a}}
\sum_{\vec{\delta}\in n+H_{n'-n,a}}(-1)^{a(n'-k)+(l-i)(k'-1)}\\
&\times(-1)^{-|\vec{\gamma}|-|\vec{\delta}|}
\Big(\tau_{i+j,\lambda\mu}^{[\vec{\delta}\cup\vec{\bf n},\overrightarrow{\bf k'}\setminus\vec{\gamma}]}\otimes\tau_{1-i-j,\lambda\mu}^{[\overrightarrow{\bf n'}\setminus\vec{\delta},\vec{\gamma}\cup\vec{\bf k}]}\Big).
\end{align*}
In particular,
\begin{align*}
S^{k'-k+1-j}\Big(\tau_{j,\lambda\mu}^{[\overrightarrow{\bf n'},\overrightarrow{\bf k'}]}\otimes\tau_{1-j,\lambda\mu}^{[\vec{\bf n},\vec{\bf k}]}\Big)=&(-1)^{\frac{(k'-k)(k'-k-3)}{2}}(k'-k+1-j)!
\tau_{1,\lambda\mu}^{[\overrightarrow{\bf n'},\vec{\bf k}]}\otimes\tau_{0,\lambda\mu}^{[\vec{\bf n},\overrightarrow{\bf k'}]},\\
S^{n'-n+1-j}\Big(\tau_{j,\lambda\mu}^{[\vec{\bf n},\vec{\bf k}]}\otimes\tau_{1-j,\lambda\mu}^{[\overrightarrow{\bf n'},\overrightarrow{\bf k'}]}\Big)=&(-1)^{(n'-n)(k'-k)+\frac{(n'-n)(n'-n-3)}{2}}(n'-n+1-j)!
\tau_{1,\lambda\mu}^{[\overrightarrow{\bf n'},\vec{\bf k}]}\otimes\tau_{0,\lambda\mu}^{[\vec{\bf n},\overrightarrow{\bf k'}]},\\
S^{n'-n+k'-k+1-j}\Big(\tau_{j,\lambda\mu}^{[\vec{\bf n},\overrightarrow{\bf k'}]}\otimes\tau_{1-j,\lambda\mu}^{[\overrightarrow{\bf n'},\vec{\bf k}]}\Big)
=&(-1)^{(n'-n)(k'-k)+\frac{(n'-n)(n'-n-3)}{2}+\frac{(k'-k)(k'-k-3)}{2}}\\
&\times(n'-n+k'-k+1-j)!\tau_{1,\lambda\mu}^{[\overrightarrow{\bf n'},\vec{\bf k}]}\otimes\tau_{0,\lambda\mu}^{[\vec{\bf n},\overrightarrow{\bf k'}]}.
\end{align*}
\end{proposition}

\subsection{Examples of the bilinear equations in the modified KP Darboux transformation}
\label{subsectmkpdtexp}
Here we will present some examples of the bilinear equations in the subsection above. For $(n',k',n,k)=(1,0,0,0)$, $(n',k',n,k)=(0,1,0,0)$ and $(n',k',n,k)=(1,1,0,0)$,
\begin{align*}
&S(\tau_0^{[1,0]}\otimes\tau_1^{[0,0]})=\tau_1^{[1,0]}\otimes \tau_0^{[0,0]},\quad S(\tau_0^{[0,0]}\otimes\tau_1^{[1,0]})=\tau_1^{[0,0]}\otimes \tau_0^{[1,0]}-\tau_0^{[1,0]}\otimes \tau_1^{[0,0]},\\
&S(\tau_0^{[0,0]}\otimes\tau_1^{[0,1]})=\tau_1^{[0,0]}\otimes \tau_0^{[0,1]},\quad S(\tau_0^{[0,1]}\otimes\tau_1^{[0,0]})=\tau_1^{[0,1]}\otimes \tau_0^{[0,0]}-\tau_0^{[0,0]}\otimes \tau_1^{[0,1]},\\
&S(\tau_0^{[1,1]}\otimes\tau_1^{[0,0]})=\tau_1^{[1,1]}\otimes \tau_0^{[0,0]}-\tau_0^{[1,0]}\otimes \tau_1^{[0,1]},\\ &S(\tau_0^{[0,0]}\otimes\tau_1^{[1,1]})=\tau_1^{-[0,0]}\otimes \tau_0^{[1,1]}+\tau_0^{[1,0]}\otimes \tau_1^{[0,1]},\\
&S(\tau_0^{[1,0]}\otimes\tau_1^{[0,1]})=\tau_1^{[1,0]}\otimes\tau_0^{[0,1]},\quad S(\tau_0^{[0,1]}\otimes\tau_1^{[1,0]})=\tau_0^{[1,1]}\otimes\tau_1^{[0,0]}
-\tau_0^{[0,0]}\otimes\tau_1^{[1,1]}+\tau_1^{[0,1]}\otimes\tau_0^{[1,0]},
\end{align*}
and
\begin{align*}
&S(\tau_1^{[1,0]}\otimes\tau_0^{[0,0]})=0,\quad S(\tau_1^{[0,0]}\otimes\tau_0^{[1,0]})=-\tau_1^{[1,0]}\otimes \tau_0^{[0,0]},\\
&S(\tau_1^{[0,0]}\otimes\tau_0^{[0,1]})=0,\quad S(\tau_1^{[0,1]}\otimes\tau_0^{[0,0]})=-\tau_1^{[0,0]}\otimes \tau_0^{[0,1]},\\
&S(\tau_1^{[1,1]}\otimes\tau_0^{[0,0]})=-\tau_1^{[1,0]}\otimes \tau_0^{[0,1]}=-S(\tau_1^{[0,0]}\otimes\tau_0^{[1,1]}),\\
&S(\tau_1^{[1,0]}\otimes\tau_0^{[0,1]})=0,\quad S(\tau_1^{[0,1]}\otimes\tau_0^{[1,0]})=-\tau_1^{[0,0]}\otimes\tau_0^{[1,1]}
+\tau_1^{[1,1]}\otimes\tau_0^{[0,0]}.
\end{align*}
Here we have set $\tau_j^{[0,1]}=\tau_j^{\pm[0,1]}$, $\tau_j^{[1,0]}=\tau_j^{\pm[1,0]}$ and $\tau_j^{[1,1]}=\tau_j^{+[1,1]}$.
These bilinear relations can be written into
\begin{align*}
&S(\tau_j^{[l,0]}\otimes\tau_{1-j}^{[1-l,0]})=
-\delta_{l0}\tau_j^{[1,0]}\otimes\tau_{1-j}^{[0,0]}
+(1-j)\tau_{1-j}^{[l,0]}\otimes\tau_{j}^{[1-l,0]},\\
&S(\tau_j^{[0,m]}\otimes\tau_{1-j}^{[0,1-m]})=
-\delta_{m1}\tau_j^{[0,0]}\otimes\tau_{1-j}^{[0,1]}
+(1-j)\tau_{1-j}^{[0,m]}\otimes\tau_{j}^{[0,1-m]},\\
&S(\tau_j^{[l,m]}\otimes\tau_{1-j}^{[1-l,1-m]})=
-\delta_{m1}\tau_j^{[l,0]}\otimes\tau_{1-j}^{[1-l,1]}
+\delta_{l0}\tau_j^{[1,m]}\otimes\tau_{1-j}^{[0,1-m]}
+(1-j)\tau_{1-j}^{[l,m]}\otimes\tau_{j}^{[1-l,1-m]}.
\end{align*}
The corresponding Bosonic forms are
\begin{align*}
\bullet\quad{\rm Res}_\lambda &\lambda^{2(j+l-1)}\tau_j^{[l,0]}(t-\varepsilon(\lambda^{-1}))
\tau_{1-j}^{[1-l,0]}(t'+\varepsilon(\lambda^{-1}))e^{\xi(t,\lambda)-\xi(t',\lambda)}\\
=&-\delta_{l0}\tau_j^{[1,0]}(t)\tau_{1-j}^{[0,0]}(t')
+(1-j)\tau_{1-j}^{[l,0]}(t)\tau_{j}^{[1-l,0]}(t'),\\
\bullet\quad{\rm Res}_\lambda &\lambda^{2(j-m)}\tau_j^{[0,m]}(t-\varepsilon(\lambda^{-1}))
\tau_{1-j}^{[0,1-m]}(t'+\varepsilon(\lambda^{-1}))e^{\xi(t,\lambda)-\xi(t',\lambda)}\\
=&-\delta_{m1}\tau_j^{[0,0]}(t)\tau_{1-j}^{[0,1]}(t')
+(1-j)\tau_{1-j}^{[0,m]}(t)\tau_{j}^{[0,1-m]}(t'),\\
\bullet\quad{\rm Res}_\lambda &\lambda^{2(j+l-m)-1}\tau_j^{[l,m]}(t-\varepsilon(\lambda^{-1}))
\tau_{1-j}^{[1-l,1-m]}(t'+\varepsilon(\lambda^{-1}))e^{\xi(t,\lambda)-\xi(t',\lambda)}\\
=&-\delta_{m1}\tau_j^{[l,0]}(t)\tau_{1-j}^{[1-l,1]}(t')
+\delta_{l0}\tau_j^{[1,m]}(t)\tau_{1-j}^{[0,1-m]}(t')
+(1-j)\tau_{1-j}^{[l,m]}(t)\tau_{j}^{[1-l,1-m]}(t').
\end{align*}
If we insert the relations below derived by Table II into the relation above,
\begin{align*}
\tau_{0}^{[1,0]}(t)=\Phi(t)\tau_1(t),&\quad \tau_{1}^{[1,0]}(t)=-\frac{\Phi_x(t)\tau_1^2(t)}{\tau_0(t)},\\
\tau_{0}^{[0,1]}(t)=\frac{\Psi_x(t)\tau_0^2(t)}{\tau_1(t)},&\quad \tau_{1}^{[0,1]}(t)=\Psi(t)\tau_0(t),\\
\tau_{0}^{[1,1]}(t)=-\Omega(\Phi(t),\Psi_x(t))\tau_0(t),&\quad \tau_{1}^{[1,1]}(t)=\hat\Omega(\Phi_x(t),\Psi(t))\tau_1(t),
\end{align*}
then one can obtain many important relations about the modified KP hierarchy. Here we only take $S(\tau_0^{[1,0]}\otimes\tau_1^{[0,0]})$, $S(\tau_0^{[1,1]}\otimes\tau_1^{[0,0]})$ and $S(\tau_1^{[0,1]}\otimes\tau_0^{[1,0]})$ as examples, which are
\begin{align}
{\rm Res}_\lambda &\Phi(t-\varepsilon(\lambda^{-1}))\tau_1(t-\varepsilon(\lambda^{-1}))
\tau_{1}(t'+\varepsilon(\lambda^{-1}))e^{\xi(t,\lambda)-\xi(t',\lambda)}=
-\frac{\Phi_x(t)\tau_1^2(t)}{\tau_0(t)}\tau_{0}(t'),\\
{\rm Res}_\lambda&\Omega(\Phi(t-\varepsilon(\lambda^{-1})),\Psi_x(t-\varepsilon(\lambda^{-1})))
w(t,\lambda)w^*(t',\lambda)=\Phi(t)\Psi(t')
-\hat\Omega(\Phi_x(t),\Psi(t)),\label{resomegawwstar}\\
{\rm Res}_\lambda&
\Psi(t-\varepsilon(\lambda^{-1}))\Phi(t'+\varepsilon(\lambda^{-1}))w(t,\lambda)w^*(t',\lambda)
=\Omega(\Phi(t'),\Psi_x(t'))
+\hat\Omega(\Phi_x(t),\Psi(t)).\label{respsiphiwwstar}
\end{align}
Here we would like to point out (\ref{resomegawwstar}) or \eqref{respsiphiwwstar} can lead to
\begin{align}
&\Omega(\Phi(t-\varepsilon(\lambda^{-1})),\Psi_x(t-\varepsilon(\lambda^{-1})))
+\hat\Omega(\Phi_x(t),\Psi(t))=\Phi(t)\Psi(t-\varepsilon(\lambda^{-1})),\label{omegplus}
\end{align}
which can be viewed as the
analogue of (\ref{stminusadd}) in the modified KP hierarchy.
(\ref{omegplus}) is also very important in the proof of the ASvM formula in the mKP case and derivation of the bilinear equations of the $l$-constrained modified KP hierarchy: $L^l=(L^l)_{\geq1}+\Phi\partial^{-1}\Psi\partial$ (one can refer to \cite{Chengjgp2018,Chenjnmp2019} for these results).
\section{The Darboux transformations of the BKP hierarchy}
\subsection{Reviews on the Darboux transformations of the BKP hierarchy}
The Darboux transformation of the BKP hierarchy \cite{He2007,He2006,Yang2020} can be constructed by the union of $T_d$ and $T_i$ (or $T_D$ and $T_I$), that is
\begin{align}
T(\Phi_B)=T_i(\Phi_B)T_d(\Phi_B)=T_I(\Phi_B)T_D(\Phi_B)
=1-2\Phi_B\partial^{-1}\Phi_{B,x},
\end{align}
where $\Phi_B$ is the eigenfunction of the BKP hierarchy, satisfying $\Phi_{B,t_{2k+1}}=(L_B^{2k+1})_{\geq 1}(\Phi_B)$. Under $T(\Phi_B)$, the eigenfunction $\Phi_{B1}(t)$ and the tau function $\tau_B(t)$ will become into
\begin{align}
\Phi_{B1}(t)\xrightarrow{T(\Phi_B(t))}\Phi_{B1}(t)^{[1]}=
\frac{\Omega_B(\Phi_{B1}(t),\Phi_B(t))}{\Phi_B(t)},\quad \tau_{B}(t)\xrightarrow{T(\Phi_B(t))}\tau_{B}(t)^{[1]}=\Phi_B(t)\tau_{B}(t),\label{trbkpobj}
\end{align}
where $A^{[\vec{\bf k}]}$ or $A^{[k]}$ means the transformed objects under $k$-step $T(\Phi_B)$, i.e., $T_{B}^{[\vec{\bf k}]}=T(\Phi_{B,n}^{[n-1]})\cdots T(\Phi_{B1}^{[0]})$, while $A^{[\{i\}]}$ denotes the transformed one under $T(\Phi_{Bi})$. In particular by using (\ref{omegab}),
\begin{align*}
\psi_B(t,\lambda)^{[1]}=T(\Phi_B(t))(\psi_B(t,\lambda))
=\frac{\tau_B^{[1]}(t-2\widetilde{\varepsilon}(\lambda^{-1}))}
{\tau_B^{[1]}(t)}e^{\tilde{\xi}(t,\lambda)}.
\end{align*}

Next we try to express the Darboux transformation of the BKP hierarchy in the Fermionic pictures. Firstly, assume $\tau_B(t)=\langle 0|e^{H_B(t)}g|0\rangle$ for $g\in O_\infty$ (we also use $\tau_B=g|0\rangle$). Then if denote $$\alpha_B=\sqrt{2}{\rm Res}_{\lambda}\lambda^{-1}\rho_B(\lambda)\phi(\lambda)\in V_B=\oplus_{l\in\mathbb{Z}}\mathbb{C}\phi_l$$
with $\rho_B(\lambda)$ given in Proposition \ref{propspecbkp}, then according to (\ref{bkpwavefunction}) and (\ref{bkpspphi}) one can find
\begin{align*}
\Phi_B(t)=\frac{\langle 0|e^{H_B(t)}\alpha_B g|0\rangle}{\langle 0|e^{H_B(t)}g|0\rangle}.
\end{align*}
Therefore in the Fermionic picture,
\begin{align}
\tau_{B}\xrightarrow{T(\Phi_B(t))}\tau_{B}^{[1]}=\alpha_B\tau_{B}.\label{fbkpdt}
\end{align}
Note that according to Lemma \ref{sbbeta},
\begin{align*}
S_B(\alpha_B\tau_{B}\otimes \alpha_B\tau_{B})=&
\tau_B\otimes\alpha_B^2\tau_B-\alpha_B^2\tau_B\otimes\tau_B
+\frac{1}{2}\alpha_B\tau_{B}\otimes \alpha_B\tau_{B}=\frac{1}{2}\alpha_B\tau_{B}\otimes\alpha_B\tau_{B},
\end{align*}
where we have used $\alpha_B^2$ is a constant. The corresponding Bosonic form is as follows,
\begin{align*}
{\rm Res}_\lambda \lambda^{-1}\Phi_B(t-2\widetilde{\varepsilon}(\lambda^{-1}))\tau_B(t-2\widetilde{\varepsilon}(\lambda^{-1}))
&\Phi_B(t'+2\widetilde{\varepsilon}(\lambda^{-1}))\tau_B(t'-2\widetilde{\varepsilon}(\lambda^{-1}))
e^{\xi(t,\lambda)-\xi(t',\lambda)}\\
&=\Phi_B(t)\tau_B(t)
\Phi_B(t')\tau_B(t').
\end{align*}
This relation is also obtained in \cite{Loris1999}.
\subsection{Bilinear equations in the Darboux transformation of the BKP hierarchy}\label{subsectbilinbkp}
First of all, let's discuss the transformed BSEP under the BKP Darboux transformation, in order to see the changes of the BKP tau functions under the successive applications of the Darboux transformation. By using (\ref{trbkpobj}),
\begin{lemma}\label{sepbdt}
Under the BKP Darboux transformation $T(\Phi_{B})$,
\begin{eqnarray*}
\Omega(\Phi^{[1]}_{B1}(t),\psi^{[1]}_B(t,\lambda)_x)
=\Omega(\Phi_{B1}(t),\psi_B(t,\lambda)_x)
-\frac{2\Omega(\Phi_{B1}(t),\Phi_B(t)_x)}{\Phi_B(t)^2}\Omega(\Phi_B(t),\psi_B(t,\lambda)_x).
\end{eqnarray*}
\end{lemma}
\begin{proof}
By using $T(\Phi_B)^*\cdot\partial \cdot T(\Phi_B)=\partial$, one can find $\psi^{[1]}_B(t,\lambda)_x=(T(\Phi_B)^{-1})^*(\psi_B(t,\lambda)_x)$. Then
\begin{align*}
\Omega(\Phi^{[1]}_{B1}(t),\psi^{[1]}_B(t,\lambda)_x)
=&-\Phi_B^{-1}\cdot(T_d(\Phi_B)^{-1})^*(\psi_B(t,\lambda)_x)\cdot
\Omega\Big(T_d(\Phi_B)\big(\Phi_{B1}\big),\Phi_B\Big)\\
&+\Omega\Big(T_d(\Phi_B)\big(\Phi_{B1}\big),
(T_d(\Phi_B)^{-1})^*(\psi_B(t,\lambda)_x)\Big),
\end{align*}
leads to the corresponding result.
\end{proof}
\begin{proposition}\label{propbkpdt}
Under $n$-step BKP Darboux transformations,
\begin{eqnarray*}
\tau_{B}^{[\vec{\bf n}]}=\alpha_{B,n}^{[n-1]}\cdots\alpha_{B,2}^{[1]}\alpha_{B,1}^{[0]} \tau_{B}^{[0]}=\sum_{j=0}^{[n/2]}\sum_{\vec{\gamma}\in H_{n,n-2j}}C^{[\vec{\bf n}]}_{j,\vec{\gamma}}\alpha_{B,\vec{\gamma}}\tau_{B}^{[0]},
\end{eqnarray*}
where $\alpha_{B,\vec{\gamma}}=\alpha_{B,\gamma_{n-2j}}\cdots\alpha_{B,\gamma_{1}}$ and $\alpha_{B,i}$ is corresponding to the BKP eigenfunction $\Phi_{Bi}(t)$ and $C^{[\vec{\bf n}]}_{j,\vec{\gamma}}$ is the constant independent of $t$ with $C^{[\vec{\bf n}]}_{0,\vec{\bf n}}=1$.
\end{proposition}
\begin{proof}
According to Lemma \ref{sepbdt} and $\alpha^{[j]}_{B,i}=\sqrt{2}{\rm Res}_\lambda\lambda^{-1}\rho^{[j]}_{B,i}(\lambda)\phi(\lambda)$, one can find for $i>j$,
\begin{eqnarray*}
\rho^{[j]}_{B,i}(\lambda)=\rho^{[j-1]}_{B,i}(\lambda)+c_{B,i}^{[j-1]}\rho^{[j-1]}_{B,j}(\lambda),\quad
\alpha^{[j]}_{B,i}=\alpha^{[j-1]}_{B,i}+c_{B,i}^{[j-1]}
\alpha^{[j-1]}_{B,j},
\end{eqnarray*}
where $c_{B,i}^{[j-1]}=-2\Omega(\Phi_{Bi}^{[j-1]}(t'),\Phi_{Bj}^{[j-1]}(t')_x)/\Phi_{Bj}^{[j-1]}(t')^2$ is some constant independent of $t$. Therefore
\begin{align*}
\alpha^{[j]}_{B,i}=\alpha^{[0]}_{B,i}+\sum_{l=0}^ja^{[j]}_{i,l}\alpha^{[0]}_{B,l},\quad i>j.
\end{align*}
Here $a^{[j]}_{i,l}$ is the constant satisfying $a^{[j]}_{i,j}=c_{B,i}^{[j-1]}$ and $a^{[j+1]}_{i,l}=a^{[j]}_{i,l}+C_{B,i}^{[j]}a_{j+1,l}^{[j]}$. Then this proposition can be proved by induction on $n$. The recursion relations for constants $C^{[\vec{\bf n}]}_{0,\vec{\bf n}}=1$ are very complicated. Since in what follows we only need to know $C^{[\vec{\bf n}]}_{j,\vec{\gamma}}$ is constant, so we omit its recursion relation here.
\end{proof}
{\noindent \bf Remark:} Different from the KP and mKP case, $\alpha_{Bi}$ can not commute or anticommute with each other, and $\alpha_{Bi}^2\neq 0$. So the expression of $\tau_B^{[\vec{\bf n}]}$ is more complicated, compared with $\tau^{\pm[\vec{\bf n}]}$ in KP case and $\tau_i^{\pm[\vec{\bf n}]}$ ($i=0,1$) in mKP case.\\
{\noindent \bf Remark:} Since $\tau_{B}^{[\vec{\bf n}]}$ is the tau function of the BKP hierarchy, it satisfies $S_B\Big(\tau_{B}^{[n]}\otimes\tau_{B}^{[n]}\Big)=\frac{1}{2}\tau_{B}^{[n]}
\otimes\tau_{B}^{[n]}$. Particularly for $n=2$, the corresponding Bosonic form is \begin{align*}
&{\rm Res}_\lambda \lambda^{-1} \Omega_B(\phi_{B2}(t-2\widetilde{\varepsilon}(\lambda^{-1})),\phi_{B1}(t-2\widetilde{\varepsilon}(\lambda^{-1})))
\tau_B(t-2\widetilde{\varepsilon}(\lambda^{-1}))\\
&\times\Omega_B(\phi_{B2}(t'-2\widetilde{\varepsilon}(\lambda^{-1})),\phi_{B1}(t'-2\widetilde{\varepsilon}(\lambda^{-1})))
\tau_B(t'+\varepsilon(\lambda^{-1}))e^{\xi(t,\lambda)-\xi(t',\lambda)}\\
=&\Omega_B(\phi_{B2}(t),\phi_{B1}(t))
\Omega_B(\phi_{B2}(t'),\phi_{B1}(t'))
\tau_B(t)\tau_B(t').
\end{align*}
This relation is also obtained in \cite{Loris1999}, where the corresponding proof is very complicated. But here in the Fermionic approach, we can easily get it, which tells us that the Fermionic method is more efficient.

\begin{corollary}\label{coralpbtau}
$\alpha_{B,\vec{\bf n}}\tau_B^{[0]}$ can be written into
\begin{align*}
\alpha_{B,\vec{\bf n}}\tau_B^{[0]}=\sum_{l=0}^{[n/2]}\sum_{\vec{\delta}\in H_{n,n-2l}}\tilde{C}^{[\vec{\bf n}]}_{l,\vec{\delta}}\tau_{B}^{[\vec{\delta}]},
\end{align*}
where $\tilde{C}^{[\vec{\bf n}]}_{l,\vec{\delta}}$ satisfies $\tilde{C}^{[\vec{\bf n}]}_{0,\vec{\bf n}}=1$ and $\tilde{C}^{[\vec{\bf n}]}_{l,\vec{\delta}}=-\sum_{i=0}^{j}\sum_{\vec{\gamma}\in H_{n,n-2i}}C^{[\vec{\bf n}]}_{i,\vec{\gamma}}\tilde{C}
^{[\vec{\gamma}]}_{l-i,\vec{\delta}}$ for $l\geq 1$, $\vec{\delta}\in H_j=\{\vec{\gamma}=(\gamma_{n-2l},\cdots,\gamma_{1})|\vec{\gamma}\in H_{n-2j,n-2l}, \vec{\gamma}\notin H_{n-2j-2,n-2l}\}$ with $H_{n,n-2l}=\cup_{j=0}^lH_j$.
\end{corollary}
\begin{proof}
It is obviously correct for $n=0$ and $n=1$. If we assume it is corrected for $<n$, then by Proposition \ref{propbkpdt},
\begin{align*}
&\alpha_{B,\vec{\bf n}}\tau_B^{[0]}=\tau_B^{[\vec{\bf n}]}-\sum_{i=1}^{[n/2]}\sum_{\vec{\gamma}\in H_{n,n-2i}}C^{[\vec{\bf n}]}_{i,\vec{\gamma}}\alpha_{B,\vec{\gamma}}\tau_{B}^{[0]}\\
=&\tau_B^{[\vec{\bf n}]}-\sum_{i=1}^{[n/2]}\sum_{l=0}^{[n/2]-i}\sum_{\vec{\gamma}\in H_{n,n-2i}}\sum_{\vec{\delta}\in H_{n-2i,n-2i-2l}}C^{[\vec{\bf n}]}_{i,\vec{\gamma}}\tilde{C}
^{[\vec{\gamma}]}_{l,\vec{\delta}}\tau_{B}^{[\gamma_{\vec{\delta}}]}\\
=&\tau_B^{[\vec{\bf n}]}-\sum_{l=1}^{[n/2]}\sum_{j=0}^{l}\sum_{\vec{\delta}\in H_j}\sum_{i=0}^j\sum_{\vec{\gamma}\in H_{n,n-2i}}C^{[\vec{\bf n}]}_{i,\vec{\gamma}}\tilde{C}
^{[\vec{\gamma}]}_{l-i,\vec{\delta}}\tau_{B}^{[\gamma_{\vec{\delta}}]}
\end{align*}
Then by noting that $H_{n,n-2l}=\cup_{j=0}^lH_j$, this corollary can be proved.
\end{proof}

Denote $\mathcal{O}_n$ to be the linear combination of the transformed tau functions, under no more than $n$-step $T_B$ with the corresponding generating eigenfunctions in $T_B^{[\vec{\bf n}]}$. Then one can obtain the following theorem on the bilinear equations of $\tau_B^{[n]}$.
\begin{theorem}
Given $n'\geq n$, one has the following relations
\begin{align*}
S_B^l\Big(\tau_{B}^{[n']}\otimes\tau_{B}^{[n]}\Big)=&
(-1)^{(n'-n)l}\sum_{j=0}^l2^{-(l-j)}C_l^j\sum_{k=0}^{[j/2]}a_{n'-n,j,k}\sum_{\vec{\gamma}\in n+H_{n'-n,j-2k}}(-1)^{-|\vec{\gamma}|+n(j-2k)}\\
&\times\Big(\tau_B^{[\overrightarrow{\bf n'}\setminus\vec{\gamma}]}\otimes\tau_B^{[\vec{\gamma}\cup\vec{\bf n}]}+\mathcal{O}_{n'-j+2k-2}\otimes\mathcal{O}_{n+j-2k}
+\mathcal{O}_{n'-j+2k}\otimes\mathcal{O}_{n+j-2k-2}\Big),\\
S_B^l\Big(\tau_{B}^{[n]}\otimes\tau_{B}^{[n']}\Big)=&
(-1)^{(n'-n)l}\sum_{j=0}^l2^{-(l-j)}C_l^j\sum_{k=0}^{[j/2]}a_{n'-n,j,k}\sum_{\vec{\gamma}\in n+H_{n'-n,j-2k}}(-1)^{-|\vec{\gamma}|+n(j-2k)}\\
&\times\Big(\tau_B^{[\vec{\gamma}\cup\vec{\bf n}]}\otimes\tau_B^{[\overrightarrow{\bf n'}\setminus\vec{\gamma}]}+\mathcal{O}_{n+j-2k-2}\otimes\mathcal{O}_{n'-j+2k}
+\mathcal{O}_{n+j-2k}\otimes\mathcal{O}_{n'-j+2k-2}\Big).
\end{align*}
\end{theorem}
\begin{proof}
Since $S_B^l\Big(\tau_{B}^{[n]}\otimes\tau_{B}^{[n]}\Big)=2^{-l}
\Big(\tau_{B}^{[n]}\otimes\tau_{B}^{[n]}\Big)$, so we can only consider $S_B^l\Big(\tau_{B}^{[n]}\otimes\tau_{B}^{[0]}\Big)$. According to Lemma \ref{sblbeta}, Proposition \ref{propbkpdt} and Corollary \ref{coralpbtau},
\begin{align*}
&S_B^l\Big(\tau_{B}^{[n]}\otimes\tau_{B}^{[0]}\Big)=\sum_{j=0}^{[n/2]}\sum_{\vec{\gamma}\in H_{n,n-2j}}S_B^l(\alpha_{B,\vec{\gamma}}\otimes 1) (\tau_{B}^{[0]}\otimes \tau_{B}^{[0]})\\
=&(-1)^{nl}\sum_{j=0}^{[n/2]}\sum_{\vec{\gamma}\in H_{n,n-2j}}C^{[\vec{\bf n}]}_{j,\vec{\gamma}}\sum_{i=0}^l2^{-(l-j)}C_l^j\sum_{k=0}^{[i/2]}a_{n-2j,i,k}
\sum_{\vec{\delta}\in H_{n-2j,i-2k}}(-1)^{-|\vec{\delta}|}\\
&\times\Big(\tau_B^{[\vec{\gamma}\setminus\gamma_{\vec{\delta}}]}
\otimes\tau_B^{[\vec{\delta}]}+\mathcal{O}_{n-2j-i+2k-2}\otimes\mathcal{O}_{i-2k}
+\mathcal{O}_{n-2j-i+2k}\otimes\mathcal{O}_{i-2k-2}).
\end{align*}
Further by using $C^{[\vec{\bf n}]}_{0,\vec{\bf n}}=1$, one can obtain the corresponding result.
\end{proof}
\noindent{\bf Remark:} It can be proved that the results of $S_B^l\Big(\tau_{B}^{[n']}\otimes\tau_{B}^{[n]}\Big)$ and $S_B^l\Big(\tau_{B}^{[n]}\otimes\tau_{B}^{[n']}\Big)$ are equivalent in the Bosonic form by using the fact ${\rm Res}_\lambda f(\lambda)=-{\rm Res}_\lambda f(-\lambda)$.

\subsection{Examples of the bilinear equations in the BKP Darboux transformation}
\label{subsectbkpdtexp}
In this section, one give some example for the bilinear equations for the BKP hierarchy.
\begin{eqnarray*}
&&S_B(\tau_B^{[0]}\otimes\tau_B^{[1]})=\tau_B^{[0]}\otimes\tau_B^{[1]}-\frac{1}{2}\tau_B^{[1]}\otimes\tau_B^{[0]},\\
&&S_B(\tau_B^{[2]}\otimes\tau_B^{[0]})=\tau_{B}^{[\{1\}]}\otimes\tau_{B}^{[\{2\}]}
-\tau_{B}^{[\{2\}]}\otimes\tau_{B}^{[\{1\}]}+\frac{1}{2}\tau_B^{[2]}\otimes\tau_B^{[0]}\\
&&S_B(\tau_B^{[\{1\}]}\otimes\tau_B^{[\{2\}]})=\tau_{B}^{[2]}\otimes\tau_{B}^{[0]}
-\tau_{B}^{[0]}\otimes\tau_{B}^{[2]}+\frac{1}{2}\tau_B^{[\{1\}]}\otimes\tau_B^{[\{2\}]}.
\end{eqnarray*}
The corresponding Bosonic forms are
\begin{align}
{\rm Res}_\lambda \lambda^{-1} &\tau_B^{[1]}(t-2\widetilde{\varepsilon}(\lambda^{-1}))
\tau_B^{[0]}(t'+\varepsilon(\lambda^{-1}))e^{\widetilde{\xi}(t,\lambda)-\widetilde{\xi}(t',\lambda)}
=2\tau_B^{[0]}(t)\tau_B^{[1]}(t')-\tau_B^{[1]}(t)\tau_B^{[0]}(t'),\label{stau1tau0mbkp}\\
{\rm Res}_\lambda \lambda^{-1} &\Omega_B(\Phi_{B2}(t-2\widetilde{\varepsilon}(\lambda^{-1})),\Phi_{B1}(t-2\widetilde{\varepsilon}(\lambda^{-1})))
w(t,\lambda)w(t',\lambda)\nonumber\\
&=2\Phi_{B1}(t)\Phi_{B2}(t')-2\Phi_{B2}(t)\Phi_{B1}(t')+
\Omega_B(\Phi_{B2}(t),\Phi_{B1}(t)),\label{resomeww}\\
{\rm Res}_\lambda \lambda^{-1} &\Phi_{B1}(t-2\widetilde{\varepsilon}(\lambda^{-1}))
\Phi_{B2}(t'+2\widetilde{\varepsilon}(\lambda^{-1}))w(t,\lambda)w(t',\lambda)\nonumber\\
&=2\Omega_B(\Phi_{B2}(t),\Phi_{B1}(t))-2\Omega_B(\Phi_{B2}(t'),\Phi_{B1}(t'))
+\Phi_{B1}(t)\Phi_{B2}(t'),\label{resphib12ww}
\end{align}
where we have used $\tau_{B}^{[\{1\}]}(t)=\Phi_{B1}(t)\tau_B(t)$, $\tau_{B}^{[\{2\}]}(t)=\Phi_{B2}(t)\tau_B(t)$ and $\tau_B^{[2]}=\Omega_B(\Phi_{B2}(t),\Phi_{B1}(t))\tau_B(t)$.
Here the first relation (\ref{stau1tau0mbkp}) is just the bilinear equation of the modified BKP hierarchy.
By noting that $\Omega_B(\Phi_B,1)=\Phi_B$, so if we set $\Phi_{B1}=1$ (one should recall that $1$ is also the eigenfunction of the BKP hierarchy), then
the relation (\ref{resomeww}) is transformed into the bilinear equations of the modified BKP hierarchy\cite{Jimbo1983,Wang2019}. Further it can be found that (\ref{resomeww}) and \eqref{resphib12ww} can give rise to the following relation \cite{Cheng2010}.
\begin{align}
&\Omega_B(\phi_{B2}(t-2\widetilde{\varepsilon}(\lambda^{-1})),\phi_{B1}(t-2\widetilde{\varepsilon}(\lambda^{-1})))-
\Omega_B(\phi_{B2}(t),\phi_{B1}(t))\nonumber\\
=&\Phi_{B1}(t)\Phi_{B2}(t-2\widetilde{\varepsilon}(\lambda^{-1}))
-\Phi_{B2}(t)\Phi_{B1}(t-2\widetilde{\varepsilon}(\lambda^{-1})),
\end{align}
which is also very important in deriving the ASvM formula in the additional symmetries of the BKP hierarchy \cite{Tu2007,Li2015} and getting the bilinear equations \cite{Shen2011} of the $l$ constrained BKP hierarchy: $L_B^l=(L_B^l)_{\geq 0}+\Phi_{B1}\partial^{-1}\Phi_{B2,x}-\Phi_{B2}\partial^{-1}\Phi_{B1,x}$ by using the method in Subsection \ref{subsectkpdtexp}.
\section{Conclusions and Discussions}
In this paper, we have established various bilinear equations of the transformed tau functions under the Darboux transformations for the KP, modified KP and BKP hierarchies, which are given in Subsection \ref{subsectbilinkp}, Subsection \ref{subsectbilinmkp} and Subsection \ref{subsectbilinbkp} respectively. All these results are based upon two key things. One is the important relations about the free Fermions in Subsection \ref{subsectionrelationfreeferm}, another is the transformed tau functions in Fermionic forms under the successive applications of the Darboux transformations (see Proposition \ref{kptdtitau}, Proposition \ref{mkptdti} and Proposition \ref{propbkpdt}). Here we would like to point out the major difficulty in this paper. Note that free Fermionic field $\alpha\in V$ do not commute or anticommute with $\alpha^*\in V^*$, and similarly the neutral free Fermionic field $\alpha_B\in V_B$ also do not commute or anticommute with each other. Also $\alpha_B^2\neq 0$, different from $\alpha^2=\alpha^{*2}=0$. All these facts bring much difficulty for mixed using $T_d$ and $T_i$ (or $T_D$ and $T_I$), and for the BKP Darboux transformation. The bilinear equations in KP case involving mixed using $T_d$ and $T_i$, and the ones for the modified KP and BKP hierarchy should be new ones. These bilinear equations are usually very hard to prove in the Bosonic forms, while they can be easily obtained by using the Fermionic approach. The corresponding examples are given in Subsection \ref{subsectkpdtexp}, Subsection \ref{subsectmkpdtexp} and Subsection \ref{subsectbkpdtexp} respectively.

There should be some questions needing further discussions. Though we have obtained all the possible bilinear equations in this paper, it will be very interesting to discuss the relations among them and check which bilinear equations are essential in determining the transformed tau functions in the Darboux chain. What's more, since the bilinear equations can be viewed as the Pl\"ucker relations in the infinite dimensional Grassmannian, it will be very meaningful to determine this kind of correspondence. Just as showed in \cite{Helminckcjm2001,Helminckprims2001}, $T_d$ and $T_i$ are interpreted as the transformations of the points in the KP Grassmannian, but the corresponding Pl\"ucker relations are not considered. In the future, we will consider these questions.

\end{document}